\documentclass[pra,twocolumn,superscriptaddress,preprintnumbers,nofootinbib]{revtex4-2}

\usepackage[english]{babel}
\usepackage[utf8]{inputenc}
\usepackage{graphicx}
\usepackage{amsmath}
\usepackage{amsthm}
\usepackage{bm}
\usepackage{layout}
\usepackage{float}
\usepackage{amsfonts}
\usepackage{amssymb}%
\usepackage{array}%
\usepackage[margin=0.75in]{geometry}
\usepackage{color}
\usepackage{soul}
\usepackage{mathtools}
\usepackage[colorlinks=true,citecolor=blue]{hyperref}
\theoremstyle{definition}

\newtheorem{remark}{Remark}

\theoremstyle{theorem}
\newtheorem{theorem}{Theorem}
\newtheorem{lemma}{Lemma}

\newtheorem{corollary}{Corollary}

\usepackage{physics}

\usepackage{ulem}

\renewcommand{\Re}{\mathrm{Re}}

\renewcommand{\emph}[1]{\textit{#1}}

\newcounter{para}

\makeatletter
\newcommand*\bigcdot{\mathpalette\bigcdot@{.5}}
\newcommand*\bigcdot@[2]{\mathbin{\vcenter{\hbox{\scalebox{#2}{$\m@th#1\bullet$}}}}}

\newcommand{\llangle}[1][]{\savebox{\@brx}{\(\m@th{#1\langle}\)}%
  \mathopen{\copy\@brx\kern-0.5\wd\@brx\usebox{\@brx}}}
\newcommand{\rrangle}[1][]{\savebox{\@brx}{\(\m@th{#1\rangle}\)}%
  \mathclose{\copy\@brx\kern-0.5\wd\@brx\usebox{\@brx}}}

\makeatother

\usepackage{array}
\newcolumntype{L}{>{$}l<{$}} 
\newcolumntype{C}{>{$}c<{$}} 
\newcolumntype{R}{>{$}r<{$}} 

\usepackage{xcolor}

\usepackage{mathtools}

\usepackage{physics}
\usepackage{bbm}
\usepackage{booktabs}

\newtheorem*{theorem*}{Theorem}
\newtheorem*{lemma*}{Lemma}
\newtheorem{proposition}{Proposition}

\usepackage{comment}
\usepackage{ragged2e}
\usepackage{atbegshi,picture}
\usepackage{multirow}
\usepackage{macros}
\usepackage{enumitem}

\usepackage[colorinlistoftodos]{todonotes}  
\usepackage{marginnote}

\usepackage{amsmath}
\makeatletter
\let\save@mathaccent\mathaccent
\newcommand*\if@single[3]{%
  \setbox0\hbox{${\mathaccent"0362{#1}}^H$}%
  \setbox2\hbox{${\mathaccent"0362{\kern0pt#1}}^H$}%
  \ifdim\ht0=\ht2 #3\else #2\fi
  }
\newcommand*\rel@kern[1]{\kern#1\dimexpr\macc@kerna}
\newcommand*\widebar[1]{\@ifnextchar^{{\wide@bar{#1}{0}}}{\wide@bar{#1}{1}}}
\newcommand*\wide@bar[2]{\if@single{#1}{\wide@bar@{#1}{#2}{1}}{\wide@bar@{#1}{#2}{2}}}
\newcommand*\wide@bar@[3]{%
  \begingroup
  \def\mathaccent##1##2{%
    \let\mathaccent\save@mathaccent
    \if#32 \let\macc@nucleus\first@char \fi
    \setbox\z@\hbox{$\macc@style{\macc@nucleus}_{}$}%
    \setbox\tw@\hbox{$\macc@style{\macc@nucleus}{}_{}$}%
    \dimen@\wd\tw@
    \advance\dimen@-\wd\z@
    \divide\dimen@ 3
    \@tempdima\wd\tw@
    \advance\@tempdima-\scriptspace
    \divide\@tempdima 10
    \advance\dimen@-\@tempdima
    \ifdim\dimen@>\z@ \dimen@0pt\fi
    \rel@kern{0.6}\kern-\dimen@
    \if#31
      \overline{\rel@kern{-0.6}\kern\dimen@\macc@nucleus\rel@kern{0.4}\kern\dimen@}%
      \advance\dimen@0.4\dimexpr\macc@kerna
      \let\final@kern#2%
      \ifdim\dimen@<\z@ \let\final@kern1\fi
      \if\final@kern1 \kern-\dimen@\fi
    \else
      \overline{\rel@kern{-0.6}\kern\dimen@#1}%
    \fi
  }%
  \macc@depth\@ne
  \let\math@bgroup\@empty \let\math@egroup\macc@set@skewchar
  \mathsurround\z@ \frozen@everymath{\mathgroup\macc@group\relax}%
  \macc@set@skewchar\relax
  \let\mathaccentV\macc@nested@a
  \if#31
    \macc@nested@a\relax111{#1}%
  \else
    \def\gobble@till@marker##1\endmarker{}%
    \futurelet\first@char\gobble@till@marker#1\endmarker
    \ifcat\noexpand\first@char A\else
      \def\first@char{}%
    \fi
    \macc@nested@a\relax111{\first@char}%
  \fi
  \endgroup
}
\makeatother

\begin{document}
\newcommand{\MITphys}{Center for Theoretical Physics---a Leinweber Institute, Massachusetts Institute of Technology, Cambridge, MA 02139, USA}
\newcommand{\MITstats}{Statistics and Data Science Center, Massachusetts Institute of Technology, Cambridge, MA 02139, USA}
\newcommand{\MITeecs}{Department of EECS, Massachusetts Institute of Technology, Cambridge, MA 02139, USA}
\newcommand{\nocontentsline}[3]{}
\let\origcontentsline\addcontentsline
\newcommand\stoptoc{\let\addcontentsline\nocontentsline}
\newcommand\resumetoc{\let\addcontentsline\origcontentsline}

\preprint{MIT-CTP/5936}
\title{
How much can we learn from quantum random circuit sampling?
}
\date{\today}
\author{Tudor Manole}
\thanks{These authors contributed equally to this work.}
\affiliation{\MITstats}
\author{Daniel K. Mark}
\thanks{These authors contributed equally to this work.}
\affiliation{\MITphys}
\author{Wenjie Gong}
\affiliation{\MITphys}
\author{Bingtian Ye}
\affiliation{\MITphys}

\author{Yury Polyanskiy}
\email{yp@mit.edu}
\affiliation{\MITstats}
\affiliation{\MITeecs}
\author{Soonwon Choi}
\email{soonwon@mit.edu}
\affiliation{\MITphys}

\begin{abstract}
Benchmarking quantum devices is a foundational task for the sustained development of quantum technologies. However, accurate \textit{in situ} characterization of large-scale quantum devices remains a formidable challenge: such systems experience many different sources of errors, and cannot be simulated on classical computers.
Here, we introduce new benchmarking methods based on random circuit sampling (RCS), that substantially extend the scope of conventional approaches. Unlike existing benchmarks that report only a single quantity---the circuit fidelity---our framework extracts rich diagnostic information, including spatiotemporal error profiles, correlated and contextual errors, and biased readout errors, without requiring any modifications of the experiment.
Furthermore, we develop techniques that achieve this task without classically intractable simulations of the quantum circuit, by leveraging \textit{side information}, in the form of bitstring samples obtained from reference quantum devices. 
Our approach is based on advanced high-dimensional statistical modeling of RCS data. We sharply characterize the information-theoretic limits of error estimation, deriving matching upper and lower bounds on the sample complexity across all regimes of side information.
We identify surprising phase transitions in learnability as the amount of side information varies. We demonstrate our methods using publicly available RCS data from a state-of-the-art superconducting processor, obtaining \textit{in situ} characterizations that are qualitatively consistent yet quantitatively distinct from component-level calibrations. Our results establish both practical benchmarking protocols for current and future quantum computers and fundamental information-theoretic limits on how much can be learned from RCS data.

\end{abstract}

\maketitle

\stoptoc
\section{Introduction}
\begin{figure*}[t!]    
\includegraphics[width=\textwidth]{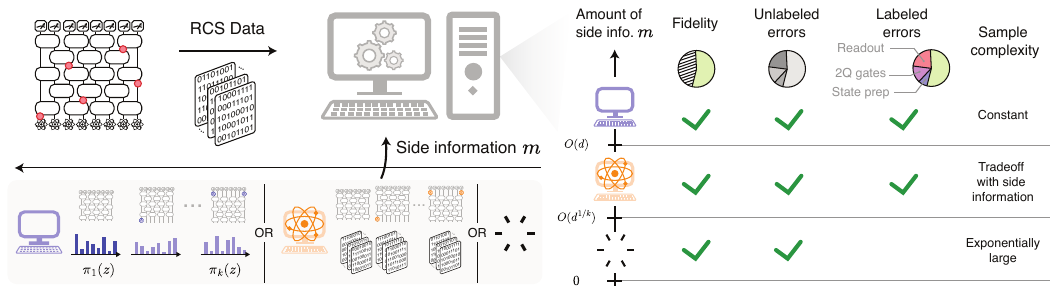}
    \caption{Overview of this work. We develop methods to learn many noise parameters from random circuit sampling (RCS) bitstring data.
    Our protocol makes use of side information, which can be  the classically computed bitstring distribution, or samples of these distributions obtained from a reference quantum computer.    
    Our analysis includes the special case where no side information is available --- even in this case, error rates can be learned given enough RCS data.
    Our method allows extracting more information than previously known benchmarking methods:
    in addition to the state fidelity, we can estimate the error rates for many types of errors,
    including state preparation errors, correlated multi-qubit errors, contextual errors which
    depend on previous gates applied, and readout errors. We find a phase diagram that dictate the hardness, and types of information that can be learned from RCS data, as a function of the amount of side information available.
    The sample complexity in each regime is analyzed.
    } 
    \label{fig:figure0}
\end{figure*}

Quantum information processing has made remarkable progress in recent years, reaching major milestones such as the demonstration of beyond-classical computational tasks~\cite{arute2019quantum, shaw2023benchmarking, bluvstein2024logical,haghshenas2025digital}, and the realization of quantum error correction and early fault-tolerant operations~\cite{egan2021fault,google2023suppressing,sivak2023real,bluvstein2024logical,paetznick2024demonstration,google2025quantum,bluvstein2025architectural}.
As quantum devices become increasingly advanced, we also need improved methods to characterize and benchmark them. 

Quantum processors experience a variety of errors, including coherent errors~\cite{shaw2024universal,kaufmann2025characterization,steckmann2025error}, errors that vary over space and time~\cite{arute2019quantum}, correlated multi-qubit errors~\cite{mcewen2022resolving,wilen2021correlated}, leakage errors~\cite{mcewen2021removing,scholl2023erasure,ma2023high}, and \textit{contextual} errors whose presence depend on the choice of earlier operations.
This diversity of error types reflects the diversity of physical mechanisms for error in quantum devices. Examples include excitation into higher transmon levels~\cite{mcewen2021removing}, spontaneous emission of photons from atoms~\cite{scholl2023erasure,ma2023high}, and slow ringdown~\cite{wei2024characterizing} or fluctuations~\cite{shaw2024universal} control pulses. These various types of errors must be identified and quantified in the effort to improve quantum hardware components, their interconnects, and systems architecture~\cite{wei2024characterizing}.
In this regard, traditional methods often  fall short in accurately characterizing complex devices.
In most common practices, a quantum system is characterized component-by-component and operation-by-operation, separately, rather than when they work together as in a full quantum circuit~\cite{arute2019quantum}.
In such approaches, certain types of errors may be missed or incorrectly estimated~\cite{chen2025randomized}.

Cross-entropy benchmarking (XEB) is the state-of-the art approach to characterize large scale quantum devices~\cite{boixo2018characterizing}. The XEB fidelity is a proxy
for the quantum fidelity and provides
a single-number
summary of circuit performance.
It is remarkably versatile, being
applicable across different hardware platforms, both for physical and logical circuits, and beyond the ideal setting of deep random unitary
circuits~\cite{mark2023benchmarking,shaw2023benchmarking,bluvstein2024logical,andersen2025thermalization}; hence it is one of a few industry-standard approaches for benchmarking quantum systems~\citep{baldwin2022re,cross2019validating,mayer2021theory,proctor2022measuring}, and has seen wide adoption~\cite{arute2019quantum, wu2021strong, liu2021benchmarking,shaw2023benchmarking, bluvstein2024logical,decross2024computational,
andersen2025thermalization, gao2025establishing}.
However, XEB has two major limitations. First, it aggregates all types of errors into a single metric---the state infidelity---thereby obscuring detailed information needed to guide improvements in quantum hardware. Second, it relies on classical simulation of the ideal circuit output, which limits its applicability to relatively small system sizes.

In this work, we present new benchmarking methods based on random circuit sampling (RCS).
Our approach requires no modification in experiments, as it relies on exactly the same data needed for XEB. However, our methods significantly extend the scope of the XEB fidelity, largely addressing the aforementioned limitations and providing detailed information about noise.
Explicitly, starting from the bitstring data obtained from noisy RCS, our method produces an elaborate report that contains a rich set of information such as the estimated circuit fidelity, spatiotemporal profiles of single- or two-qubit errors, correlated or contextual errors,   biased readout errors, and, in certain cases
identifies the dominant physical processes behind the errors.
This report can be made in a computationally efficient manner, provided sufficient amount of \textit{side information} that
describe the expected sampling distributions in the absence of unwanted errors.
We consider three different regimes of side information: A) maximal side information 
where the expected output distribution is available from explicit classical computation, 
B) partial side information, where the expected distributions can be inferred from $m<\infty$ bitstring samples generated from a clean reference quantum computer, and 
C) no side information where there is no reference data (Fig.\,\ref{fig:figure0}).
Formally, A) and C) correspond to $m=\infty$ and $m=0$, respectively.
The regimes B) and C) are increasingly relevant for demonstrations of beyond-classical circuit sampling.

Our approach to benchmarking is based on improved modeling of noisy random circuit
output, building upon
a model introduced by~\cite{rinott2022statistical}.
We show that this modeling approach
adequately resolves the different types and spacetime positions of errors, using 
advanced statistical data-processing
algorithms 
of the high-dimensional measurement data.   
In statistical terms, we describe RCS
data as arising from a mixture of 
random high-dimensional
probability distributions,
representing the various error channels present in the
random quantum circuit. 
We provide a detailed description
of our statistical setup, which is accessible
to statisticians,  
in Appendix~\ref{sec:estimators}.  

Our algorithms require more RCS samples when there is less side-information available. Interestingly, we discover \textit{phase transitions} in the sample complexity
of error learning, 
which allow us to precisely characterize
the regimes in which side information is beneficial. In all settings, we prove that our methods achieve optimal sample complexity, saturating information-theoretic lower bounds. 
These sample complexity bounds
provide valuable insight into the scope of possible
applications of our benchmarking methods. As
an example, our results imply that one can benchmark the full execution of an RCS
experiment on  a quantum processor of approximately
50 qubits in an efficient manner both in memory and  computation, given sufficient
side information from 
a reference quantum computer. 

Finally, we apply our methods to existing, publicly available RCS data~\cite{arute2019quantum} to characterize the diverse sources of error that arise 
in a state-of-the-art quantum processor.
Unlike earlier error characterizations of the quantum processor, our report (Fig.~\ref{fig:google_data}) provides \textit{in-situ} information about the errors experienced in full operation of the circuit.
The extracted error rates and breakdown into different sources are consistent with the expected behavior.

\section{Overview}
Our method builds upon the operating principle of cross-entropy benchmarking (XEB)~\cite{boixo2018characterizing}: that the output of a random circuit is a highly entangled wavefunction $|\psi\rangle$ which is a superposition over all $N$-bit strings. Measuring this state in the computational basis amounts to sampling from the probability distribution $\pi_1(z) \equiv |\langle z| \psi \rangle|^2$ over $d=2^N$ possible bitstrings. In a typical circuit, this probability distribution is highly complex: its individual entries $\pi_1(z)$ fluctuate strongly among bitstrings $z$,
and form a many-body \textit{speckle pattern} that is essentially unique to the state $|\psi\rangle$ (Fig.~\ref{fig:figure0}) in practical settings\footnote{We note that states of the form $\exp(i\theta P)|\psi\rangle$ (for some Pauli operator $P$ consisting of $Z$ operators) displays the same speckle pattern $\pi_1(z)$. We are indifferent to those cases since they do not change computational outcomes anyways.}. Moreover, for a sufficiently deep circuit, $|\psi\rangle$ is well-approximated as a Haar-random state~\cite{harrow2009random} and hence $\pi_1(z)$ satisfies universal statistical properties on the distribution of $\{\pi_1(z)\}$ values, known as the Porter-Thomas distribution~\cite{porter1956fluctuations,arute2019quantum}.
This universality provides the foundation for the XEB fidelity.

The XEB fidelity is an approximation of the quantum fidelity $\langle \psi| \rho |\psi\rangle $ between the ideal target state $|\psi\rangle$  produced by a programmed quantum circuit and the experimentally prepared mixed state $\rho$. It compares experimental bitstrings against the ideal distribution $\pi_1$ to estimate the fraction of samples drawn from $\pi_1$.
The \textit{white noise model}, often evoked to explain the XEB fidelity, makes this notion of the ``fraction" concrete. It posits that the experimental distribution $p(z) \equiv \langle z |\rho |z\rangle$ is of the form
\begin{equation}
   \text{White noise model: } p(z) = F \pi_1(z) + (1-F)/d. 
   \label{eq:white_noise_model}
\end{equation}
Experimental samples are therefore drawn from $\pi_1(z)$ with probability $F$ and from the uniform distribution $\pi_\text{wh}(z)\equiv 1/d$ with probability $1-F$. This relation arises, for example, when $\rho = F|\psi\rangle\langle\psi| + (1-F)I_d/d$
is the globally depolarized state (where $I_d/d$ is the maximally mixed state), but can also arise from the combined effect of many errors throughout the circuit~\cite{dalzell2024random}.
The operational meaning of $F$ is the probability that an experiment successfully executes the circuit with no error, hence estimating the quantum fidelity.

In this work, we adopt a more refined \textit{$k$-component model}~\cite{rinott2022statistical}
\begin{equation}
    \text{$k$-component model: } p(z) = \sum_{i=1}^k c_i \pi_i(z)~.
\end{equation}
Each term $\pi_i(z)$ is the distribution associated with each physical error pattern $i$.
For example, during one execution of a deep circuit, a particular qubit at site $a$ may experience a Pauli $X$ error at a specific time $t$. Alternatively, a pair of qubits at site $a$ and $b$ experience Pauli $X$ and $Y$ errors at times $t_1$ and $t_2$, respectively.
We imagine all such patterns of errors (``events") that may reasonably occur in the circuit and enumerate them by the index $i$. We identify the special case $i=1$ with the perfect execution of the circuit.

Formally, we can understand that the index $i$ enumerates over the ensemble of \textit{quantum trajectories} obtained from an unraveling of error channels~\cite{nielsen2010quantum}. Each trajectory is associated with a pure wavefunction evolving under the programmed unitary circuit interspersed by error (Kraus) operators at specific spacetime locations. Rapid scrambling of random unitary circuits~\cite{nahum2018operator,fisher2023random} ensure that such trajectory states are, with high probability, also Haar-random and hence their measurement distribution $\pi_i(z)$ can be approximated as vectors independently sampled from the Porter-Thomas distribution~\cite{rinott2022statistical,shaw2024universal}.   

In the simplest version of our
protocols, we essentially perform cross-entropy benchmarking on each distribution $\pi_i(z)$ to estimate its coefficient $c_i$, representing the probability of the particular error event $i$. 
We consider a total of $k$ error events. Advanced statistical algorithms allow us to utilize the information provided by bitstring measurements and simultaneously and efficiently estimate all $c_i$.
This enables \textit{in-situ} characterization of errors in a quantum circuit, complementing prevailing approaches of error characterization by single- and two-qubit experiments. Error rates may differ between these isolated experiments and full operation of a quantum circuit when all components are simultaneously in operation.
This also enables the characterization of complex error types such as correlated and contextual errors.

For large quantum systems, obtaining exact knowledge of $\pi_i(z)$ for many different error patterns is computationally intractable.
To this end, we study the error estimation task when $\pi_i(z)$ is known only partially.
Specifically,  
we envision that information about $\pi_i(z)$ is obtained by sampling from a \textit{reference} quantum computer that prepares the ideal state $|\psi\rangle$ and any of the noise trajectory states. 
Such states can be obtained, for example, by stringent quantum error detection~\cite{bluvstein2024logical,ma2023high,scholl2023erasure}.
Drawing inspiration from a related line of work in the statistical literature~\citep{angelopoulos2023prediction,xia2024,gerber2024likelihood}, we dub this \textit{side information}, which enables our use of the experimental data for parameter estimation.
We quantify the amount of side information as the number of measurements $m$ of each trajectory state, yielding bitstrings $W_{im}$ sampled from each $\pi_{i}$. By comparing $W_{im}$ and bitstring data from noisy RCS, we produce the same kind of benchmark report.

We find phase transitions in sample complexity as a function of the amount of side information (Fig.~\ref{fig:figure0}): in the full information, $m=\infty$ phase (which includes the case where $\pi_i$ are classically simulated), the $c_i$'s can be estimated with a number of samples independent of system size. As $m$ decreases and crosses the phase boundary at $m=O(d)$ (\textit{i.e.}, sublinear in $d$), we enter the partial side-information phase which features a sample complexity tradeoff: the sample complexity is set by the \textit{product} $n\times m$ of the numbers of experimental and side-information samples.  Estimation is feasible as long as $nm \geq d \log k$.
Surprisingly, even in the $m=0$ limit with no side information, estimation is still possible by making use of the universal Porter-Thomas properties of each $\pi_i$. This transition occurs at $m=O(d^{1/k})$ which depends on the number of errors $k$ considered in our model. In this phase, only the \textit{unlabeled} $\{c_i\}$ can be estimated, i.e.~we can determine the values of the $c_i$'s, such as the largest $c_i$, but cannot assign the indices $i$ to each value. Detailed expressions for the sample complexities are summarized in Table~\ref{tab:sample_complexities}.

Across all regimes, the key properties we utilize are the typical properties of Porter-Thomas (or Dirichlet random) distributions $\pi_i$. These properties \textit{concentrate}: random instances are close to the average value with (exponentially) small fluctuations. For example, two random distributions have fixed overlap $\sum_z \pi_i(z) \pi_j(z) = (1+\delta_{ij})/d + O(d^{-3/2}) $ and hence are in some sense approximately orthonormal.
This bilinear structure is tremendously helpful: the product 
$\sum_z \pi_i(z) p(z)$ can be estimated with only a few samples from 
a distribution $p(z)$, far fewer than the (exponentially large) number needed to accurately estimate each entry $p(z)$. As long as the number of signals $k$ is   less than the dimension $d$, $k$ different products $\sum_z \pi_i(z) p(z)$ can be straightforwardly distinguished with minimal overhead.

In the partial side information phase, we utilize a \textit{collision estimator} which counts the number of bitstrings $z$ seen in both the RCS and reference data.
The Porter-Thomas nature of the distributions
$\pi_i$ ensures that they fluctuate
around the uniform distribution, and in this
regime, the classical birthday paradox ensures that 
collisions begin to occur
once the \textit{product} of sample sizes
$n\times m$ exceeds
the dimension $d$~\citep{wendl2003collision}. 
This is precisely the threshold
at which our methods can reliably estimate $c_i$,
up to a logarithmic correction in $k$.

Even in the absence of \textit{any} side information ($m=0$), the typicality of high-dimensional
Dirichlet random distributions means that information about $c$ is present in the measurement data, independent of 
the distributions $\pi_i$, and hence not requiring knowledge of them. As a simple example, under Eq.~\eqref{eq:our_model}, the sum-of-squares $\sum_i c_i^2$ is well estimated by the collision probability $d\sum_z p(z)^2-1$~\cite{shaw2024universal,andersen2025thermalization}. This, and higher moments of $p(z)$, can be estimated with enough RCS data.

The high-dimensional
nature 
of RCS data---in which the Hilbert space
dimension $d$ significantly exceeds the sample
size $n$---is both a blessing and  
a curse throughout our analysis. 
This high-dimensionality
enables the universal
behavior of random
Porter-Thomas distributions,
leading to particularly simple 
and practical algorithms that generalize
XEB, 
 while on the other hand placing lower
 bounds on the sample complexity 
 of error characterization which 
 sometimes grow
 exponentially in system size. 
 From a technical lens, 
 our analysis leverages tools
 from high-dimensional statistical theory~\cite{wainwright2019},
 though in the setting of count-based 
 data which is heteroscedastic in nature, 
 unlike the more
 classical setting of Gaussian additive noise.  

We show that our rates of estimation are optimal. We establish information-theoretic lower bounds on the sample complexity of the inference task. We establish matching upper bounds by explicitly presenting optimal statistical estimators, sharply resolving the question of parameter estimation with RCS data. 

We proceed to apply our methods to synthetic data obtained from numerical simulations as well as publicly available data produced in a quantum experiment~\cite{arute2019quantum}. We confirm that we can identify time-varying error rates and non-local correlated errors. From the real-world data, we successfully estimate state preparation, rates of errors affecting single- and two-qubits, as well as biased readout errors, resolved on each qubit. Our results are qualitatively consistent with anticipated values, but, crucially, our approach estimates those error rates \textit{in-situ} whereas the previously reported values rely on data obtained from separate experiments.

On one hand, our explicit estimators and algorithms imply that our approach is the best approach for parameter estimation from RCS data. In practical terms, we expect our methods to be applicable to quantum circuits of up to $N=40-50$ (possibly logical) qubits, even in the \textit{beyond-classical regime}.  On the other hand, our lower bound implies we cannot do better than this: there is a fundamental limit on how much we can learn from RCS data, and each side-information or experimental measurement contains an exponentially small amount of data. 
In order to circumvent our lower bound, one must use quantum circuits with additional structure, such as those in mirror benchmarking protocols~\cite{mayer2021theory,proctor2022scalable,gong2025robust} or random Clifford circuits~\cite{magesan2011scalable,magesan2012characterizing}.
\\
 
\noindent {\bf Related Work.} As mentioned previously, the work of~\cite{rinott2022statistical}
was the first to propose 
model~\eqref{eq:our_model},
and statistical estimators 
for $c$ when $m=\infty$. 
When $m=0$, Ref.~\cite{arute2019quantum}
introduced a method called
speckle XEB, for estimating
the  XEB fidelity  under the white noise model~\eqref{eq:white_noise_model}   without knowledge of $\pi_1$, using the empirical second moment of the bitstring data (subsequently termed ``self-XEB" in Ref.~\cite{andersen2025thermalization}). 
Their procedure
was generalized by the
work of~\cite{shaw2024universal} to model~\eqref{eq:our_model},
who showed that higher-order empirical moments of the 
bitstring data can be used
to estimate the moments of 
the unordered vector $c$.
They also derived an upper bound on the sample complexity of estimating
the second moment of the vector $c$, which can be viewed
as a precursor to our
sample complexity bounds to come,
for the special case $k=2$. 

Estimating the  overlap fidelity
between two quantum 
states prepared on separate
quantum computers has been studied as 
``cross-platform verification" in Refs.~\cite{greganti2021cross,zhu2022cross}, and specifically in the context of randomized measurements in Ref.~\cite{elben2020cross}. 
However, not much is understood about its sample complexity, and the task of learning multiple parameters in a cross-platform approach has not been 
explored before our work.

\section{Learning from bitstrings}
\label{sec:rcs}
\subsection{Random circuit sampling as a statistical mixture model}
Having provided a high-level overview of our results, in this section we begin our technical discussion.
Our benchmarking methods are developed for RCS data~\cite{boixo2018characterizing}, in which a quantum processor
executes a circuit obtained
by composing randomly-sampled single- and two-qubit gates, a popular benchmark in the field with several demonstrations~\citep{arute2019quantum,wu2021strong,liu2021benchmarking,zhu2022quantum,moses2023racetrack,decross2024computational,bluvstein2024logical,gao2025establishing}.

Random circuit sampling has the advantage of being an unbiased measure of quality of a quantum device: its underlying gates are randomized and hence the measurement outcomes are not biased towards one particular type of error, nor are they tailored to a particular circuit, which may have highly structured outcomes. In this respect, it is similar to randomized benchmarking~\cite{knill2008randomized,magesan2011scalable,wallman2014randomized, harper2019statistical,erhard2019characterizing, proctor2022scalable} as well as the suite of tools known as the \textit{randomized measurement toolbox}~\cite{elben2023randomized}. 

Originally motivated as a quantum advantage demonstration,
 RCS has since become a 
general-purpose and widely-used tool for benchmarking quantum hardware~\citep{lall2025review}. 
Furthermore, random quantum circuits 
are good approximations for Haar-random unitaries~\cite{harrow2009random}, 
which makes them useful for quantum information tasks~\cite{ambainis2007quantum} including state learning~\cite{huang2020predicting,elben2023randomized} and random-number generation, as well as for the study of questions in basic science, such as quantum
chaos and thermalization~\citep{zhou2019emergent,hayden2007black,nahum2018operator}.

In an RCS experiment,
an $N$-qubit random circuit produces an output state $|\psi\rangle \equiv U |\psi_0\rangle$, which is then measured in the computational basis. Each experiment yields a random $N$-bitstring $z\in \{0,1\}^N$, 
whose  probability distribution is given by $\pi_1(z)\equiv|\langle z| \psi\rangle|^2$
in the noiseless case. In practice, there are uncontrolled errors, and a noisy physical quantum device instead
 transforms  $|\psi_0\rangle$ into a 
 mixed state $\rho$, with a corresponding measurement distribution $p(z) = \langle z|\rho|z\rangle$. 

The most widely-used
technique for analyzing RCS data
is the (linear) cross-entropy benchmark
(XEB)~\cite{boixo2018characterizing,arute2019quantum}. 
The XEB estimates the many-body fidelity $F\equiv\langle \psi|\rho|\psi\rangle$ by comparing the experimental distribution $p(z)$ with the ideal one $\pi_1(z)$:
\begin{equation} \label{eq:fidelity} 
F_\text{XEB} = d \sum_{z \in \{0,1\}^N} p(z)\pi_1(z)-1.
\end{equation} 
Experimental samples provide an \textit{empirical} estimate of $p(z)$ which furnishes an unbiased estimator $\hat{F}_\text{XEB}$ (Appendix~\ref{sec:estimators}). In turn, the XEB furnishes an accurate approximation of 
the quantum fidelity, $F_\text{XEB}\approx F$
in sufficiently deep random circuits with local noise~\cite{ware2023sharpphasetransitionlinear,morvan2024phase}. 

The above relation is justified by two properties. First, random unitary circuit dynamics results in distributions $\pi_1$ with highly typical properties. This is referred to by the
{\it Porter-Thomas} distribution~\citep{porter1956fluctuations,boixo2018characterizing} which governs the distribution of values $\pi_1(z)$. Mathematically, this is equivalent to assuming that $\pi_1$ is a random probability vector sampled from the \textit{Dirichlet distribution}, i.e.~uniformly random on the probability simplex
(see condition~\ref{assm:pt}
below for a formal
definition). Second, one needs to make an assumption about the bitstring output of the ``noisy part" of the state. The simplest such model is the white noise model [Eq.~\eqref{eq:white_noise_model}]. However, the XEB remains an accurate estimate of the fidelity even in the more general situation, repeated here:
we assume that the quantum device
may experience $k-1$ different 
incoherent error patterns, and thus
that the random state $\rho$ outputs bitstrings according to the bitstring distribution:
\begin{align} 
\label{eq:our_model}
p_c(z|\Pi) = \sum_{i=1}^k c_i \pi_i(z),\quad z \in \{0,1\}^{N},
\end{align}
where $\pi_1$ is the ideal {\it random} probability distribution,  $\pi_2,\dots,\pi_k$
are the {\it random} probability
distributions of $k$ different incoherent error sources
in the circuit, and $c=(c_1,\dots,c_k)$ is the corresponding
vector of probabilities (``error weights"). See Refs.~\cite{rinott2022statistical,shaw2024universal} for similar models. We collect the distributions $\pi_i$ into a single matrix $\Pi \in \mathbb{R}^{k\times d}$ with entries $\Pi_{ij} = \pi_i(z_j)$, where the index $i$ denotes the error type, and $j$ denotes the bitstring index, and explicitly highlighted the dependence of $p_c$ on $\Pi$.
The matrix $\Pi$ is random and depends on the choice of random circuit and on the errors in the model (see Eq.~\eqref{eq:error_channels} below). 
In statistical language, 
we recognize model~\eqref{eq:our_model} as a mixture
model consisting of $k$ components, 
the first of which corresponds
to the ideal bitstring
distribution $\pi_1$, which occurs with a probability $c_1$ that can be viewed as an analogue of the XEB fidelity $F$. The
remaining terms of the mixture model
correspond to $k-1$ noise sources $\pi_i$, 
each occurring with probability $c_i$.

Eq.~\eqref{eq:our_model} is not only more physically realistic, it also allows for learning beyond the single-number
summary of the circuit infidelity provided by XEB. Learning the coefficients $c_i$ in our model provides detailed information about the error processes that contribute to this infidelity. 

Physically, Eq.~\eqref{eq:our_model} arises from 
the following model of the noisy state 
\begin{equation}
    \rho = \mathcal{R}_{J} \circ \mathcal{U}_J \circ \cdots \circ \mathcal{U}_1\circ \mathcal{R}_0 [|\psi_0\rangle \langle\psi_0|]
\end{equation} 
where $\mathcal{R}_i$ denotes the error channel at circuit layer $i$ and $\mathcal{U}_i[~\cdot~] \equiv U_i[~\cdot~] U^\dagger_i$ denotes the ideal quantum unitary acting on layer $i$. Each error channel consists of a number of physical errors, denoted by Kraus operators $K_\ell$~\citep{nielsen2010quantum}:
\begin{equation}
\mathcal{R}_i[\rho] = \sum_\ell \Gamma_{\ell}^{(i)} K_{\ell} \rho K_{\ell}^\dagger,
\label{eq:error_channels}
\end{equation}
with physical error rates $\Gamma_\ell^{(i)}$ that can depend both on error type, location, and time (layer). 

Eq.~\eqref{eq:our_model} is related to \eqref{eq:error_channels} by the following: Each distribution $\pi_i$ is associated with a particular \textit{trajectory} $(K_{\ell_0}, K_{\ell_1},\cdots, K_{\ell_J})$ that denotes 
a sequence of Kraus operators. For instance, the ideal distribution $\pi_1$ corresponds 
to the trajectory where all $K_{\ell_i} = \mathbb{I}$, and $c_1$ is the 
probability that no error occurred. Other trajectories include those where 
one $K_{\ell_i}$ is non-trivial, indicating the occurrence of an error at layer $i$, 
of type $\ell_i$. Its corresponding distribution is given by
\begin{equation}
\label{eq:pi_Kraus}
\pi_{(\ell_i,i)}(z) = |\langle z|U_J \cdots K_{\ell_i} U_i \cdots U_1|\psi_0\rangle|^2.
\end{equation}
When the operators $K_{\ell_i}$ are unitary, e.g.~for a Pauli error channel,
equation~\eqref{eq:pi_Kraus}
is a probability distribution (non-negative and summing to 1), and due to the operator spreading~\cite{fisher2023random} in the random circuit, each $\pi_{(\ell_i,i)}(z)$ is an independent Dirichlet-random distribution (see condition~\ref{assm:pt}
below). 

This condition is necessary for the statistical model \eqref{eq:our_model} and our theoretical analysis, but will not be necessary for the analysis of real data in Section~\ref{sec:google}: our estimators are robust to deviations from Assumption~\ref{assm:pt}.

Current XEB approaches typically require complete knowledge of the ideal distribution $\pi_1$, which
is extremely challenging to classically simulate when the system size $N$ is approximately greater than
30.
In practice, a typical workaround is to estimate the XEB fidelity based on patches of disconnected circuits, or with specially structured circuits that can be simulated~\citep{arute2019quantum}. However, such methods require dedicated experiments and are not guaranteed to provide an accurate estimate of the global circuit fidelity. 
Relaxing the assumption that $\pi_1$ is classically computable not only addresses practical needs, it also defines a theoretically rich statistical problem. Fixing notation we will use in the rest of this work, we denote the bitstring measurements (``RCS data") as i.i.d. samples
\begin{align} \label{eq:main_sample}
Z_1,\dots,Z_n \,\big|\,\Pi \sim p_c(\cdot|\Pi),
\end{align}
where we have explicitly highlighted the dependence of $p_c$ on the random circuit realization. This determines the matrix~$\Pi$, a deterministic function of
the random choice of circuit, which we equivalently treat as a random variable in itself~\ref{assm:pt}. We will frequently summarize these measurements in terms of the empirical counts
$Y_j= \sum_{i=1}^n   \delta_{Z_i,z_j}$,
indexed by $j=1,\dots,d$. For example, the statistical estimator for the XEB 
fidelity is simply $\hat F_n = (d/n) \sum_{j=1}^d Y_j\pi_1(z_j)-1$, 
arising from the approximation $p(z_j)\approx Y_j/n$. 

We additionally assume that the practitioner has access
to {\it side information} in the form of $m$ bitstring samples drawn from a 
reference quantum computer
which perfectly implements the ideal circuit $\pi_1$, and any of the noisy
circuits $\pi_i$:
\begin{align} 
\label{eq:side_info}
W_{i1},\dots,W_{im} \,\big|\,\Pi \sim \pi_{i},\quad i=1,\dots,k.
\end{align} 

The parameter $m$ allows us to systematically analyze different classes of protocols.
On one extreme, when $m=\infty$, 
we interpret the matrix $\Pi$
as being perfectly known. This corresponds
to the conventional situation
in which the bitstring distributions
can be classically simulated, as in the XEB setup. 
Meanwhile, the $m= 0$ limit 
indicates that no side information is given (see Refs.~\cite{shaw2024universal,andersen2025thermalization} for earlier work in this limit). This regime is particularly relevant for large system sizes and deep circuits, where no classical simulation or reference quantum computation is viable. Even in 
this situation, our benchmarking algorithms
provide nontrivial information about the characteristics of noise.
Our information-theoretic
phase transitions indicate that the boundaries 
of these two regimes
occur at $m\geq d$ and $m\leq d^{1/k}$,
respectively.
When $m$ lies between these two extremes, a reference quantum computer provides partial \textit{side information} about $\Pi$ in the form of samples from the ideal distribution $\pi_1$, and all noisy distributions $\pi_i$, and we develop 
algorithms which efficiently leverage this side
information.


\subsection{Estimators}
 \label{sec:main_estimators}
Given $n$ samples from a distribution of the form Eq.\,\eqref{eq:our_model}, our task is to estimate the error weights $c=(c_1,\dots,c_k)$, with the aid of $m$ samples of side information that give us knowledge about $\Pi$.

How might it be possible to estimate a large number $k$ 
of parameters from a single realization of a random circuit? The key is that our data is high-dimensional: they are samples drawn from a $d=2^N$ dimensional probability distribution $p(z)$. For a collection of $k$ different circuits (here representing the original circuit with injected errors), it is highly likely that their output distributions are linearly independent. In other words, in a high-dimensional space, different errors distort the output distribution in different ways and hence can be distinguished in the measurement data.

We develop several estimators to estimate the parameter vector $c=(c_1,\dots,c_k)$, suitable suitable for various regimes of side information $m$. 
We discuss several
of these estimators in what
follows, deferring a more complete
discussion to Appendix~\ref{sec:estimators},
including further discussion
of related statistical literature.

\subsubsection{Regime A: Classical Simulation ($m=\infty$)}
In the simplest case with classically-computed side information,  
 the matrix
$\Pi$ is known to the practitioner. 
In this regime, we
estimate the parameter $c_i$ in two steps. 
We first observe that
the products $\zeta_i = \sum_z \pi_i(z) p_c(z|\Pi)$
satisfy 
\begin{align} \label{eq:main_zeta_ortho} 
\zeta_i  
 =  \sum_\ell c_\ell  \sum_z \pi_i(z) \pi_\ell(z)
 =\frac {1+c_i} d + O(d^{-3/2}),
 \end{align}
as a result of the concentration of the Porter-Thomas
rows of $\Pi$.
One can form unbiased estimators $(1/n)\sum_{j=1}^d Y_j \pi_i(z_j)$ of 
these products, 
which leads to a first 
simple estimator of $c$:
\begin{align} 
\label{eq:main_ortho}
\hat c_i^{\mathrm{XEB}} = \frac d n\sum_{j=1}^d 
Y_j \pi_i(z_j) - 1,\quad i=1,\dots,k.
\end{align} 
This estimator was first
proposed by~\citep[Eq.~(5.1)]{rinott2022statistical}.
We refer to the vector $\hat c^{\mathrm{XEB}}$
as the (generalized) XEB
estimator. Much like the XEB fidelity
estimator $\hat F_n$, 
this generalized XEB estimator
 achieves a sample complexity
 which does not depend
on the Hilbert space dimension~$d$.
It does, however, depend
linearly on the number of errors $k$, 
and this dependence can be mitigated
by appropriate regularization. 
For example,
given an appropriate tuning parameter $\lambda
> 0$, 
we will show that
the {\it hard-thresholded}~\citep{donoho1994ideal}
XEB estimator 
\begin{align}\label{eq:HT} 
\hat c_i^{\mathrm{HT}}  
 = \begin{cases} 
 \hat c_i^{\mathrm{XEB}}, & \hat c_i^{\mathrm{XEB}}  > \lambda, \\
 0, &\mathrm{otherwise}, 
\end{cases}
\quad i=1,\dots,k,
\end{align}
has sample complexity 
which merely degrades logarithmically
with the number of errors $k$. 
Roughly speaking, this favorable
dependence on $k$ arises from the fact
that the vector $c$ has bounded $\ell_1$ norm,
and is therefore approximately sparse, 
a fact which is leveraged by estimator~\eqref{eq:HT}~\citep{donoho1994minimax,raskutti2011,li2018}.

As we discuss in Appendix~\ref{sec:estimators},
the XEB estimator
can be viewed as an approximation
of the ordinary
least squares estimator 
for performing a linear regression
of the histogram~$Y$ onto~$\Pi^\top$,  
 whereas
our model is a multinomial regression model for which the 
canonical estimator is the
maximum likelihood estimator 
(MLE), defined by
\begin{align} \label{eq:main_mle}
\hat{c}^\mathrm{MLE} = \argmax_{\gamma \in \Delta_k} \sum_{j=1}^d Y_j \log(\Pi_{\cdot j
}^\top \gamma).
\end{align}
This estimator was also noted by~\cite[Eq.~(5.2)]{rinott2022statistical}.
Much like in the RCS literature, where linearization of the XEB 
fidelity is typically adopted, we do not 
  find that 
the MLE and XEB estimators
  differ appreciably  
due to the small magnitude of $\Pi$,
and hence the 
  mild heteroscedasticity
of the histogram $Y$ (cf.\, Appendix~\ref{sec:estimators}). 
Indeed,  we have found that the
XEB and MLE estimators perform similarly
when the independent Porter-Thomas
assumption on the rows of $\Pi$ holds.
The XEB estimator has the disadvantage
of not being robust to deviations
from this modeling assumption, 
but it has the advantage
of being     easily computable even for
large system sizes $d=2^N$, 
whereas the program~\eqref{eq:main_mle} can be 
 somewhat more costly to optimize 
despite its convexity. 
Another advantage of the MLE is the fact
that it is free of
tuning parameters, yet still
provides accurate estimates when $k$ is large~\citep{bing2022}.
Roughly speaking, this behavior is
due to the restriction of the optimization problem~\eqref{eq:main_mle}
to the simplex, which significantly reduces the volume of the search
space despite the potentially large magnitude of~$k$.

\subsubsection{Regime B: Partial Side Information ($1 < m < \infty$)}
 
When the amount of side information $m$ is finite but nonzero, we recommend the following adaptation of the 
XEB estimator:
\begin{align}
\label{eq:main_ortho_eiv}
\hat c_i^{\mathrm{coll}} = \frac {d} {nm}\sum_{\ell=1}^n \sum_{r=1}^m \delta_{Z_\ell,W_{ir}} - 1,\quad i=1,\dots,k, 
\end{align} 
Up to centering
and scaling, this estimator consists
of counting the number of collisions
between the primary bitstring samples
$\{Z_\ell\}$ and each of the side
samples~$\{W_{ir}\}$. 
Once again, $\hat c^{\mathrm{coll}}$ can  
be regularized using hard-thresholding,  and 
we will show that the resulting
estimator achieves the optimal sample
complexity of estimating~$c$.

An important
practical benefit
of the collision estimator is the fact
that its computational complexity
scales as $O(k(n+m))$, while
its memory complexity scales as $O(k\cdot \min\{n,m\})$,
neither of which depend
on the Hilbert
space dimension $d$. 
This estimator can therefore be used
in the beyond-classical regime where
objects of dimension $d$ cannot 
easily be stored in the memory
of a classical processor.
Another practical benefit of this estimator
is its linear structure, which allows it
to be updated as more bitstring
data becomes available,
without needing to be recomputed. 

As before, it is also natural to consider
the maximum likelihood estimator, which is now given by 
\begin{equation} 
\label{eq:side_info_mle}
\argmax_{\gamma\in\Delta_k}~ \log
  \int_{\Delta_d^k} ~
\prod_{j=1}^d \bigg((\Pi_{\cdot j}^\top\gamma)^{Y_{j}}
\prod_{i=1}^k \pi_{ij}^{V_{ij}} \bigg)~ d\Pi,
\end{equation}
where the integral is taken over the set of $k\times d$
matrices whose rows are constrained to the $d$-dimensional simplex.
Unlike equation~\eqref{eq:main_mle},
this optimization problem is nonconvex.
In Appendix~\ref{sec:estimators}, we develop
a heuristic optimization algorithm
for this problem
by interpreting equation~\eqref{eq:side_info_mle}
as a partition function which integrates
over states $\Pi$. 
Using a 
mean-field approximation, 
we   derive
an algorithm that maximizes
the corresponding 
variational Gibbs free entropy, 
and consists of iteratively solving
the following fixed-point equation
with respect to $\gamma \in \Delta_k$:
\begin{equation}
\label{eq:em_fixed_point}
n = \sum_{j=1}^d \frac{Y_j S_{ij}}{\sum_{r=1}^k S_{rj} \gamma_r},\quad 
\text{where } S_{ij} = \exp\{\psi(1+V_{ij})\},
\end{equation}
for $i=1,\dots,k$, where $\psi$ denotes the di-gamma function. This iteration
is, once again, computable with
time and memory
complexity that are independent of the Hilbert
space dimension $d$, since the histogram
$Y$ is supported on at most $n$ entries.
A close analogue
of the fixed point 
equation~\eqref{eq:em_fixed_point}   arises
in a statistical method for text analysis 
known
as latent Dirichlet allocation~\citep{blei2003}. 
In that context, it has
been argued~\citep{ghorbani2019instability,celentano2023local,celentano2023mean} that the mean-field approximation
can be significantly
improved by working
with an analogue
of the Thouless-Anderson-Palmer (TAP)
free entropy, 
and we believe it is an interesting
avenue of future work to adapt
such ideas to our model. 
We defer further discussion to Appendix~\ref{sec:estimators}.

\subsubsection{Regime C: Blind Source Separation ($m=0$)}
\label{sec:moment_estimator_main}

In the most difficult regime where $m=0$, 
model~\eqref{eq:our_model} 
is invariant to relabeling the mixture
components. Remarkably, even in this
regime we are able to estimate 
$c$ up to reordering its elements,
thus allowing us
to identify 
error patterns without any prior knowledge
of  how different errors affects the measurement outcome probabilities.

Our strategy in this setting is to leverage the fact
that the first $k$  moments
$m_j(c) = \sum_{i=1}^k c_i^j$ 
uniquely characterize $c$ up to ordering~\citep{hundrieser2025}. 
Indeed, this characterization is a consequence
of {\it Newton's identities}, which
assert that the coefficients
of the polynomial $f(z) =\prod_{i=1}^k (z-c_i)$
can be written solely in terms of 
$m_1(c), \dots, m_k(c)$ (cf.~Appendix~\ref{app:elementary_symmetric_polynomials}). 
The concetration of
high-dimensional Porter-Thomas
distributions allows us to
identify and estimate the moments~$m_j(c)$
directly from the bitstring data. These moment estimators can, in turn, be used
to construct an estimator $\hat f$ of the $k$-degree polynomial $f$.  
Our estimator of $c$, denoted $\hat c^{\mathrm{mom}}$, is then given
by the   collection of $k$ roots of $\hat f$.
We defer a rigorous description of this estimator
to Appendix~\ref{app:moment}.

\subsection{Sample Complexity of Error Learning}
\label{sec:main_results}

We now state our main results regarding the sample complexity of error estimation under model~\eqref{eq:our_model}. 
It will be convenient to state our
sample complexity bounds in terms
of the {\it minimax estimation risk}, a standard
statistical benchmark for quantifying the best possible error
that can be achieved by a statistical estimator uniformly over the space
$\Delta_k$.  
Concretely, the minimax risk $\calM(n,d,k,m)$ is defined as the smallest achievable upper bound epsilon for the average $\ell_2$ distance between the estimated values $\hat{c}$ and the worst case true value c:
$\max_c \mathbb{E} \| \hat{c} - c \|_2 < \epsilon$.
Here, the averaging is taken over randomness of the measured samples from an experiment, reference computers, as well as the Porter Thomas distributions $\Pi$ (arising from random circuit choices).

Our results
are stated under two conditions. 
First, we impose the following assumptions on the problem
parameters $n,d,k,m$. 
\begin{enumerate}[leftmargin=0.8cm,listparindent=-\leftmargin,label=\textbf{(S)}]   
\item  \label{assm:sample_size} 
Let $\rho = \min\{m/d,1\}$. Then, there exists an arbitrarily small constant $\gamma > 0$
such that the following assertions hold. 
\begin{enumerate}[left=-0.3cm]
\item[(i)] $n^{1+\gamma} \leq d$. 
\item[(ii)] Either $nm \leq d$
or $nm > d^{1+\gamma}$.  
\item[(iii)] Either $k \leq \sqrt{n\rho}$, or $k > (\sqrt{n\rho})^{1+\gamma}$.
\item[(iv)] Either $k\leq d \leq m$, or $d > m^{1+\gamma}$ and $k^{1+\gamma} < \frac d m$. 
\end{enumerate}
\end{enumerate}
Condition~(i) requires the sample size to be smaller than the Hilbert space dimension $d$, 
which is the most practical regime for RCS experiments. 
Once the sample size exceeds $d$,
a number of different approaches
based on (approximate)
quantum state tomography become available~\citep{wright2016learn}. 
Conditions (ii) and (iii) are mild assumptions made   for ease of exposition; they preclude the problem parameters from falling in narrow regimes where logarithmic corrections appear in our sample complexity bounds, which we do not bother to characterize sharply.
Condition (iv) is not needed for our upper bounds, but is used in our lower bounds; 
this condition allows the number of errors $k$ to be on the same order as $d$
when $m\geq d$, but somewhat limits the magnitude of $k$ when $m$ is smaller than $d$. 

Second, as discussed in Section~\ref{sec:rcs}, 
we assume the random unitary circuit
is sufficiently
deep for the following Porter-Thomas
assumption to be met. 
\begin{enumerate}[leftmargin=0.8cm,listparindent=-\leftmargin,label=\textbf{(PT)}]   
\item  \label{assm:pt} 
The random matrix $\Pi \in \bbR^{k\times d}$
has  mutually
independent rows $\Pi_{i\cdot}$ which
follow the flat Dirichlet distribution on 
the $(d-1)$-dimensional simplex. 
That is, for each $i=1,\dots,k$, one can
write
$\Pi_{i\cdot} = (X_{i1},\dots,X_{id}) / \sum_j X_{ij},$
where the random variables
$X_{ij}$ are independent, and follow a 
Porter-Thomas distribution:
 $$\bbP(X_{ij} > x) = e^{-dx},\quad \text{for all } x \geq 0.$$ 
\end{enumerate}
Although 
Assumption~\ref{assm:pt} is needed for much of
our theory, it is not needed for several of  
our estimators, as we explore in Appendix~\ref{sec:estimators}.

 \begin{table}[b]
\centering
\begin{tikzpicture}[scale=2.5]
 
  \def\xA{0}    
  \def\xB{0.8}  
  \def\xC{1.6} 
  \def\xD{2.4}   
  \def\yTop{1.2}   
  \def\yTt{0.6}    
  \def\yTo{0}    
  \def\yo{0}       
  \def\yB{-0.7}   
 
  \draw[->] (\xA,\yo) -- (\xD+0.02,\yo) node[below right] {\hspace{-0.4in}Side info. $m$};
  \draw[->] (\xA,\yo) -- (\xA,\yTop+0.2) node[above right] {\hspace{-0.25in} Number of Errors $k$};
  \node[below left] at (\xA,\yo) {$0$};
 
  \foreach \xx/\lab in {\xB/{$d^{1/k}$},\xC/{$d$}}{
    \draw (\xx,0.06) -- (\xx,-0.06);
    \node[below left] at (\xx,\yo) {\lab};
  }; 
  \draw (0.06,0.6) -- (-0.06,0.6);
  \node[left=6pt] at (\xA,0.6) {$\displaystyle \epsilon^{-2}$};

  \draw (0.06,1.2) -- (-0.06,1.2);
  \node[left=6pt] at (\xA,1.2) {$d$};
 
  \draw[dotted] (\xA,\yTt) -- (\xD,\yTt);
  \draw[dotted] (\xA,\yTop) -- (\xD,\yTop);
  \draw[dotted] (\xC,\yo) -- (\xC,\yTop);
 
  \draw[dotted] (\xA,\yTt) rectangle (\xC,\yTop);
  \draw[dotted] (\xC,\yTt) rectangle (\xD,\yTop);
  \draw[dotted] (\xA,\yTo) rectangle (\xC,\yTt);
  \draw[dotted] (\xC,\yTo) rectangle (\xD,\yTt);

  \node at ({(\xA+\xC)/2},{(\yTt+\yTop)/2}){$\displaystyle\frac{d\log k}{m\epsilon^4}$};
  \node at ({(\xC+\xD)/2},{(\yTt+\yTop)/2}){$\displaystyle\frac{\log k}{\epsilon^4}$};
  \node at ({(\xA+\xC)/2},{(\yTo+\yTt)/2}){$\displaystyle\frac{dk}{m\epsilon^2}$};
  \node at ({(\xC+\xD)/2},{(\yTo+\yTt)/2}){$\displaystyle\frac{k}{\epsilon^2}$}; 
  
  \draw[dotted] (\xA,\yB-0.1) rectangle (\xB,\yo-0.25);
  \draw[dotted] (\xB,\yB-0.1) rectangle (\xC,\yo-0.25);
  \draw[dotted] (\xC,\yB-0.1) rectangle (\xD,\yo-0.25);

  \node at ({(\xA+\xB-0.15)/2},{(\yB-0.3+\yo)/2}) {\Large ~~$\frac{C_kd^{1- 1 /k}}{\epsilon^2}$};
  \node at ({(\xB+\xC-0.15)/2},{(\yB-0.35+\yo)/2}) {\Large ~~$\frac{C_kd}{m\epsilon^2}$};
  \node at ({(\xC+\xD-0.15)/2},{(\yB-0.35+\yo)/2}) {\Large ~~$\frac {C_k} {\epsilon^2}$};
 
  \def\yTicks{-0.25}  
  \draw[dotted] (\xA,\yTicks) -- (\xD,\yTicks);
  
\node[align=center, rotate=270] 
    at ({\xD + 0.15}, {(\yTt+\yTop)/2 - 0.28}) 
    {Labeled Errors};

 \node[align=center, rotate=270] 
    at ({\xD + 0.15}, {(\yB+\yo)/2 -0.19}) 
    {Unlabeled \\Errors};
\end{tikzpicture}

\caption{ 
Heuristic summary of the sample complexities for 
noise learning as a function of number of errors 
$k$ and amount of side information $m$
(Theorems~\ref{thm:main_unsorted}--\ref{thm:main_sorted}). 
The sample complexity for labeled errors
(top) represents the 
smallest sample size $n$ for which 
the vector $c$ can be estimated to accuracy
$\epsilon$ under the $\ell_2$ norm, whereas
the sample complexity for unlabeled
errors (bottom) is measured 
for the unordered collection of error rates. 
$C_k$ denotes a generic constant depending on $k$.
}
\label{tab:sample_complexities}
\end{table}

In what follows, for any two nonnegative-valued functions $f,g$, 
we write $f(x) \asymp_a g(x)$ if there exist 
 constants $C_1,C_2 > 0$, possibly depending on a quantity $a$, such that $C_1 f(x) \leq g(x) \leq C_2 f(x)$
for all $x$.  Our first main result is stated as follows. 
\begin{theorem}
\label{thm:main_unsorted}
Under conditions~\ref{assm:pt} and~\ref{assm:sample_size},  we have 
\begin{equation*}
\calM(n,d,k,m) \asymp_\gamma 
\min\left\{ \left(\frac{k}{n \rho}\right)^{\frac 1 2}, \left(\frac{\log k}{n \rho}\right)^{\frac 1 4},1\right\},
\end{equation*}
where $\rho = \min\{m/d,1\}$.
\end{theorem}

Theorem~\ref{thm:main_unsorted}
reveals several distinct regimes
in the  sample complexity
of error learning. 
When $k$ is held fixed, and the amount of side information $m$
exceeds $d$, the sample
complexity scales as $\epsilon^{-2}$, which
is dimension-independent and 
coincides with the rate of decay
of the traditional central limit
theorem. The same sample complexity
is achievable for estimating the XEB
under the white noise
model~\eqref{eq:white_noise_model}~\citep{arute2019quantum}, and is enabled
by the fact that the matrix $\Pi$
can be accurately estimated
when $m$ is so large. 
On the other hand, when
$m < d$, although the matrix $\Pi$ is not consistently estimable, the parameter vector $c$ can be consistently estimated so long as the {\it product} $nm$ exceeds the Hilbert space dimension $d$.
As described above, this product
scaling can be interpreted 
via the birthday paradox.

Theorem~\ref{thm:main_unsorted}
also sharply characterizes the dependence of the sample complexity on the number of errors $k$. When $k$ is large, we find that the sample complexity merely degrades logarithmically in $k$, at the price of a quadratically slower dependence on the accuracy parameter $\epsilon$. In particular, Theorem~\ref{thm:main_unsorted} implies that when $nm> d \log k$,
one can   estimate
a number of errors which
is comparable to the Hilbert
space dimension.

\begin{figure*}[tb!]
    \centering
\includegraphics[width=0.99\linewidth]{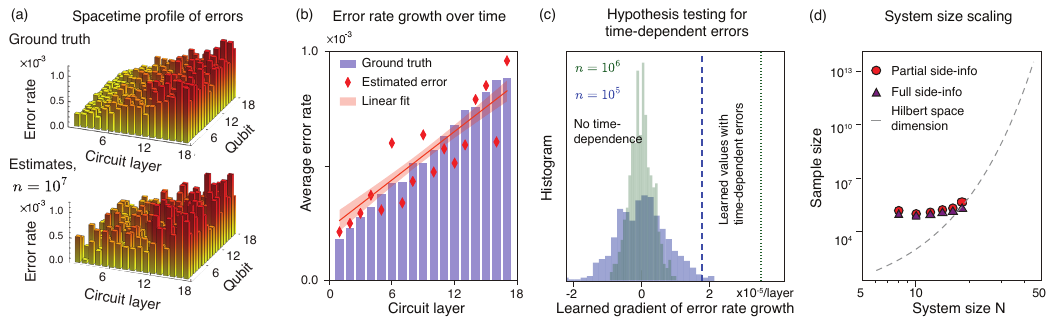}
\caption{Learning time-dependent error rates. With synthetic data, we demonstrate the use of our protocol to learn about errors that grow over time. (a)(b)(c) We simulate a $N=18$ one-dimensional brickwork random circuit subject to single-site $X,Y,Z$ Pauli errors, whose error rates (per qubit, per layer) grow from $2.5\times 10^{-4}$ in the first layer to $10^{-3}$ in the last layer. In order to simulate utility-scale circuits of different system sizes, we also set circuit depth equal to the system size while reducing the per-layer error rates such that the total fidelity is fixed at $F\approx 0.5$. (a) Upper panel: the ground truth values of the error rates at each spacetime location. Lower panel: estimated error rates with $n=10^7$ RCS samples and perfect side information (i.e., $m=\infty$). (b) By averaging over qubits, only $10^6$ samples are required to learn the increasing rate with high precision. The blue bars indicate the ground truth; red diamonds mark the estimated error rate, and the shaded red line indicates an extracted rate of error growth (a linear fit to the red diamonds). A non-zero linear fit gradient indicates increasing error rates. (c) 
Model validation between time-dependent and time-independent error models. To ensure that the learned time-dependence is statistically significant, we compare the extracted gradient (vertical dashed line) against the distribution
of gradients learned under the null hypothesis of time-independent error rates, obtained
via parametric bootstrap (Appendix~\ref{appendix:numerical}). The histogram of such gradients provides a confidence interval and $p$-values for time-dependent errors: $10^5$ and $10^6$ samples (indicated in green and blue respectively) are sufficient to learn the error rate growth with statistical significance. 
(d) System-size dependence of the sample complexity for model validation. As the system size increases, although the Hilbert space dimension increases exponentially (dashed line), the required sample size for model validation grows only polynomially with system size (orange circles). 
This sample complexity is defined as the number of RCS samples required to discriminate between a fixed rate of error growth and no error growth with $5\sigma$ significance. 
With increasing system size, classical simulation will not be feasible. In addition, we simulate the case of incomplete side-information (purple triangles) where $m=n$, i.e.~the number of side-information samples (per error component) is the same as the number of RCS samples. The sample complexity does not differ significantly between the two cases. 
}
\label{fig:timedep_error}
\end{figure*}

Theorem~\ref{thm:main_unsorted}
also indicates that the minimax risk
approaches a nondecreasing rate
of convergence when the amount of
side information approaches zero. 
This is perhaps unsurprising since 
the parameters $c_i$ are only uniquely defined
up to ordering in the absence of side information. 
Remarkably, however, our next result shows
that it is still possible to estimate the {\it unlabeled} entries of $c$ with a number of samples
that scales sublinearly in $d$. 
In what follows, we 
denote by
$\calM_<(n,d,k,m)$
the {\it unlabeled} minimax risk, 
namely the smallest real number
$\epsilon \in (0,1)$ 
for which there exists an estimator $\hat c$
such that for any $c \in \Delta$, one has
$$\min_{\sigma \in \calS_k} \left(\sum_{i=1}^k |\hat c_{\sigma(i)} - c_i|^2\right)^{\frac 1 2}\leq \epsilon,$$
where $\calS_k$ is the set of permutations on $[k]$.  
The following result
sharply characterizes the unlabeled minimax
risk when $k$ is held fixed.
\begin{theorem}
\label{thm:main_sorted}
Under conditions~\ref{assm:pt} and~\ref{assm:sample_size}, we have
\begin{equation*}
\calM_<(n,d,k,m) \asymp_{k,\gamma} ~
\frac 1 {\sqrt n} \cdot \begin{cases}
\sqrt{d^{1-\frac 1 k}}, & 0\leq m < d^{1/k}, \\
\sqrt{d/m}, & d^{1/k} \leq m < d, \\
1, & d \leq m < \infty.
\end{cases} 
\end{equation*}
\end{theorem}
Theorem~\ref{thm:main_sorted}
shows that, even 
in the absence of side
information, the ordered vector
$c$ can be consistently
recovered when $n \geq d^{\frac{k-1}{k}}$.
Although this rate degrades 
exponentially in the system size, 
its exponent is sublinear, contrary
to tomographic methods which typically
suffer from superlinear 
exponents for recovery of the full
underlying quantum state $\rho$~\cite{wright2016learn}. 
This gap can make a significant
difference in practice; for instance, 
if one adopts a two-component model
with $k=2$, akin to the white noise model~\eqref{eq:white_noise_model}, 
then Theorem~\ref{thm:main_sorted}
shows that the fidelity can be recovered
with only $\sqrt d$ samples, without
any information about
the bitstring distributions $\pi_1$ and $\pi_2$. 
This observation is consistent with the past work of~\cite{shaw2024universal}, which indicated that the second moment of $c$ can be recovered with $\sqrt d$ samples.

It would be natural to expect that 
any amount of side information $m$
would improve the sample complexity beyond the $m=0$ case, 
however Theorem~\ref{thm:main_sorted}
surprisingly shows that this
is not the case: the sample complexity
remains constant for all $m\leq d^{1/k}$.
Beyond this point, however, a phase
transition occurs, and 
the sample complexity improves
{\it linearly} with the amount of side information, scaling analogously
as in the case of Theorem~\ref{thm:main_unsorted}.

In the regime $m < d^{1/k}$, 
the lower bound 
of Theorem~\ref{thm:main_sorted}
is achieved by error vectors
$c$
which
are close to being uniform. Remarkably, 
it turns out that faster rates
of convergence are achievable 
when some of the entries of $c$ are
separated from each other. 
We make this fact precise in Appendix~\ref{app:moment},
where we show that   the error of estimating $c$   
improves as a function of the separation between
its entries. We  highlight here an implication of this
result for fidelity estimation. 
In what follows, we denote
by $c_{(1)} \geq \dots \geq c_{(k)}$ 
the sorted entries of $c$, and we interpret
$F:=c_{(1)}$ as the fidelity. 
\begin{proposition}\label{prop:main_fidelity}
Let conditions~\ref{assm:pt}
and~\ref{assm:sample_size} hold with $m=0$,
and fix $\delta > 0$.
Then, there exists an estimator 
$\hat F$ such that for any 
$c \in \Delta_k$ with $c_{(1)} > c_{(2)} + \delta$, we have
$$\bbE_c\big|\hat F - F\big| \leq C
\sqrt{\frac{d^{k-1}}{n^k}},$$
for a constant $C > 0$ depending only
on $\delta,k,\gamma$.
\end{proposition}
This result is achieved by taking $\hat F$ to be
the largest entry of the moment
estimator described in Section~\ref{sec:moment_estimator_main},
and does not rely on knowledge of $\delta$.
This highlights an important property
of the moment estimator: it can estimate the
fidelity more accurately than the other entries of $c$.
This estimator does so \textit{adaptively}, without requiring assumptions of $c$ or modification of the algorithm itself.
To see this, if we heuristically
set $\gamma =0 $ and take $n$ to be on the same 
order as $d$ (i.e.~where estimation in regime C is feasible), then, absent any side information,
the whole vector $c$
is estimable at the rate $n^{-1/2k}$---which
degrades exponentially in $k$---whereas
its largest entry is estimable at the 
faster rate $n^{-1/2}$, whenever it is $\delta$-separated
from the remaining entries. While the required exponential samples with system size currently limits this to a theoretical result, it hints at the possibility of practical estimators with similar properties.

\section{Simulation Study}
\label{sec:numerical}
To demonstrate the utility of our methods, we analyze synthetic data in two distinct scenarios. 
On the one hand, we consider Regime A where $\Pi$ is classically computed, and estimate $c$ with the 
maximum likelihood estimator~\eqref{eq:main_mle}.
On the other hand, we consider Regime B where $\Pi$ is only available
through side-information, in which case we use the 
variational estimator~\eqref{eq:em_fixed_point}. 
We additionally report simulations for Regime C in Appendix~\ref{sec:estimators}.

\subsection{Learning time-dependent errors}
We first showcase the use of our technique to detect the presence of time-dependent error rates. 
Such time-dependence can exist in various quantum platforms due to distinct physical reasons, including non-Markovian noise~\cite{rower2023evolution}, burst errors~\cite{hirasaki2023detection}, and atomic heating in an optical tweezer~\cite{de2018analysis}.
To this end, we numerically simulate a depth-16 circuit of a one-dimensional $N=18$ qubit chain. We perform a circuit-level noise simulation with a random quantum circuit. At every layer and every qubit, we inject Pauli $X,Y$ and $Z$ single qubit errors with space- and time-dependent probabilities $c$, corresponding to single-qubit Pauli channels with varying rates. 

In this simulation, we set the average single-qubit error rates to grow by a factor of 4 over the course of the entire circuit (see Appendix~\ref{appendix:numerical}). 
Our estimators successfully extract the individual time-dependent error rates (Fig.~\ref{fig:timedep_error}a) with $10^7$ samples, within the ability of the existing state-of-the-art quantum platforms. 
Since our aim is to study whether the error rate increases over time, we also perform a statistical
test for the null hypothesis that the error rate is constant across layers,  which in principle should require fewer samples. This indeed turns out to be the case: the null hypothesis can be rejected with approximately $10^5$ samples at level 0.95 (Fig.~\ref{fig:timedep_error}c), and with overwhelming significance when the sample size is of order $10^6$. 

For this system size, $10^5$ samples is comparable to the Hilbert space dimension $2^{18} = 262~144$. However, we find that the number of samples required for this hypothesis test grows slowly with system size, and we expect it to be far below the Hilbert space dimension for systems of sizes $N>20$ (Fig.~\ref{fig:timedep_error}d).

\begin{figure}[btp!]
    \centering
\includegraphics[width=0.99\linewidth]{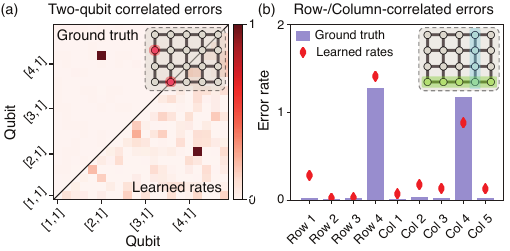}
\caption{Reconstructing correlated errors from synthetic data simulated for a $4\times 5$, depth-5 circuit. We consider two types of correlated error: (a) two-qubit $XX$ errors, or (b) multi-qubit $XX\cdots X$ errors along one row or one column to simulate errors induced along a shared control line. In both settings, we also include time-independent single-site Pauli errors at every qubit with a rate $2\times 10^{-3}$, chosen such that the total many-body fidelity is $F\approx 0.5$. (a) We learn the rates (averaged over layers) of two-qubit errors $X_u X_v$ for all pairs of qubits simultaneously and represent them on a 2D plot: specifically the correlated error rates $c_{u.v}-c_{u}c_{v}$, which subtracts the expected two-qubit error rates from independent single-qubit errors on qubits $u$ and $v$. We refer to this difference as the correlated error rates. Upper left half: ground truth: two qubits, highlighted in the inset, experience correlated $XX$ errors at a rate of $10^{-3}$ per layer. Lower right half: extracted error rates from $10^7$ samples correctly identify the correlated pair. 
(b) We also learn the rates of correlated errors on all the qubits in the same row or column. 
Blue bars: ground truth where one row and one column (inset) experience correlated errors. Red diamonds: extracted error rates. Again, $10^7$ samples are sufficient to reliably learn about correlated errors. 
}
\label{fig:correlated_error}
\end{figure}

\label{sec:real_data}
\begin{figure*}[!tb]
    \centering
    \includegraphics[width=\linewidth]{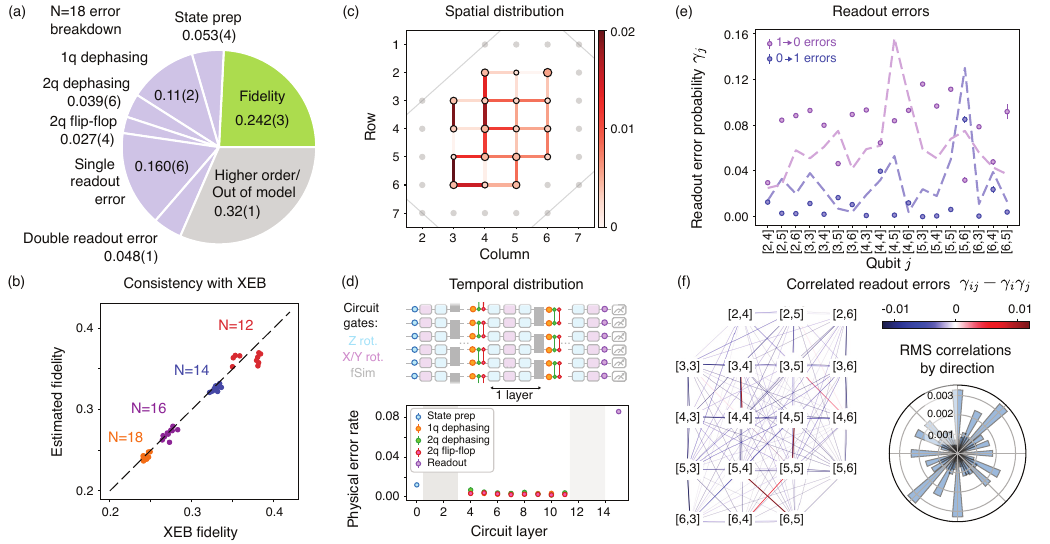}
\caption{Analysis of experimental RCS data. We apply our methods to study the publicly-available RCS data from Ref.~\cite{arute2019quantum}, results here shown for $N=18$. Using the MLE, we resolve different types of errors in many spacetime locations. We simulate state-preparation, single-qubit (1q) dephasing errors, two-qubit (2q) gate dephasing and flip-flop errors, and single and double readout errors (full details in App.~\ref{app:google}), for a total of $k=461$ total errors. (a) We summarize the combined contributions of each error type. Quoted values and error bars indicate the sample mean and its standard error over 10 random circuits. Modeled errors account for 68\% of the total weight: a remaining 32\% weight is fitted to the \textit{white noise} term representing errors outside our model such as multiple errors, consistent with expectations for this fidelity value (App.~\ref{app:google}). Note that the rate of 1q dephasing errors we learn here are the rates of errors that can be \textit{described} as single-qubit $Z_j$ operators: these errors may also arise from two-qubit \textit{gates}, and hence the proportion of 1q and 2q errors here are comparable, even though we expect them to primarily arise from two-qubit gates.
(b) Converting the results of our benchmarking report into a many-body fidelity (App.~\ref{app:converting_rates}) yields results in close quantitative agreement with the XEB fidelity.
(c) Learned error rates show considerable variation among qubits, in a consistent fashion over random circuit realizations.
We plot the total rates of the 2q dephasing and flip-flop errors on nearest neighbors, indicated by the color of the red links. We also plot the single-qubit dephasing error rates, indicated by the size and color of each qubit. Qubits are arranged according to their physical layout on the device (borders and unused qubits for the $N=18$ dataset in gray). The magnitude of learned error rates is consistent between system size, random circuit agreement, and their sum over error channels is consistent with Ref.~\cite{arute2019quantum} (main text).
(d) Our procedure also yields time-resolved error rates,
revealing approximately time-independent errors in the middle of the circuit. $1\rightarrow 0$ readout errors were found to be the largest type of error. Above, we depict the positioning of our modeled errors in the circuit: in the ideal circuit, a single ``layer" consists of four gates applied to each qubit, and we insert errors at layers in the circuit. Errors inserted near the start and end of the circuit have unusual properties, and we omit errors in the first and last three layers (gray regions) to avoid additional complications (see App.~\ref{app:converting_rates}).
(e) We explicitly compare our estimated rates (points) of readout errors with those reported in Ref.~\cite{arute2019quantum} (dashed lines). The average rates of readout errors are quantitatively similar, with deviations on certain qubits: these may arise from the fact that only a subset of qubits are simultaneously measured, which may hence experience error rates different from when all qubits are simultaneously read out (Fig. S24 of Ref.~\cite{arute2019quantum}). Error bars indicate standard error over 10 random circuits.
(f) Learning correlated readout errors: We estimate the physical error rate $\hat{\gamma}_{ij}$ of double readout errors on qubits $i$ and $j$, and compare it to the rates of independent errors $\hat{\gamma}_i,\hat{\gamma}_j$: the difference $\hat{\gamma}_{ij} - \hat{\gamma}_i \hat{\gamma}_j$ estimates the rates of correlated readout errors. We indicate these correlations with the thickness and colors of lines between all pairs of qubits $i$ and $j$. These correlations can be as large as a $1\%$ rate, although typical values are closer to $0.2\%$. 
We see correlations between many pairs of qubits, with stronger correlations (surprisingly \textit{negative}) between nearest neighbors as well as along the diagonals. We summarize these with a polar plot of the root-mean-squared (RMS) correlations averaged along each direction. Note that this is an \textit{average} over qubit pairs with a given orientation, ignoring their separation, and not simply a sum, which would weight certain directions over others because of the different number of qubit pairs for each orientation.
}
    \label{fig:google_data}
\end{figure*}
\subsection{Learning correlated errors}
We also demonstrate the detection of weak correlated errors, such as two-qubit correlated $X$ errors (Fig.~\ref{fig:correlated_error}a) and multi-qubit $XX\cdots X$ errors along a row or column of qubits in a 2D geometry (Fig.~\ref{fig:correlated_error}b). 
Such errors may occur, for example, due to qubit crosstalk~\cite{abrams2019methods,barrett2023learning} or control-line or readout multiplexing~\cite{heinsoo2018rapid} and are fundamentally inaccessible to calibration experiments on isolated subsets of qubits. 
Despite these errors being weak (respectively $0.2\%$ and $0.1\%$), $10^7$ samples suffice to detect them. 
Note also that our method does not depend on geometric locality, and is able to detect correlations between spatially separated qubits. 

\section{Analysis of experimental RCS data}
\label{sec:google}
Finally, we apply our method to the Google Quantum AI random circuit sampling (RCS) data from Ref.~\cite{arute2019quantum}. In this experiment, random quantum circuits with sizes ranging from $N=12$ to $N=53$ were executed on a 2D grid. The dataset, which is publicly available, contains
ten random circuit realizations per system size $N$ and 500,000 measurement outcomes per circuit.

We perform exact simulations of random circuits up to system size $N=18$, incorporating a variety of error mechanisms at each spacetime location. We specifically consider several sources of error: state-preparation errors, single-qubit dephasing errors, two-qubit controlled-Z dephasing and flip-flop ($|01\rangle\langle10| + |10\rangle\langle01|$) errors which may arise due to dressing by higher transmon levels, and readout errors which may be biased, i.e.~have unequal error rates between $1\rightarrow0$ and $0\rightarrow 1$ processes~\cite{arute2019quantum}, for a total of $k=461$ distinct errors.
Errors near the beginning and the end of the circuit are correlated not only with other errors, but with the ideal state. Therefore to simplify our first analysis, we only simulate errors beyond the first and last three layers, i.e.~the middle 8 circuit layers, see Appendix~\ref{app:google} for details.

Using our estimator~\eqref{eq:main_mle} on publicly-available RCS data from Ref.~\cite{arute2019quantum}, we extract physical error rates associated with each of the above sources across spacetime locations in the circuit. The results are summarized in Fig.~\ref{fig:google_data}: the data is highly rich and can be examined along multiple axes, including its behavior over space, time, and its magnitudes resolved by error type and location. Since $1\rightarrow 0$ readout errors are a dominant process and their calibration values were reported in Fig. S24 of Ref.~\cite{arute2019quantum}, we compare our learned values against reported values. We find similar average values of readout error, but slightly different qubit-to-qubit values. This could be due to the fact that readout error rates differ when the qubits were individually read out as opposed to simultaneously read out~\cite{arute2019quantum}. Meanwhile, the data sets we consider lie between both extremes: approximately half the qubits are simultaneously read out.

Physically-realistic error sources introduce systematic deviations from the i.i.d.~Porter-Thomas assumption~\ref{assm:pt} in several ways. As a result, one cannot directly equate the many-body fidelity to the coefficient $c_1$. We develop a theory for converting the learned rates $c$ into physically meaningful quantities such as the many-body fidelity and the physical per-qubit, per-layer error rate, detailed in Appendix~\ref{app:converting_rates}.
We summarize the dependence of the error rates over time and over qubits: error rates exhibit inhomogeneity over qubits but remain approximately constant across time, with larger state-preparation and readout errors at the initial and final circuit layers, respectively. Our estimates are consistent across the ten random circuit realizations analyzed, with reproducible trends being observed.

The orders of magnitude of the learned error rates in the bulk of the circuit agree with reported values of two-qubit gate errors, with a physical error rate (combined over 1q dephasing, 2q dephasing and 2q flip-flop) of $0.010(2)$ per qubit per layer, c.f. the reported mean two-qubit cycle benchmarking error rate of $0.0093(4)$ (Fig.~2 of Ref.~\cite{arute2019quantum}). Note that our `1q' and `2q' error sources refer to errors that can be expressed in terms of single- and two-qubit operators, and are not able to distinguish whether these come from single or two-qubit gates: this may be addressed with more careful positioning of errors in the circuit sequence [Fig.~\ref{fig:google_data}(d)].

We also investigate correlated readout errors by learning the rates of readout errors occurring on two (potentially distant) qubits simultaneously. While we see a large number of correlations which may be due to statistical noise, we also observe a certain directional-dependence even of non-local correlated readout errors. These may arise, for example from multiplexed qubit readout lines~\cite{arute2019quantum}.

Finally, we find that our learning procedure is remarkably robust. 
Even if our $k$-component model does not specify all sources of error, the error rates for the components that are modeled have estimated values that are stable to the presence or absence of other components in the model. This is illustrated in the close agreement between the XEB fidelity and the fidelity estimated by our protocol in Fig.~\ref{fig:google_data}(b), which holds true for other quantities as well, such as the fractions of each error component in Fig.~\ref{fig:google_data}(a). To a large extent, this is because of the approximate orthogonality of the components $\pi_i$, which implies that estimators such as $\hat{c}^\text{XEB}_i$ are accurate. 
Such estimators estimate each component without reference to the other components of the model and hence are inherently robust to this type of misspecification. However, in our error model in this section, components are non-orthogonal (App.~\ref{app:fidelity_conversion}) yet our protocol retains this stability.
Indeed, such stability is the operating principle behind the XEB: knowledge of the error processes are not required to estimate the many-body fidelity. This enables refining our error models hierarchically by systematically adding error sources according to expected significance, up to a desired level of precision.

\section{Outlook}
Our work demonstrates the utility of novel data processing methods to extract 
detailed information from random unitary circuits. 
These have become
an industry standard for benchmarking
quantum devices, for which many existing datasets are publicly available.
Our methods pave the way to a more accurate understanding of the errors that quantum computers experience, which come in many forms, often unexpected. We anticipate further extensions 
of our methods not only to learn about errors, but also to sense unknown, complicated signals~\cite{gong2025robust}. We also
anticipate possible applications of our methods
to basic science experiments, to learn about the properties of possibly exotic states prepared on a quantum computer~\cite{altman2021quantum}.

From a statistical lens, the quantum information setting poses a unique set of new challenges
for statisticians: its discrete data in the form of bitstring counts differs from traditional physics experiments involving continuous-variable data, and its high dimensionality means that each individual measurement reveals little about the underlying distribution. In our setting, we also encountered a synergistic dual role of randomness in the quantum circuit: random quantum circuits are hard to classically simulate and serves as a task that separates classical from quantum computers. However, they also have many typical properties which 
we exploit, and which led us to study a 
new family of high-dimensional
latent variable models.

We anticipate that our methods developed in the setting of partial side information generalize to cases when the reference quantum computer is noisy. 
In the meantime, however, our methods are still applicable when clean side information samples can be obtained using quantum error detection~\citep{bluvstein2024logical} or error correction methods~\cite{google2023suppressing,google2025quantum}.

Our results point to the broader relevance of high-dimensional statistical methods in quantum computing. The data produced by quantum devices are inherently high-dimensional, and the number of accessible samples is often far smaller than the dimension of the underlying Hilbert space. This imbalance makes quantum data analysis a natural arena for ideas from high-dimensional inference. We anticipate not only the fruitful application of existing tools---such as those introduced in early pioneering work on compressed-sensing-based quantum state tomography~\cite{flammia2012quantum}---but also the development of new statistical frameworks tailored to the distinct structure and constraints of quantum computing.

\section*{Acknowledgements}
TM would like to thank Florentina Bunea and Marten Wegkamp
for discussions related to
this work, and  for
bringing his attention to 
Ref.~\citep{bing2022}. 
We thank Trond Andersen, Dmitry Abanin, Bryce Kobrin, Elizabeth Bennewitz, Nikita Astrakhantsev, Weijie Wu, Kostyantyn Kechedzhi, Dvir Kafri and Joonhee Choi for insightful discussions.
We acknowledge support by 
the NSF QLCI Award OMA-2016245, 
the Center for Ultracold Atoms, an NSF Physics Frontiers Center (NSF Grant PHY-1734011), 
and the
NSF CAREER award 2237244. 
TM gratefully acknowledges the support of a Norbert Wiener fellowship.
WG is supported by the Hertz Foundation Fellowship.
 
\bibliography{manuscript_rcs_arxiv}

\clearpage
\newpage
\onecolumngrid

\appendix

\setcounter{figure}{0}
\renewcommand{\figurename}{Fig.}
\renewcommand{\thefigure}{S\arabic{figure}}
\setcounter{table}{0}
\renewcommand{\tablename}{Table}
\renewcommand{\thetable}{S\arabic{table}}

\makeatletter
\long\def\@footnotetext#1{\insert\footins{\reset@font\footnotesize
  \interlinepenalty\interfootnotelinepenalty
  \splittopskip\footnotesep
  \splitmaxdepth \dp\strutbox \floatingpenalty \@MM
  \hsize\columnwidth \@parboxrestore
  \protected@edef\@currentlabel{%
   \csname p@footnote\endcsname\@thefnmark
  }%
  \color@begingroup
    \@makefntext{%
      \rule\z@\footnotesep\ignorespaces#1\@finalstrut\strutbox}%
  \color@endgroup}}%
\makeatother

\interfootnotelinepenalty=10000

\newgeometry{left=1.5in, right=1.5in, top=1.25in, bottom=1.25in}

\allowdisplaybreaks

\begin{center}
{\large\bf  SUPPLEMENTARY MATERIAL}
\end{center}

\tableofcontents

\resumetoc

\section{Notation} 
\label{sec:notation}
Throughout the  
manuscript, we adopt
the following notation. 
\begin{itemize} 
\item 
Given a vector $x\in \bbR^d$, 
$\|x\|_p$ denotes the $\ell_p$ 
norm of $x$, and $\|x\|$ denotes
its $\ell_2$ norm 
when no subscript is specified.
\item Given $a,b \in \bbR$, we write $a\wedge b=\min\{a,b\}$,
$a\vee b = \max\{a,b\}$, and $(a)_+ = \max\{0,a\}$. 
\item 
For any matrix $\Pi$, 
$\|\Pi\|$ denotes its Frobenius norm, and  $\Pi_{\cdot j}$ and $\Pi_{i\cdot}$ denote its $j$-th column and $i$-th row, respectively.
\item 
For any two vectors $c,c'\in \bbC^k$, we write
$$W(c,c')  = \inf_{\sigma\in \calS_k} \sum_{i=1}^k |c_{\sigma(i)} - c_{i}'|,$$
where the infimum is over the permutation group $\calS_k$ on $[k]$.
\item 
$\Delta_k$ denotes the $(k-1)$-dimensional simplex, namely the set of vectors $x \in \bbR^k$
with nonnegative entries such that $\|x\|_1=1$.
Furthermore, given $1 \leq k_0 \leq k$, 
$\Delta_{k,k_0}$ denotes
the set of all elements $c \in \Delta_k$
which have exactly $k_0$ distinct entries.
\item 
Let $\bbN= \{1,2,\dots\}$
and $\bbN_0 = \{0,1,\dots\}$. 
For  integers $m,r \in \bbN$, we denote the $r$-th falling factorial of $m$ by 
$(m)_r = m(m-1)\dots(m-r+1)$.
\item 
Given a random variable $X$, its $r$-th moment and cumulant (when they exist) are denoted by 
$m_r(X)$ and $\kappa_r(X)$ respectively. By abuse of notation,  also abbreviate these quantities by 
$m_r(f)$ and $\kappa_r(f)$ when $f$ is the distribution or density of $X$. 
We abuse notation by writing $m_\alpha(x) =  \sum_{i=1}^k x_i^\alpha$ for any $x \in \bbR^k$.  

\item 
We denote by $I(A)$ the indicator function of a set
$A$, and by $\delta_{i,j} = I(i=j)$ the Kronecker delta
function.

\item 
We denote by $\calE_{d}$ the 
Porter-Thomas
or exponential distribution with parameter $d$, and we use the 
abbreviation $\calE_d^k\equiv \calE_d^{\otimes k}$ for its $k$-fold product distribution.
That is, 
$$\mathrm{d} \calE_{d}^k(\varpi) := d^k \exp\left(d\|\varpi\|_1\right) \mathrm{d}\varpi,\quad \varpi\in\bbR_+^k.$$
We also abbreviate
by $\calD_d$ the flat Dirichlet (i.e. uniform) distribution
over $\Delta_d$. 
We   abbreviate by 
$\mathrm{Mult}(n;p_1,\dots,p_d)$
the multinomial distribution 
with $n$ trials, $d$ categories, and success
probabilities $(p_1,\dots,p_d)\in \Delta_d$. 
 Furthermore, $\mathrm{Poi}(\lambda)$
 denotes the Poisson distribution with
 intensity parameter $\lambda > 0$,
 and $\mathrm{Gamma}(\alpha,\lambda)$
 denoted the Gamma distribution
 with shape parameter
 $\alpha >0$ and rate parameter
 $\lambda > 0$ (its mean is $\alpha/\lambda$).

\item 
For a random variable $V$ and  $\alpha \in (0,2)$,
we define the Orlicz
norm of $V$ by
\begin{equation} 
\label{eq:orlicz}
\|V\|_{\psi_{\alpha}} = \inf\big\{ \eta > 0 : \bbE\big[e^{(|V| / \eta)^{\alpha}} \big] \leq 2\big\}.
\end{equation} 

 \item 
Given two sequences of nonnegative real
numbers $(a_n)_{n=1}^\infty$ and $(b_n)_{n=1}^\infty$, 
we write $a_n \lesssim b_n$ if there exists
a constant $C > 0$ such that $a_n \leq C b_n$
for all $n \geq 1$, and we write $a_n \asymp b_n$
if $a_n \lesssim b_n \lesssim a_n$. Throughout
the manuscript, the constant
$C$ may always
depend on the parameter
$\gamma$ arising in condition~\ref{assm:sample_size}.
In some cases, when it is clear
from context, the constant $C$ may
also depend on $k$, or other problem
parameters.

\item 
We define the hard- and soft-thresholding
functions, with a parameter $\lambda > 0$, for all $x \in \bbR^d$, 
by
\begin{align} \label{eq:thresholding_functions}
\calH_\lambda(x) = 
\begin{cases} 
x, & |x| \geq \lambda , \\
0, & |x| \leq \lambda,
\end{cases}
\qquad 
\calS_\lambda(x) = 
\mathrm{sign}(x) \cdot \max\{|x|-\lambda,0\}.
\end{align}

\item 
Given two probability measures
$P$ and $Q$, admitting densities $f$
and $g$ with respect to a 
$\sigma$-finite dominating measure
$\nu$, 
we  make use of the standard statistical
divergences: 
the total variation distance $\mathrm{TV}(P,Q) =\frac 1 2 \int |f-g|d\nu$, 
the Hellinger distance $H^2(P,Q) = \int (\sqrt f - \sqrt g)^2d\nu$, the Kullback-Leibler divergence
$\mathrm{KL}(P\|Q) = \int \log(f/g) fd\nu$, and
the $\chi^2$-divergence
$\chi^2(P\|Q) = \int \frac{(f-g)^2}{g} d\nu$. 
If $J$ is a joint distribution of $P$ and $Q$, 
then its
mutual information is denoted by 
$I(X;Y) = \KL(P\otimes Q\|J)$ for $(X,Y)\sim J$.
If $X$ is discrete, then $H(X) = -\bbE[\log f(X)]$ denotes the Shannon entropy of $P$.

\end{itemize}

\section{Comparison of Estimators}
\label{sec:estimators}

In this Appendix, we expand
on our discussion from Section~\ref{sec:main_estimators},
and discuss various  estimators
for  the   parameter vector $c=(c_1,\dots,c_k)$
in Regimes~A--C.  

Let us begin by recalling our
statistical model.
Let $\Pi =(\pi_{ij}) \in \bbR^{k\times d}$
be a random matrix whose
 rows belong to
 the simplex $\Delta_d$. Under condition~\ref{assm:pt}, 
 the rows of $\Pi$ are i.i.d., and distributed
 according to the uniform distribution over $\Delta_d$, 
 also known as the flat Dirichlet law $\calD_d$.  
Although $\Pi$ is assumed
to satisfy condition~\ref{assm:pt}
for our theoretical results, 
this condition is not required
for all estimators  below.
Under this known marginal distribution
of $\Pi$, we adopt the sampling
model put forth in Section~\ref{sec:rcs}:
Given an unknown error
vector $c \in \Delta_k$, one 
draws conditionally independent observations
of the form
\begin{align*}
Z_1,\dots,Z_n \,\big|\,\Pi &\sim \sum_{j=1}^d (\Pi_{\cdot j}^\top c) \delta_{z_j}\\ 
W_{i1},\dots,W_{im} \,\big|\,\Pi &\sim \sum_{j=1}^d \pi_{ij}\delta_{z_j},\quad i=1,\dots,k,
\end{align*}
where $z_j \in \{0,1\}^N$ denotes the binary enumeration of the integer $j$,
and $d=2^N$.
Furthermore, let
$Y_j = \sum_{\ell=1}^n I(Z_\ell=z_j)$
and $V_{ij} := \sum_{\ell=1}^m I(W_{i\ell} = z_j)$ 
denote the induced
histograms, 
for all $i=1,\dots,k$ and $j=1\dots,d$. 
Recall that we consider three regimes:
\begin{enumerate}
\item {\bf Regime A} ($m=\infty$):
The user observes
$Z_1,\dots,Z_n$ and the matrix $\Pi$. 
\item {\bf Regime B} ($0 <  m < \infty$):
The user observes  
$Z_1,\dots,Z_n$ and $W_{11},\dots,W_{km}$.
\item {\bf Regime C} ($m=0$):
The user merely observes 
$Z_1,\dots,Z_n$. 
\end{enumerate} 
In the following subsections, 
we discuss several practical estimators for $c$ in these three settings,
and connect them to existing statistical literature. 
We also report simulation studies comparing the numerical performance of these estimators. Whenever
possible, we state upper bounds on their sample complexity.  
We begin with the simplest
case of Regime~A. 

\subsection{Regime A}

In the simplest case $m=\infty$
where the matrix 
$\Pi$ is known to the practitioner,
our model reduces to a
multinomial generalized linear model
with identity link function,
in the sense that the histogram 
$Y$ satisfies the relation
\begin{equation}\label{eq:multinomial_regression}
Y \sim \mathrm{Mult}(n;\Pi^\top c).
\end{equation}
Since the rows of $\Pi\in \bbR^{k\times d}$
define probability mass functions 
with support size $d$, model~\eqref{eq:multinomial_regression} 
can also be interpreted as a mixture of known multinomial distributions
with unknown mixing weights. 
This is a well-studied
model~\citep{rao1957maximum,birch1964new,fienberg,bing2022,manole2021estimating} for which the most
natural estimator is, perhaps, 
the maximum likelihood estimator,
which we describe below. 
Nevertheless, we begin by detailing how the simple
XEB estimator~\eqref{eq:main_ortho}
arises in this model. 

\subsubsection{The XEB Estimator}
Under a Poisson
approximation of the multinomial distribution---which we will justify in the next section---one can think of the 
histogram entries $Y_j$ as being approximately
independent, and distributed 
as 
\begin{align}\label{eq:looks_like_poisson}
Y_j \sim \mathrm{Poi}(n\Pi_{\cdot j}^\top c),
\quad j=1,\dots,d.
\end{align}
Under this approximation, 
one has 
\begin{align} 
Y = n\Pi^\top c+ \epsilon,
\end{align}
where, conditionally on $\Pi$, $\epsilon$  is a vector
with independent and mean-zero entries
satisfying $\Cov[\epsilon|\Pi] = \diag(n\Pi^\top c)$.
This defines a linear regression model with
heteroscedastic errors.
Under condition~\ref{assm:pt}, 
the marginal covariance of $\epsilon$
is simply $\Cov[\epsilon] = nI_d/d$, 
which is homoscedastic,
and is of the same order
of magnitude as the entries of $\Pi$.
Therefore, if one's goal
is to obtain an estimator
which performs well 
{\it unconditionally}---on
average over $\Pi$---one possible
approach to estimating $c$
is to fit the ordinary
least-squares estimator
of $c$, which ignores the heteroscedastic nature of the problem, 
\begin{align}\label{eq:ols}
\hat c^{\mathrm{OLS}} = \argmin_{x \in \bbR^d} \|Y - n\Pi^\top x\|_2^2 = (\Pi\Pi^\top)^{-1} \Pi Y / n.
\end{align}
One can also
restrict the minimization to 
the natural parameter
space $\Delta_k$ of $c$, and such a restriction
has regularization advantages which
we describe below, 
but no longer admits a closed form.
Under condition~\ref{assm:pt}
and for large values of $d$, the matrix
$\Pi\Pi^\top$ concentrates
rapidly around its mean value,
which is well-approximated by the matrix
 $$A = (I_k+\one_k\one_k^\top)/d.$$
Thus, an even simpler estimator
of $c$ is given by
\begin{align} 
\label{eq:ortho_with_A}
A^{-1} \Pi Y/n = \frac{d}{n}  \left(  \Pi Y  - \frac{\one_k^\top \Pi Y}{k+1} \one_k\right).
\end{align}
Up to centering, this approximation of the least-squares
estimator is similar
to that arising in linear regression with   orthogonal
design matrices~\citep{johnstone2019}. 
The XEB estimator~\eqref{eq:main_ortho},
\begin{equation}\label{eq:xeb_app}
\hat c^{\mathrm{XEB}} = 
(d/n) \Pi Y - \one_k
\end{equation}
can now be understood
as a variant of equation~\eqref{eq:ortho_with_A}
in which the second term is simply 
approximated by $\one_k$---an approximation
which can be justified when $d$ is large using
the fact
that the mean of $Y$ is $n\Pi^\top c$ 
with $\one_k^\top c = 1$.

As we have already indicated, the XEB
estimator can be improved when $k$ is large 
via its hard- or soft-thresholded counterparts: 
\begin{equation} \label{eq:xeb_thresh}
\hat c^{\mathrm{XHT}} = \calH_\lambda(\hat c^{\mathrm{XEB}}),\quad 
\hat c^{\mathrm{XST}} = \calS_\lambda(\hat c^{\mathrm{XEB}}).
\end{equation}
Here,  $\lambda > 0$
is a tuning parameter, and 
$\calH_\lambda,\calS_\lambda$ denote  the thresholding functions
defined in equation~\eqref{eq:thresholding_functions}. Once again,  these estimators
are similar in spirit to hard- and soft-thresholding estimators
in linear regression models
with orthogonal design matrix~\citep{donoho1994ideal,johnstone2019}. 
In our context, they satisfy the following upper bound. 
\begin{proposition}\label{prop:regime_A_ub}
Under condition~\ref{assm:pt}, there exists a universal constant $C > 0$ such that for all 
$1\leq k,n\leq d$, there exists $\lambda \geq 0$ such  that 
$$\sup_{c\in\Delta_k} \bbE_c \|\hat c^{\mathrm{XHT}} - c\|_2 \leq C \cdot \min \Big\{ (k/n)^{1/2}, (\log k / n)^{1/4}\Big\}.$$
\end{proposition}
Proposition~\ref{prop:regime_A_ub} is a special case of upper bounds for Regime B which
we will develop in the next section, thus we leave it without explicit proof. 
It is worth emphasizing that Proposition~\ref{prop:regime_A_ub} exhibits
the same convergence rate
as typically seen
in $\ell_1$-sparse linear regression problems~\cite{donoho1994minimax,raskutti2011}. 
This is perhaps surprising, 
since the heteroscedastic nature
of multinomial regression problems
can alter the minimax estimation rate, as we
shall see in our discussion of likelihood
estimators below. 
However, one of the implicit observations in  
our upper and lower bounds is the fact
that, under condition~\ref{assm:pt}, 
the heterosedasticity
of our model is sufficiently
mild for it to behave like
a homoscedastic 
model.  This   observation
is what allows us to establish the minimax
optimality of simple
estimators, like the (regularized) XEB estimator, which have the advantage
of being computable in closed-form---an important benefit for large-scale quantum computing problems
where $d$ grows exponentially with system size.
Let us emphasize that the work
of~\cite{li2018} also found
regularized least squares estimators
to be minimax-optimal in $\ell_1$-constrained
Poisson regression problems.

\subsubsection{The Maximum Likelihood Estimator}

Although the XEB estimator is minimax optimal and simple to compute,
it relies heavily on the Porter-Thomas assumption~\ref{assm:pt}.
We have observed some deviations
from this assumption
in our
real data analysis.  
An alternative estimator for Regime A which does not rely on this assumption
is the maximum likelihood estimator (MLE), defined by 
\begin{align} \label{eq:mle}
\hat c^\mathrm{MLE} =\argmax_{x \in \Delta_k} \sum_{j=1}^d Y_j \log(\Pi_{\cdot j
}^\top x).
\end{align}
Unlike the XEB estimator, this optimization problem does not enjoy a closed form, 
but  is nevertheless concave and can be computed   using standard solvers.
It also has the practical advantage
of being free of tuning parameters, 
unlike our thresholded
XEB estimators.

Several theoretical
properties of the MLE   have   been investigated
by~\citet{bing2022} in the context of topic modeling
(a framework which we discuss further in Appendix~\ref{app:moment}). One of their remarkable
findings is the fact that the MLE
can identify the {\it sparsity
pattern} of the error
vector $c$ without any explicit regularization. 
Concretely, under some conditions, they show that with probability
tending to one,  
$$\supp(\hat c^{\mathrm{MLE}}) \subseteq 
\supp(c),$$
cf.\,\cite[Theorem 5]{bing2022}.
In our context, this property implies
that
if one specifies  a conservative
number of candidate errors $k$, 
many of which
may not be present in a quantum device at hand, 
then the MLE is   unlikely
to assign a positive
error rate to any of these non-existent
errors.
Although this result relies
on  conditions which are only met
in our setting in the unrealistic case $n\gg d$,  it nevertheless
hints at an important practical property of the MLE.

\citet{bing2022} also
derive $\ell_1$ sample complexity upper bounds for the MLE. 
Staying again with the condition
$n\gg d$, a special case of their results
can be informally stated as
\begin{align} \label{eq:bing_rate}
\bbE \|\hat c^{\mathrm{MLE}} - c\|_1 \lesssim \min \left\{ \kappa^{-2} \sqrt{\frac{\rho \log k}{n}}, 
\kappa^{-1} \sqrt{\frac k n}\right\},
\end{align}
where\footnote{We state the results of~\citep{bing2022} in the special case where, in their notation, 
$\underline{J} = [d]$ and $\overline{J}=[d]$. These assumptions hold in our setting 
with high probability when $n\gg d$.}
\begin{align} \label{eq:mle_rate_bing}
\kappa = \min_{v \in \bbR^k} \frac{\|\Pi^\top  v\|_1}{\|v\|_1},\quad \rho = \max_{1 \leq j \leq d} \frac{\|\Pi_{j\cdot}\|_\infty}{\Pi_{\cdot j}^\top c}.
\end{align}
The quantity $\kappa$ is an $\ell_1$ analogue of the minimal singular value of the matrix $\Pi$.   
Under
condition~\ref{assm:pt}
and $k =o( d)$, 
a simple derivation shows that
$\kappa$  scales as  $k^{-1/2}$
up to logarithmic factors,
with high probability. 
Furthermore, $\rho$ is typically
of constant order in our setting. Therefore,  
equation~\eqref{eq:bing_rate}
shows that the MLE
achieves the $\ell_1$
convergence rate $k/\sqrt n$. 
This rate is consistent
with that of  Proposition~\ref{prop:regime_A_ub}
when translated from the $\ell_2$ to $\ell_1$
norm. 
We expect that the MLE
can also be shown to achieve
the optimal convergence
rate in the realistic regime
where $d$ is arbitrarily large---and related
results can already be deduced for instance from
the work of~\cite{raginsky2010}---but we leave
a careful analysis of this
problem to future work.

\subsubsection{Numerical Comparison in Regime A}
\label{app:regime_A_numerical}
We provide a brief numerical comparison of the five   estimators
discussed in the preceding subsections: the maximum likelihood estimator~\eqref{eq:mle}, 
the   least squares estimator~\eqref{eq:ols},
the simplex-constrained least squares estimator, 
the XEB estimator~\eqref{eq:xeb_app}, and the 
hard-thresholded XEB estimator~\eqref{eq:xeb_thresh}.
For the latter estimator, we choose the tuning parameter $\lambda$ via two-fold cross-validation~\citep{arlot2010survey}. 
We apply these estimators to two different models: 
\begin{enumerate}
    \item The rows of  $\Pi$ are drawn independently and uniformly from $\Delta_d$  (i.e. assumption~\ref{assm:pt} holds).
    \item The rows of $\Pi$ are obtained from numerical simulation of a one-dimensional brick-layer quantum random circuit. In particular, the first row of $\Pi$ corresponds to perfect simulation of one specific instance of the random circuit. Each of the other rows corresponds to one single Pauli error occurring in the circuit. Different rows  correspond to different Pauli errors (either occurring on different qubits or at different layers of the circuit). 
\end{enumerate}  
For both models, we choose the Hilbert space dimension as $d=2^{16}=65,536$, and   the number of 
errors $k$ from $\{46,181,721\}$.  
Recall that the first entry of $c$ should be thought of as the fidelity
of the device---representing the probability of noiseless
execution. We always set this entry to 0.5; the remaining $k-1$ entries
correspond to the probabilities of different errors occurring in the circuit, 
and we sample them from a uniform distribution on $\Delta_{k-1}$. 

The $\ell_2$ errors of these estimators  are summarized in Fig.~\ref{fig:scaling_analysis_A}
as a function of the sample size.
A few important remarks are in order.   
First, the MLE, regularized XEB, and simplex-constrained least squares estimators
all exhibit two regimes of estimation error. 
When the sample size $n$ is relatively small, the $\ell_2$ error appears to be almost independent of $k$, and improves roughly as $n^{-1/4}$, whereas for large $n$, the $\ell_2$ error becomes linear in $k$ and scales as $n^{-1/2}$;
these scalings are consistent with Proposition~\ref{prop:collision}. 
Second, the performance of the XEB estimators plateaus
in the unrealistic regime  $n > d$, which is to be expected from its derivation.   
This effect is more severe in the brick-layer model,
which suffers from slight dependence among the row vectors in the $\Pi$ matrix. 
Third, for the unregularized estimators (XEB and least squares), the $\ell_2$ error always increases linearly in $k$. 
Finally, we find that even in realistic circuits, the MLE does not outperform 
the other estimators by an appreciable margin, even though it accounts for the mild heteroscedasticity of the model.
\begin{figure*}[!tb]
    \centering
\includegraphics[width=\linewidth]{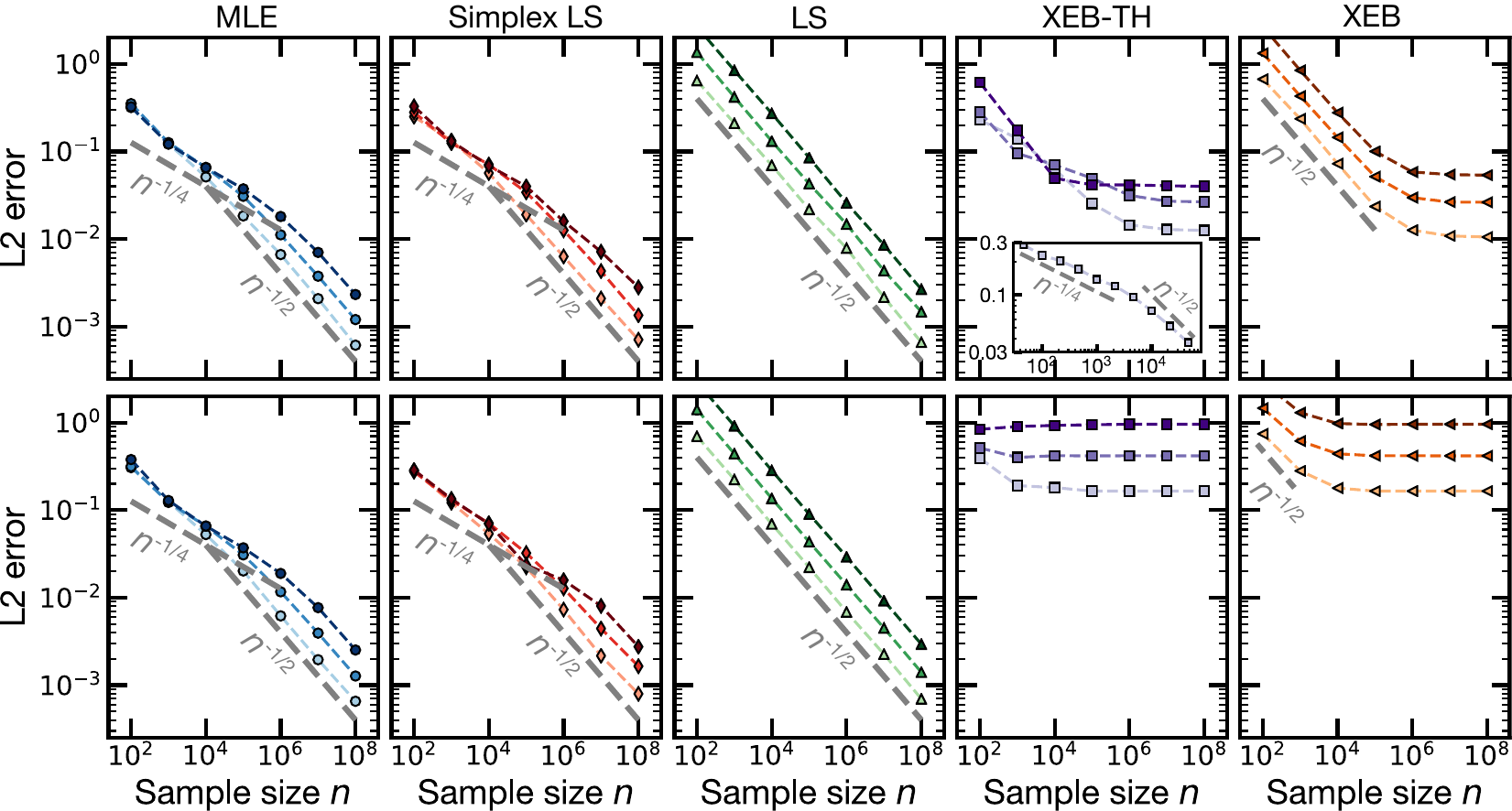}
\caption{Sample complexity of various estimators in regime A. For each panel, from the lightest to the darkest, $k$ is chosen from $\{46,181,721\}$. For the top (bottom) row, the matrix $\Pi$ is taken from the Dirichlet distribution (simulating random unitary circuits). Inset for XEB-TH: zoom-in of the regime with small sample sizes to highlight the $n^{-1/4}$ scaling. Each data point is obtained by averaging over at least 10 repetitions of simulation. }
\label{fig:scaling_analysis_A}
\end{figure*}

\subsection{Regime B}
When the amount of side information $m$ is finite but nonzero, 
our problem can be interpreted
as an {\it errors-in-variables}~\citep{carroll2006} 
multinomial regression problem, 
in which the matrix $\Pi$ is unknown, but noisy samples
$W$ from its rows are given. 
We consider three estimators
for this problem, which we discuss in turn. 
\\

\subsubsection{The Collision Estimator}
One of the benefits of the XEB
estimator from the previous section
is the fact that it is a {\it linear}
functional of $\Pi$. 
Therefore, this estimator
can   be adapted to Regime B
by  
replacing $\Pi$ with its empirical
counterpart, without incurring
any bias. 
 Concretely,  
by defining 
$$\widehat \Pi = V / m,\quad \text{i.e. } 
\widehat \Pi_{ij} = \frac 1 m \sum_{r=1}^m I(W_{ir} = j),
~~i=1,\dots,k,j=1,\dots,d,$$
we arrive at the following
counterpart of the XEB estimator,
which was already presented in equation~\eqref{eq:main_ortho_eiv}:
\begin{align}
\label{eq:collision}
\hat c^{\mathrm{coll}} = 
(d/n) \widehat\Pi Y - \one_k,
\quad \text{i.e. }
\hat c_i^{\mathrm{coll}} = \frac d {nm}\sum_{\ell=1}^n \sum_{r=1}^m
I(Z_\ell=W_{ir}) - 1,\quad i=1,\dots,k.
\end{align} 
As before, we refer to this estimator
as the {\it collision} estimator. 
We can also form
regularized variants
of this estimator via thresholding:
\begin{align}
\label{eq:collision_HT}
\hat c^{\mathrm{HT}} := \calH_\lambda(\hat c^{\mathrm{coll}}), \quad 
\hat c^{\mathrm{ST}} = \calS_\lambda(\hat c^{\mathrm{coll}}),
\end{align}
for some $\lambda > 0$. 
The following result, which we prove
in Section~\ref{sec:ub_proofs}, 
shows that the hard-thresholding  
collision estimator
achieves the optimal sample
complexity stated in Theorem~\ref{thm:main_unsorted}.
\begin{proposition}
\label{prop:collision}
There exists
a universal constant $C > 0$ such that
for all $1 \leq n,k\leq d$, $m\geq 1$,  
there exists $\lambda \geq 0$ such that
$$\sup_{c\in \Delta_k} \bbE\|\hat c^{\mathrm{HT}}-c\|_2 \leq C
\cdot \min\left\{ \left(\frac{dk}{nm_d}\right)^{\frac 1 2} ,
\left(\frac{d\log k}{nm_d}\right)^{\frac 1 4}\right\},$$
with $m_d = \min\{m,d\}$. 
\end{proposition}

As before, the collision
estimator has the advantage
of being computationally
efficient, and achieves the minimax optimal rate of convergence, 
but has the downside
of relying strongly on  
condition~\ref{assm:pt}. 
We next develop 
an estimator which 
somewhat relaxes this assumption.\\

\subsubsection{The Errors-in-Variables Estimator} 
Recall
  from the previous
section that the XEB estimator
can be understood as an approximation
of the ordinary least squares
estimator with an orthonormal design matrix
$\Pi$ (up to centering). 
One way of generalizing 
the least squares estimator to Regime B is given by the following estimator:
\begin{align} 
\argmin_{x \in \bbR^k}\Big\{ x^\top A_Vx  - 2 \frac{Y^\top V^\top x}{m}\Big\}
 = A_V^{-1}\frac{VY}{nm}, 
\end{align}
where $A_V$ is an unbiased estimator of $\Pi\Pi^\top$ based on the side information $V$. 
Again, one can also regularize this estimator by restricting
its feasible set to the simplex, at the expense of losing the closed-form solution:
\begin{align} 
\label{eq:eiv_ls}
\hat c^{\mathrm{EiV}} = 
\argmin_{x \in \Delta_k}\Big\{ x^\top A_Vx  - 2 \frac{Y^\top V^\top x}{m}\Big\}.
\end{align} 
Variants of these estimators
have previously appeared
in the work of~\citep{loh2012}, where
they were motivated by the fact that
their objective function is an unbiased
estimator
of the usual least squares
objective $\|Y - n\Pi^\top \gamma\|_2^2$,
up to addition of a constant that does not depend
on~$\gamma$. 
Analogues of such estimators
for Poissonian models  have
  appeared in the works of~\citep{jiang2022a,jiang2023},
though are based on a log-link parametrization
and do not easily extend to our setting.  

A variety of matrices $A_V$ can be used in 
the definition of $\hat c^{\mathrm{EiV}}$.
When condition~\ref{assm:pt}
happens to hold, 
one such estimator is given by
\begin{equation}
\label{eq:Av}
A_V := \mathbb{E}[\Pi\Pi^\top \,|\, V],
~~~ \text{i.e. } 
(A_V)_{i\ell} =\mu_i^\top \mu_\ell
+  \frac{(d+m)\|V_{i\cdot}+1\|_1 - \|V_{i\cdot}+1\|_2^2}{(d+m)^2(d+m+1)} I(i=\ell),
\end{equation}
where $\mu_i = (V_{i\cdot} + 1) / (d+m)$. 
Below, we will see that with this choice of $A_V$, the estimator~\eqref{eq:eiv_ls} achieves reasonable performance 
even on simulated random circuit data where   mild deviations
from the Porter-Thomas assumption can occur. 
\\

\subsubsection{The Variational Maximum Likelihood Estimator}
Our third estimator
in Regime B is the maximum likelihood
estimator (MLE), which was defined
in equation~\eqref{eq:side_info_mle}. As discussed therein, 
the MLE is nonconvex in this problem, and we proposed
to approximate it by the fixed-point equation~\eqref{eq:em_fixed_point}.
The goal of this subsection is to derive that approximation.
In statistical terms, this fixed-point equation arises as
the limit of a mean-field variational
expectation-maximization (EM) algorithm~\citep{dempster1977maximum,neal1998view,jordan1999introduction}, 
and is inspired by~\citep{blei2003}. We develop it from first
principles for completeness.

It will be convenient to rewrite the
data generating   distribution according
to the following hierarchy: Given
a matrix $\Pi$ satisfying
the Porter-Thomas assumption~\ref{assm:pt},
for $i=1,\dots,k$, 
$\ell=1,\dots,n$ and $r=1,\dots,m$,
we draw, 
\begin{align*} 
U_\ell^c \sim \textstyle \sum_{i=1}^k c_i\delta_i,\quad 
Z_\ell \,|\, (U_\ell^c,\Pi) \sim \sum_{j=1}^d (\pi_{U_\ell^cj}) \delta_{z_j},\quad 
W_{ir}\,|\,\Pi\sim \sum_{j=1}^d \pi_{ij}\delta_{z_j}.
\end{align*}
That is,  $U_\ell^c\in\{1,\dots,k\}$ denotes
a latent variable
which indicates the category from which
bitstring $Z_\ell$ was drawn. 
In this notation, the likelihood function of $c$ under Regime~B 
can be expressed as:
\begin{align*}
\calL(x) = \bbE_{U^x,\Pi} \left[\prod_{\ell=1}^n \pi_{U_\ell^x Z_\ell}
\cdot \prod_{r=1}^m \prod_{i=1}^k \pi_{iW_{ir}}
\right],\quad x \in\Delta_k,
\end{align*}
which depends on $x$ only through the 
distribution of the vector
$U^x=(U_\ell^x)_{\ell=1}^n$.
The above
expression can be interpreted 
as the partition function
of a classical physical system with states $\Pi,U^x$, and
Hamiltonian
$$\calH(U^x,\Pi|Z,W) = 
\sum_{\ell=1}^n \log (\pi_{U_\ell^x Z_\ell} )+ 
\sum_{r=1}^m \sum_{i=1}^k \log(\pi_{iW_{ir}}).
$$
By the Gibbs variational principle, 
one can express the partition function
via
\begin{equation}\label{eq:partition_fn}
\log \calL(x) =
\sup_{J} 
\Big\{ 
\bbE_{(\widebar U,\widebar \Pi)\sim J}\big[\calH(\widebar U,\widebar \Pi|Z,W)\big] - 
\KL  \big(J\,\|\,P_{U^x,\Pi|Z,W}\big) 
\Big\},
\end{equation}
where $P_{U^x,\Pi|Z,W}$ denotes
the joint probability distribution  of the latent variables
$U^x$ and $\Pi$ given the observables $Z$ and $W$, 
and the supremum is taken
over all   probability distributions
$J$ on $[k]^n \times \Delta_d^k$.
This representation suggests approximating
 the intractable distribution
$P_{U^x,\Pi|Z,W}$ 
 by a tractable family of joint
 distributions. 
We will adopt a {\it mean-field}
approximation, taking this family to consist
of all {\it independent} joint distributions $J_\phi$ on $[k]^n \times \Delta_d^k$,
whose first marginal is any discrete distribution $\phi \in \Delta_k^n$, and whose
second marginal is  given by 
the posterior law $\calD_W$ of $\Pi$ given $W$, 
$$\Pi\,|\,W \sim \calD_W := \bigotimes_{i=1}^k \mathrm{Dirichlet}\big(1+V_{i1}, \dots,
1+V_{id}\big),\quad \text{with } V_{ij} = \sum_{r=1}^m I(W_{ir} = j).$$
That is, we will restrict the supremum in equation~\eqref{eq:partition_fn}
to the set of joint distributions of the form
$$(\widebar U, \widebar\Pi) \sim J_\phi = \bigg(\bigotimes_{\ell=1}^n \sum_{i=1}^k \phi_{i\ell} \delta_i \bigg) \otimes \calD_W,
\quad \phi\in\Delta_k^n.$$ 
This leads to the following lower
bound on the partition function:
\begin{align*}
\log \calL(x) &\geq 
\sup_{\phi\in\Delta_k^n} \calF(x,\phi),\quad x\in\Delta_k,
\end{align*} 
where $\calF(x,\phi)$ denotes the mean-field free entropy,
\begin{align*} 
\calF(x,\phi) =
\bbE_{(\widebar U,\widebar \Pi)\sim J_\phi}\big[\calH(\widebar U,\widebar \Pi|Z,W)\big] - 
\KL  \big(J_\phi\,\|\,P_{U^x,\Pi|Z,W}\big),
\end{align*}
The variational EM algorithm consists of performing coordinate
ascent on the mean-field free entropy: Given
initial values $x^{(0)}\in \Delta_k$  and $\phi^{(0)} \in \Delta_k^n$,
we perform the following iterations for all $t=0,1,\dots$
\begin{enumerate}
\item[(a)] $\phi^{(t+1)} = \argmax_{\phi\in \Delta_k^n} \calF(x^{(t)},\phi)$.
\item[(b)] $x^{(t+1)} = \argmax_{x \in\Delta_k}\calF(x,\phi^{(t+1)})$.
\end{enumerate}
In order to derive the maxima in the above steps, notice first
that the free entropy can be rewritten as
\begin{align*}
\calF(x,\phi) = 
 \bbE_{(\widebar U,\widebar \Pi)\sim J_\phi}\big[\log p_{U^x,\Pi,Z,W}(\widebar U,\widebar \Pi,Z,W)\big] + H(J_\phi) - \log p_{Z,W}(Z,W),
\end{align*}
where $H$ denotes the differential entropy of $J_\phi$, and the joint law
of the random variables is given, over their support, by
$$p_{U^x,\Pi,Z,W}(\widebar U, \widebar \Pi,Z,W) = \prod_{\ell=1}^n x_{\widebar U_\ell} \widebar \pi_{\widebar U_\ell Z_\ell} 
\cdot 
\prod_{i=1}^k \prod_{r=1}^m \widebar \pi_{iW_{ir}}.$$
Letting ``$\propto$'' denote
equality
up to additive constants not depending on $x,\phi$, we deduce
\begin{align*}
    \calF(x,\phi)
     &\propto  \sum_{i=1}^k\sum_{\ell=1}^n \phi_{i\ell} \bbE_{\calD_V}[\log (x_i\widebar \pi_{iZ_\ell})]
      - \sum_{i=1}^k  \sum_{\ell=1}^n\phi_{i\ell} \log\phi_{i\ell} \\
&= \sum_{i=1}^k  \sum_{\ell=1}^n \phi_{i\ell} \log\frac{x_i}{\phi_{i\ell}} 
+ \sum_{i=1}^k \sum_{\ell=1}^n \phi_{i\ell} \big(\psi(1+W_{iZ_\ell})-\psi(d+m)\big)\\
&\propto \sum_{i=1}^k \sum_{\ell=1}^n \phi_{i\ell} \log\left(\frac{x_i \exp\{\psi(1+W_{iZ_\ell})\}}{\phi_{i\ell}}\right),
 \end{align*}
where $\psi$ denotes the di-gamma function, and the second line
is obtained in closed form using the fact that the marginal distribution of $\widebar\pi_{ij}$ is $\mathrm{Beta}(1+W_{iZ_\ell},d-1+\sum_{s \neq Z_\ell} W_{is})$.
By maximizing the above display with respect to each variable $x$ and $\phi$, we deduce that
the iterations (a) and (b) reduce to:
\begin{enumerate}
\item[(a)] $\phi_{i\ell}^{(t+1)} = x_i^{(t)}   \exp\{\psi(1+W_{iZ_\ell})\} / \sum_{r=1}^k x_r^{(t)} 
\exp\{\psi(1+W_{rZ_\ell})\}$,
\item[(b)] $x_i^{(t+1)} = \frac 1 n \sum_{\ell=1}^n \phi_{i\ell}^{(t+1)}$,
\end{enumerate}
for $t=0,1,\dots$ These iterations simplify to 
\begin{align}\label{eq:em_vi_algo} 
x_i^{(t+1)} = \frac{x_i^{(t)}}{n} \sum_{\ell=1}^n 
\frac{\exp\{\psi(1+W_{iZ_\ell})\} }{\sum_{r=1}^k x_r^{(t)}   \exp\{\psi(1+W_{rZ_\ell})\}}
 = \frac{x_i^{(t)}}{n} \sum_{j=1}^d
\frac{Y_jS_{ij} }{\sum_{r=1}^k x_r^{(t)}  S_{ij}},
\quad i=1,\dots,k,
\end{align}
where $S_{ij} = \exp\{\psi(1+V_{ij})\}$. We refer to these iterates
as the {\it variational EM estimator}. These iterates 
converge precisely to a solution of the
mean-field fixed-point equation
$$n = \sum_{j=1}^d \frac{Y_j S_{ij}}{\sum_{r=1}^k S_{rj} x_r},\quad i=1,\dots,k,$$
which completes our derivation of equation~\eqref{eq:em_fixed_point}. 

Algorithm~\eqref{eq:em_vi_algo} has the well-known property of increasing the likelihood
at each iteration, in the sense that $\calL(x^{(t+1)}) \geq \calL(x^{(t)})$
for all $t\geq 0$. 
Nevertheless, this algorithm is not guaranteed to converge to the maximum likelihood estimator, 
i.e. the global 
optimum of~$\calL$. 
In a closely-related statistical model known as latent Dirichlet allocation, which we will discuss further below,
it has recently been shown that the mean-field approximation
is negligible for computing Bayesian estimators
when $n = o(k)$~\citep{zhong2025variational}, however
this analysis corresponds most closely
to our model with $m=\infty$. In contrast, we believe that the mean-field
approximation becomes very poor when $m$ is of lower order than $n$ and $d$,
as will become clear in our numerical comparisons below. 
As we have already indicated, one avenue for possible improvement of this algorithm is 
to optimize the Thouless-Anderson-Palmer free entropy~\citep{ghorbani2019instability,celentano2023local,celentano2023mean}
instead of the mean-field free entropy. We intend to pursue this avenue in future work.

\subsubsection{Numerical Comparison in Regime B} 

We again analyze five different estimators in Regime B, which are counterparts
to those in Appendix~\ref{app:regime_A_numerical} above:
The variational EM estimator~\eqref{eq:em_vi_algo}, 
the error-in-variable least-square estimators 
with and without the simplex
constraint~\eqref{eq:eiv_ls}, and the collision-based estimators with and
without hard-thresholding~\eqref{eq:collision}--\eqref{eq:collision_HT}.
We apply these estimators to the same data as in  Appendix~\ref{app:regime_A_numerical}.
%
%
This time, we fix the number of errors $k=46$, and vary the side information sample size $m$. 

The $\ell_2$ errors of these estimators as a function of the sample size are summarized in Fig.~\ref{fig:scaling_analysis_B}.
Let us again make a few remarks. 
When $m$ is sufficiently large (in particular, larger than $d$), all estimators exhibit qualitatively similar
performance as their counterparts in Regime A. On the other hand, 
%
smaller values of $m$ have two main impacts on the the $\ell_2$ error: 
1) They set a lower bound on the attainable $\ell_2$ error, 2) 
although the $\ell_2$ error still decreases as $n^{-1/2}$, the prefactor becomes larger (with a factor of $d/m$). 
These various observations are consistent with our theoretical
predictions in Proposition~\ref{prop:collision} and Theorem~\ref{thm:main_unsorted}.
Furthermore, we observe several phenomena which parallel those of Regime A: 1) The various
estimators have comparable risks in most regimes, highlighting the fact that 
simple least squares-based estimators do not lose much despite the heteroscedastic nature of the problem, and
2) The collision estimators experience a plateau in performance when $n$ exceeds $d$. 
The variational EM estimator also experiences a pronounced plateau when $m$ is small, which 
is likely due to the poor approximation properties of the mean-field free entropy when $m$ is small.

\begin{figure*}[!tb]
    \centering
\includegraphics[width=\linewidth]{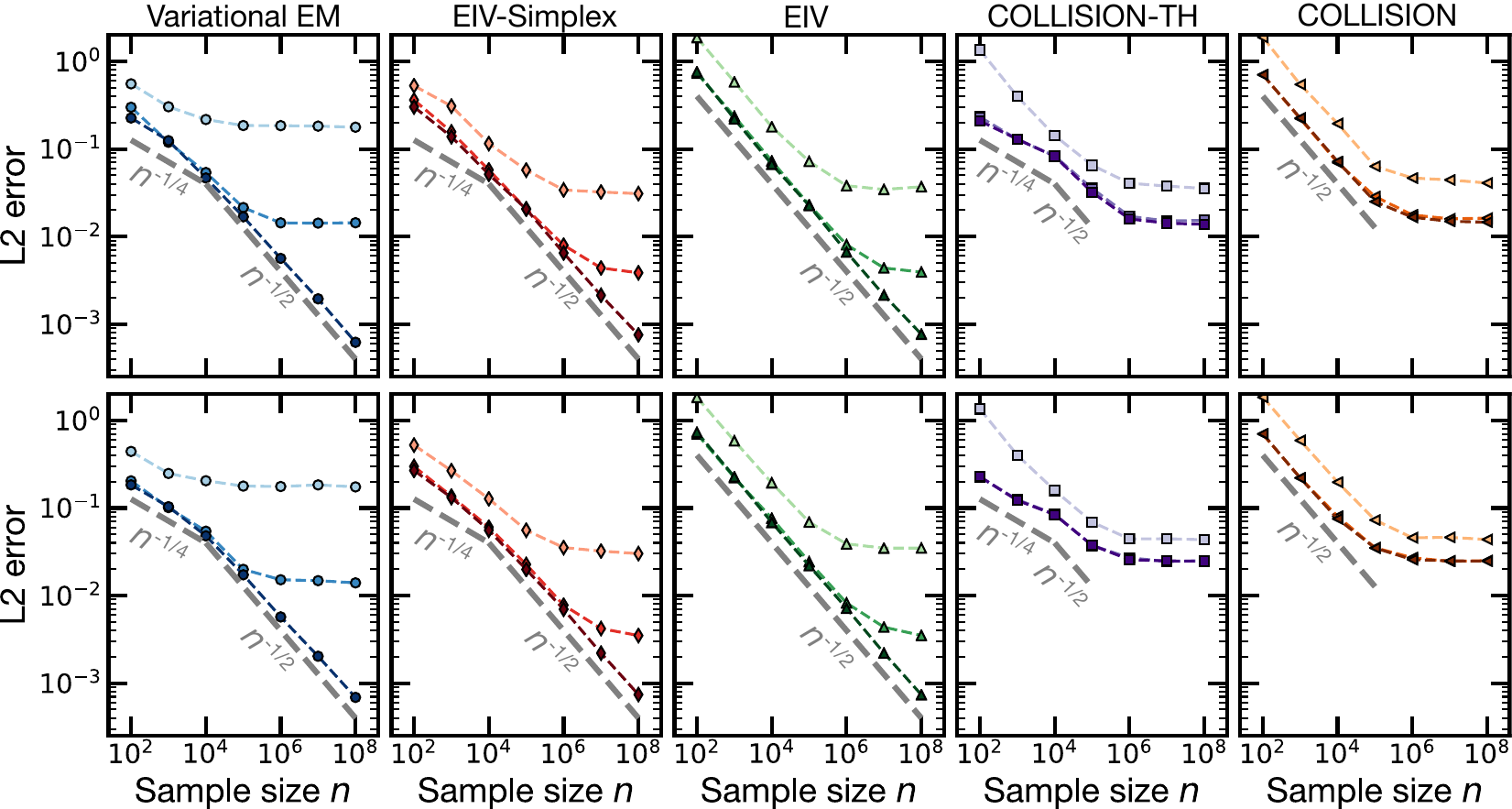}
\caption{Sample complexity of various estimators in Regime B. For each panel, from the lightest to the darkest, the side information sample size $m$ is chosen from $\{10^4,10^6,10^8\}$. For the top (bottom) row, the matrix $\Pi$ is taken from the Dirichlet distribution (simulating random unitary circuits). Each data point is obtained by averaging over at least 10 repetitions of simulation. }
\label{fig:scaling_analysis_B}
\end{figure*}

\subsection{Regime C}
Unlike Regimes A and B, which 
are  connected to multinomial
regression with measurement error, 
Regime C is closer in spirit
to latent variable or blind source separation models,  
such as nonnegative matrix factorization~\citep{donoho2003}, independent component analysis~\citep{hyvarinen2001independent}, latent class analysis~\citep{goodman1974exploratory},  and 
particularly topic models, which we discuss
in further detail below.
Moment-based estimators are known to be information-theoretically optimal
for many of these problems (e.g.~\citep{moitra2018algorithmic} and references therein). 
Inspired by this line of work, our focus will  
be to derive a moment estimator for  Regime C. We rely on   assumption~\ref{assm:pt}
throughout this subsection.

\subsubsection{The Moment Estimator}\label{app:moment}
Recall from
equation~\eqref{eq:looks_like_poisson} that the
histogram entries $Y_j$  are approximately  independent, conditionally on the latent
variable $\Pi$, and approximately follow
a $\mathrm{Poi}(n\textstyle\sum_{i=1}^k c_i\pi_{ij})$ distribution. Furthermore, 
as we recall in Lemma~\ref{lem:dirichlet} below, the flat Dirichlet 
vector $\Pi_{i\cdot}=(\pi_{i1},\dots,\pi_{id})$ is equal in distribution to
$$\Pi_{i\cdot} = \frac{(X_{i1},\dots,X_{id})}{\sum_{j=1}^d X_{ij}},\quad i=1,\dots,k,$$
where $X_{ij} \sim \calE_d$ are i.i.d. exponential  random variables. 
When the dimension $d$ is large, the denominator in the above display concentrates rapidly around
1, and in this case one can further approximate the histogram $Y$ by a random vector $\widebar Y$
which follows the distribution:
$$\widebar Y_j|X \sim \mathrm{Poi}\Big(n\sum_{i=1}^k c_i X_{ij}\Big),\quad j=1,\dots,d.$$
Due to the independence of the rows of $X=(X_{ij})$, the random variables
$\widebar Y_j$ are marginally~i.i.d.:
\begin{equation}
\label{eq:poisson_model_for_expo}
\widebar Y_j \overset{\mathrm{i.i.d.}}{\sim} \bQ_c = \int_{\bbR^k_+}\mathrm{Poi}(n \langle \varpi,c\rangle) d\calE_d^{\otimes k}(\varpi), \quad j=1,\dots,d.
\end{equation}
We will show in the following section that model~\eqref{eq:poisson_model_for_expo}
is statistically indistinguishable from the original model for $Y$, when the dimension
$d$ is sufficiently large. We will also show that model~\eqref{eq:poisson_model_for_expo} is identifiable
up to sorting the entries of $c$, in the sense that for any $c,\bar c \in \Delta_k$, 
$\bQ = \bQ_{\bar c}$ implies that $W(c,\bar c)=0$. Taking these facts for granted momentarily, 
we will derive a moment-based
estimator of $c$  motivated by model~\eqref{eq:poisson_model_for_expo}. 

Since we only seek to estimate $c$ up to permutation of its entries, it will suffice
to derive estimators
for the first
$k$ moments of $c$, namely
$$m(c) = \big(m_1(c),\dots,m_k(c)\big)^\top,
\quad \text{with }
m_p(c) = \sum_{i=1}^k c_i^p,
\quad p=1,\dots,k.$$
These moments uniquely
identify $c$ up to permutation of its entries (cf.~\citep{hundrieser2025}), since
the polynomial 
$$f_c(z) = \prod_{i=1}^k (z-c_i),
\quad z \in \bbC,$$
 is uniquely
determined, on the one hand, by its
set of roots $\{c_1,\dots,c_k\}$, 
and on the other hand, by its moment
vector $m(c)$. This fact follows
from {\it Newton's identities}
(recalled in~Appendix~\ref{app:elementary_symmetric_polynomials}), which imply 
\begin{align*} 
f_c(z) &= z^k + \sum_{j=1}^k (-1)^j e_j(c) z^{k-j},\quad z \in \bbC,\\
&\text{with } 
e_0(c) =1, ~~ e_\ell(c) = \frac k \ell \sum_{j=1}^\ell (-1)^{j-1} e_{\ell-j}(c) m_j(c),
\quad \ell=1,\dots,k.
\end{align*}
Given estimators of the moments
of $c$, say $\hat m_1,\dots,\hat m_k$, which we will define
below, one can
now  construct an 
estimator of $c$ by forming
an estimator of $f_c$, and 
returning its $k$ roots. 
Concretely, define \begin{align}\label{eq:fitted_newton_identity}
\hat e_0 = 1,\quad  	\hat e_\ell 
 	&= \frac{k}{\ell} \sum_{j=1}^{\ell} (-1)^{j-1} \hat e_{\ell-j}  \hat m_j,
 	~~\ell=1,\dots,k.
 \end{align}
and define the fitted $k$-degree polynomial
$$\hat f(z) =  z^k + \sum_{j=1}^k (-1)^j \hat e_j z^{k-j},\quad z \in \bbC.
$$
We then define 
a pilot estimator $\tilde c = (\tilde c_1,\dots,\tilde c_k)$ as the vector
of $k$ (possibly complex) roots of $\hat f$, ordered arbitrarily. 
To obtain an estimator with real coordinates, we
define our final estimator by
\begin{align} 
\label{eq:mom_estimator}
\hat c = \big( \Re(\tilde c_1),\dots,\mathrm{Re}(\tilde c_k)\big).
\end{align}
We refer the above as the 
{\bf moment estimator} for Regime C. In order to complete its
definition, we need 
to define the estimators
$\hat m_1,\dots,\hat m_k$, 
which we turn to next.  

Our starting point is inspired
by the past works of~\citep{arute2019quantum,shaw2024universal,andersen2025thermalization,rinott2022statistical}, 
which   noted that 
when $k=2,3$, the {\it cumulants}
of the histogram are
related to the moments of $c$. 
Let $\kappa_p(X)$ denote the $p$-th cumulant of any random variable $X$.
Given independent
exponential random variables $\varpi_1,\dots,\varpi_k \sim \calE_d$, define the random variable 
$\theta = \sum_{i=1}^k c_i \varpi_i$, and write 
$\xi_p = \kappa_p(\theta)$
for any $p\geq 2$. 
Recalling that cumulants are additive across sums of independent random variables,  we have
\begin{equation}
\label{eq:cumulant_to_moment}
\xi_p = \sum_{i=1}^k c_i^p \kappa_p(\varpi_{i}) = \frac{(p-1)!}{d^p} m_p(c).
\end{equation}
Thus, to estimate the moment vector $m(c):=(m_1(c),\dots,m_k(c))^\top$, it suffices
to estimate the cumulant vector $\xi = (\xi_1,\dots,\xi_k)^\top$.  
To do so, we recall 
that the cumulants $\xi_p = \kappa_p(\theta)$
are related to the moments $\eta_p = \bbE[\theta^p]$ via 
\begin{equation}
\label{eq:cumulant_bell}
\xi_p = \sum_{\ell=1}^p  (-1)^{\ell-1} (\ell-1)! B_{p,\ell}\big(\eta_1,\dots,\eta_{p-\ell+1}\big),
\end{equation}
where $B_{p,\ell}$ are the 
Bell polynomials (cf. Appendix~\ref{app:bell_polynomials}),
defined by
\begin{align}\label{eq:bell_deriv}
B_{p,\ell}(x_1,\dots,x_{p-\ell+1}) = p! \sum_{(h_1,\dots,h_{p-\ell+1})\in \calH_{p,\ell}} \prod_{i=1}^{p-\ell+1} \frac{x_i^{h_i}}{(i!)^{h_i} h_i!},
\end{align}
where $\calH_{p,\ell}$ consists
of all tuples $(h_1,\dots,h_{p-\ell+1})$ of nonnegative integers
such that:
$$\sum_{i=1}^{p-\ell+1} h_i = \ell,\quad \sum_{i=1}^{p-\ell+1} ih_i = p.$$
These expressions suggest that the cumulants $\xi_p$ can be estimated 
by first estimating
the moments $\eta_\ell$, which 
in turn can be done using the classical unbiased estimators
for a Poisson model, given by 
$$T_{j,\ell} = \begin{cases}
1/d, & \ell=1,\\
\frac{Y_j!}{n^\ell(Y_j-\ell)!},&\ell=2,3,\dots
\end{cases}$$
Specifically, one has $\eta_\ell = \bbE[T_{j,\ell}]$
for all $\ell=1,2,\dots$. 
We can use these quantities
to 
build estimators
for each term in the summations~\eqref{eq:cumulant_bell}--\eqref{eq:bell_deriv}. 
To this end, given $\bh \in \calH_{p,\ell}$, 
let $\calI_{\ell,\bh}$ be the set of all tuples  
$(S_1,\dots,S_{p-\ell+1})$ consisting of pairwise disjoint sets $S_1,\dots,S_{p-\ell+1} \subseteq \{1,\dots,d\}$
such that $|S_i| = h_i$ for all $i$.  Notice that some of the $S_i$ could be empty, and 
that $\big|\bigcup_{i} S_i\big| = p$ by definition of $\calH_{p,\ell}$. Furthermore, 
$$|\calI_{\ell,\bh}| = {d \choose \ell} {\ell \choose {\bh} },$$
where the second factor denotes a multinomial coefficient. 
%
Now, define the generalized U-Statistic:
$$W_{\bh} = \frac 1 {{d \choose \ell} {\ell \choose {\bh} }}  \sum_{(S_1,\dots,S_{p-\ell+1})\in \calI_{\ell,\bh}}
\prod_{i=1}^{p-\ell+1} \prod_{j\in S_i} T_{j,i}.$$
If the Poisson model~\eqref{eq:poisson_model_for_expo} held true, then the above would be an unbiased estimator of
$\prod_{i=1}^{p-\ell+1} \eta_i^{h_i}.$
A natural
estimator of $B_{p,\ell}(\eta_1,\dots,\eta_{p-\ell+1})$ is thus the following:
$$p! \sum_{\bh\in \calH_{p,\ell}} \frac{W_{\bh}} {\prod_{i=1}^{p-\ell+1} (i!)^{h_i}h_i!}.$$
Combining these ideas, 
a natural estimator of $\xi_p$ is
given by:
\begin{equation}
\label{eq:u-stat}
\begin{aligned}
\hxi_p &= p!  \sum_{\ell=1}^p (-1)^{\ell-1} (\ell-1)! \sum_{\bh\in \calH_{p,\ell}} \frac{W_{\bh}} {\prod_{i=1}^{p-\ell+1} (i!)^{h_i}h_i!} \\ 
 &=  p!\sum_{\ell=1}^p (-1)^{\ell-1} \frac{(d-\ell)!(\ell-1)!}{d!} \sum_{\bh\in \calH_{p,\ell}}  \sum_{(S_1,\dots,S_{p-\ell+1})\in \calI_{\ell,\bh}} \prod_{i=1}^{p-\ell+1}\prod_{j\in S_i} \frac{T_{j,i}}{i!}.
\end{aligned}
\end{equation}
By combining this expression with equation~\eqref{eq:cumulant_to_moment},
we arrive at the following estimator of $m_p(c)$ (which is unbiased under the Poisson model~\eqref{eq:poisson_model_for_expo}):
\begin{equation}
\label{eq:unbiased_moment_estimate_p}
\hat m_p = pd^p\sum_{\ell=1}^p (-1)^{\ell-1} \frac{(d-\ell)!(\ell-1)!}{d!} \sum_{\bh\in \calH_{p,\ell}}  \sum_{(S_1,\dots,S_{p-\ell+1})\in \calI_{\ell,\bh}} \prod_{i=1}^{p-\ell+1}\prod_{j\in S_i} \frac{T_{j,i}}{i!},\quad p=1,\dots,k.
\end{equation}
Together with equation~\eqref{eq:mom_estimator},
this completes our definition of the moment estimator $\hat c^\mathrm{mom}$. 
Despite its unwieldy definition, this estimator
can be readily implemented in closed form, 
with the exception of a root-finding step.

In what follows, 
we present an upper bound
on the risk of the moment
estimator under the simplified
model~\eqref{eq:poisson_model_for_expo}.
In Appendix~\ref{app:mom_corr},
we will then
show that this upper bound
readily extends to our original
multinomial sampling model.
\begin{proposition}
\label{prop:ub_blind}
Given $\gamma > 0$ arbitrarily small, let $n,d,k\geq 1$ satisfy $d^{1-\frac 1 k} \leq n \leq d^{\frac 1 {1+\gamma}}$. Assume that
 condition~\ref{assm:pt} holds. Then, 
 under model~\eqref{eq:poisson_model_for_expo}, 
 there exists a constant $C =C(k,\gamma) > 0$ such 
that 
$$\sup_{c \in \Delta_k} \bbE_c \Big[W(\hat c^{\mathrm{mom}}, c)\Big] \leq C \sqrt{\frac{d^{1-\frac 1 k}}{n}}.$$
\end{proposition}
The proof appears in Appendix~\ref{app:pf_prop_ub_blind}.
This result shows that the moment estimator achieves the minimax optimal rate
of convergence stated in Theorem~\ref{thm:main_sorted}, under the sorted loss function~$W$. 
Proposition~\ref{prop:ub_blind}
merely studies the estimation rate in the narrow regime $d^{1-1/k} \leq n \leq d^{1-\epsilon}$,
for $\epsilon$ arbitrarily small. When $n$ falls below $d^{1-1/k}$, 
Theorem~\ref{thm:main_sorted} implies that consistent estimation is not possible
uniformly over $\Delta_k$. 
The regime $n > d$ is less relevant for RCS
experiments,  and in either case,  
 finer estimators would need to be adopted 
 in this regime since our approximate model~\eqref{eq:poisson_model_for_expo}
becomes unrealistic when $n > d$ (as we discuss further in Appendix~\ref{app:model}).

Our next results will show that the rate of convergence in Proposition~\ref{prop:ub_blind}
can be significantly improved if  mild constraints are placed on the error vector $c \in \Delta_k$. Indeed, we will see that the convergence behavior of the moment estimator
is highly heterogeneous across the parameter space $\Delta_k$, and  
improves whenever the errors~$c_i$ have {\it different magnitudes}. On a technical level, these improvements
can be anticipated
from the fact that the root-finding step
in the moment estimator is better conditioned
when the underlying roots are well-separated
(cf.\,Lemma~\ref{lem:refined_stability_moments}
below for a quantitative statement
of this fact).

To elaborate, given $1 \leq k_0 < k$, let $\Delta_{k,k_0}$ denote the set of all elements $c^\star \in \Delta_k$
which admit exactly $k_0$ distinct entries, that is, for which the set $\{c^\star_{i}: 1 \leq i \leq k\}$ 
has cardinality equal to $k_0$. Furthermore, let $\Delta_{k,k_0}(\delta)$ denote the set of elements
$c^\star\in \Delta_{k,k_0}$ such that for all $1 \leq i<i'\leq k$, either $c^\star_{i} = c^\star_{i'}$, or $|c^\star_{i} - c^\star_{i'}| \geq \delta$. The following result characterizes the minimax   rate of estimating vectors $c$ which are
in the vicinity of an element in $\Delta_{k,k_0}(\delta)$; roughly-speaking, this
means that the vector $c$ has entries which tightly cluster around $k_0$
different values. 
\begin{proposition}
\label{prop:ub_blind_cluster}
Assume the same conditions as Proposition~\ref{prop:ub_blind}. Then,
under model~\eqref{eq:poisson_model_for_expo}, there exists a constant  $\epsilon> 0$
depending on $k$ 
such that for any $\delta > 0$, there exists $C = C(k,\gamma,\delta) > 0$ such that for all $1 \leq k_0\leq k$, 
$$\sup_{c^\star\in \Delta_{k,k_0}(\delta)} 
  \sup_{\substack{c \in \Delta_k \\ W(c,c^\star)\leq \epsilon}} \bbE_c\Big[ W(\hat c^\mathrm{mom},c)\Big] \leq C \left(\frac{d^{k-1}}{n^{k}}\right)^{\frac{1}{2(k-k_0+1)}}.$$
\end{proposition}
Notice that Proposition~\ref{prop:ub_blind_cluster} recovers the convergence rate of Proposition~\ref{prop:ub_blind} when 
$k_0=1$, which corresponds to no separation assumptions.
On the other hand, it is strictly faster when $k_0 < k$, and improves as $k_0$ increases. In the most extreme case where
all errors $c_i$ are well-separated, corresponding to $k_0=k$, we find that the error vector $c$ can be estimated
at the rate $\sqrt{\frac{d^{k-1}}{n^{k}}}$, assuming again that $d$ exceeds $n$.
This dependence
on the separation of the elements of $c$, as measured by the parameter
$k_0$, is qualitatively
related to 
the minimax rate of estimating finite mixture models
under partial separation assumptions~\citep{heinrich2018,ho2016convergence,ho2019singularity,wu2020,wei2023minimum,ohn2023,hundrieser2025}.

We next state a final refinement of the minimax estimation
rate of $c$. Building upon Refs.~\citep{hundrieser2025,manole2022refined},
we will now show that, not only can the entire vector $c$ be estimated at faster
rates when the coordinates of $c$ are partially separated, but some of the individual coordinates of $c$
enjoy even faster rates than those shown in Proposition~\ref{prop:ub_blind_cluster}. We will prove this by showing that Proposition~\ref{prop:ub_blind_cluster}
continues to hold when $W$ is replaced by a stronger loss function, which sharply captures
the heterogeneity in parameter estimation across the vector~$c$. 
To elaborate, let $c^\star\in \Delta_{k,k_0}$ be given, and assume that its entries are listed in decreasing order.
Recall that the $k$ entries of $c^\star$ are assumed to take on $k_0$ distinct values, 
which we denote by $v_1 > \dots > v_{k_0}$. Let $1 \leq a_j \leq k$ denote
the smallest index $i\in \{1,\dots,k\}$ such
that $c_i^\star = v_j$, and let $r_j$ denote the number of entries in $c$ which are equal to $v_j$:
$r_j = a_{j+1} - a_j$. Furthermore, for all $i=1,\dots,k$, let $j_i \in \{1,\dots,k_0\}$
denote the unique index such that $a_{j_i} \leq i < a_{j_{i+1}}$.

Given elements $c,c'\in \Delta_k$, whose entries we again assume are in decreasing order, we write
\begin{equation}
\calD_{c^\star}(c,c') = \sum_{j=1}^{k_0} \big\|c_{a_j:(a_{j+1}-1)}-c_{a_j:(a_{j+1}-1)}'\big\|^{r_j}_2,
\end{equation}
where
$c_{a_j:(a_{j+1}-1)} \in \bbR^{a_{j+1}-a_j}$ is the vector with coordinates $c_{a_j},\dots,c_{a_{j+1}-1}$.
By abuse of notation, when $c,c'$ are not ordered, we still write $\calD_{c^\star}(c,c')$ as a shorthand for
$\calD_{c^\star}(\mathrm{ord}(c),\mathrm{ord}(c'))$, where $\mathrm{ord}(c)$
is the vector consisting of the same coordinates as $c$, in decreasing order.
Before interpreting this divergence further, let us state our final result.
\begin{proposition}
\label{prop:ub_blind_local}
Assume the same conditions as Proposition~\ref{prop:ub_blind_cluster}. 
Then, under model~\eqref{eq:poisson_model_for_expo},
for any $1 \leq k_0 \leq k$ and $\delta > 0$, 
there exist constants $C,\epsilon > 0$ depending on $k,\gamma,\delta$ such that  
$$ \sup_{c^\star \in \Delta_{k,k_0}(\delta)} \sup_{\substack{c \in \Delta_k \\ W(c,c^\star)\leq \epsilon}} \bbE_c\Big[ \calD_{c^\star}(\hat c^\mathrm{mom},c)\Big] \leq C \sqrt{\frac{d^{k-1}}{n^{k}}}.$$
\end{proposition}
An implication of Propositions~\ref{prop:ub_blind_cluster}--\ref{prop:ub_blind_local}
is the following. Given $c^\star\in \Delta_{k,k_0}$ and a parameter $c$ in a small neighborhood
of $c^\star$, every coordinate of the moment estimator satisfies:
\begin{equation}\bbE |\hat c_i^\mathrm{mom} - c_i| \lesssim \left(\frac{d^{k-1}}{n^k}\right)^{\frac 1 {2r_{j_i}}},
\quad i=1,\dots,k.
\end{equation}
For example, when $k_0=1$, so that no separation assumptions are imposed, $r_{j_i}$ must
always be equal to $k$, and we recover the rate of convergence stated in Proposition~\ref{prop:ub_blind}. 
When $k_0=k$, so that all coordinates of $c$ are well-separated, we must
always have $r_{j_i} = 1$ and we recover the rate of convergence
stated in Proposition~\ref{prop:ub_blind}. When $1 < k_0 < 1$, 
the values of $r_{j_i}$ are always bounded from above by $k-k_0+1$, 
thus Proposition~\ref{prop:ub_blind_local} 
is never worse than Proposition~\ref{prop:ub_blind_cluster}, 
but generally provides different upper bounds for 
estimating different elements $c_i$, depending on their local separation structure. 

This refined convergence rate is particularly well-suited to our problem, since 
we typically expect the fidelity parameter to be appreciably larger than the remaining parameters. 
If we assume  that $c_1 = \|c\|_\infty$ denotes the fidelity, 
and if we make the reasonable assumption that $c_2,\dots,c_k < c_1 - \delta$ for a fixed $k$
and fixed $\delta > 0$, then the above result implies 
$r_1=1$, $j_1 = 1$, and the fidelity can then be estimated at the following fast rate:
$$\bbE |\hat c_1^\mathrm{mom} - c_1| \lesssim \sqrt{\frac{d^{k-1}}{n^k}}.$$
This is the basis for Proposition~\ref{prop:main_fidelity}
of the main text.
 It is worth emphasizing that these
 various local convergence rates
 are achieved  {\it adaptively}: The moment estimator does 
 not rely on any prior information about the possible separation among the atoms of $c$, 
 and achieves these multiscale convergence rates automatically.
\\

\begin{figure*}[!tb]
    \centering
\includegraphics[width=\linewidth]{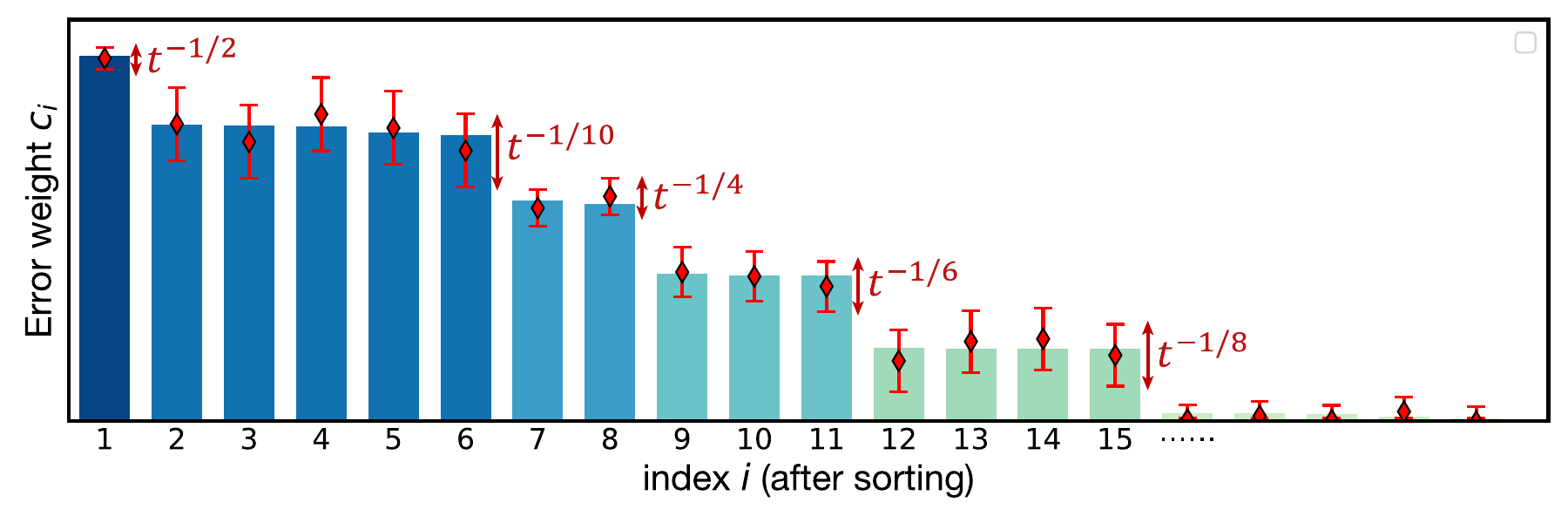}
\caption{
Illustration of the local minimax upper bound
stated in Proposition~\ref{prop:ub_blind_local}, 
with $t = n^k /  d^{k-1}$.
Different entries of the vector $c$ can be estimated at different rates, depending on their local separation
structure, thus significantly improving the rate $t^{-1/2k}$ of Proposition~\ref{prop:ub_blind}, 
which is only sharp when all entries of $c_i$ are near $1/k$.}
\label{fig:scaling_analysis_B}
\end{figure*}

\noindent{\bf Related Work.} 
Our model in Regime C
is closely-related to
to a class of
hierarchical models for text analysis
known as {\it topic models}~\citep{deerwester1990indexing,hofmann1999probabilistic}. 
The typical setup of a topic model is to observe 
a collection of $M$ random variables of the form:
\begin{align*}
Y^{(1)} &\sim \mathrm{Mult}(n;\Pi^\top c^{(1)}) \\ 
Y^{(2)} &\sim \mathrm{Mult}(n;\Pi^\top c^{(2)}) \\ 
&~\vdots \\ 
Y^{(M)} &\sim \mathrm{Mult}(n;\Pi^\top c^{(M)}),
\end{align*}
where $\Pi\in \bbR^{k\times d}$ is an unknown 
stochastic matrix, which is common
across observations, and $c^{(i)} \in \Delta_k$
are unknown mixing weights.
The goal
is to estimate $\Pi$, or the matrix  $C\in \bbR^{k\times M}$ 
comprised of columns $c^{(i)}$, for $i=1,\dots,M$. 
As stated, this problem is not statistically
identifiable without further modeling assumptions. 
A sufficient condition for identifiability
is the so-called   {\it anchor word} 
assumption~\citep{donoho2003,arora2012}
on the matrix $\Pi$, which
forms the basis for a wide array of frequentist
methods for topic modeling~\citep{anandkumar2012spectral,arora2013,arora2012computing,arora2012,bing2020c,bing2022,ke2024using,tran2023}.
A second approach consists
of treating the matrices $\Pi$ and $C$ as latent
variables, endowed with some prior distribution,
which   makes  the model identifiable
under very general conditions~\citep{do2025moment}.
The most widely-used method in this second
category is  {\it latent Dirichlet allocation} (LDA;~\cite{blei2003}), 
which consists of placing  Dirichlet
prior distributions on the columns of 
$\Pi^\top$ and~$C$, and performing
  inference for these objects
via their posterior distribution, 
typically approximated using variational inference.  

When $M=1$, our model is   closely-related
to this second line of literature, since we assume that the rows of the matrix $\Pi$ are drawn from the flat Dirichlet distribution. 
Despite the wide practical
adoption of the LDA model, 
we are only aware of a few references 
that analyze the sample complexity
of parameter estimation in this model, most of which
focus on estimating
$\Pi$ rather than~$C$~\citep{anandkumar2012spectral,nguyen2015posterior,wang2019}.
The very recent work of~\citep{do2025moment}
studies posterior contraction
rates for Bayesian estimators
of $C$, but does so
under  
pointwise and fixed-dimensional asymptotics which are incomparable to our setting.

\begin{figure}[t]
    \centering 
\includegraphics[width=0.48\columnwidth] {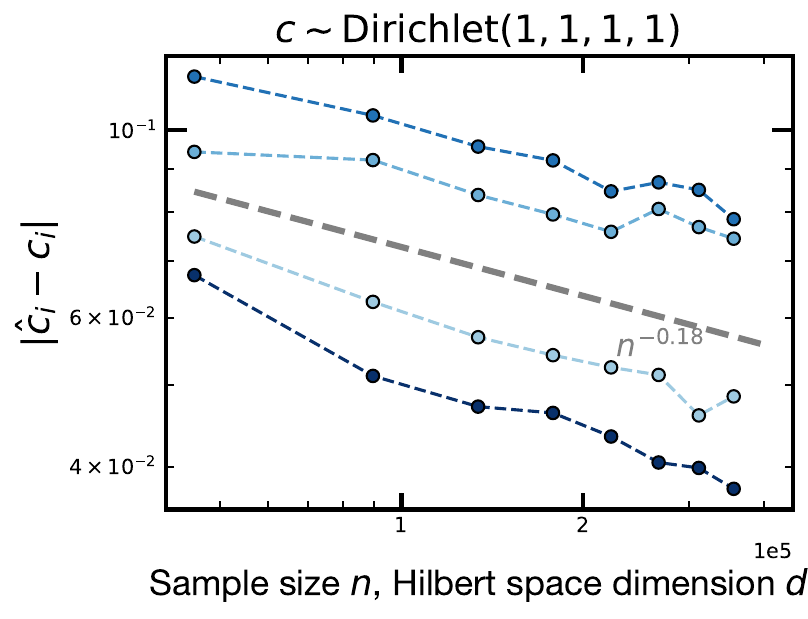}   
 \includegraphics[width=0.48\columnwidth]{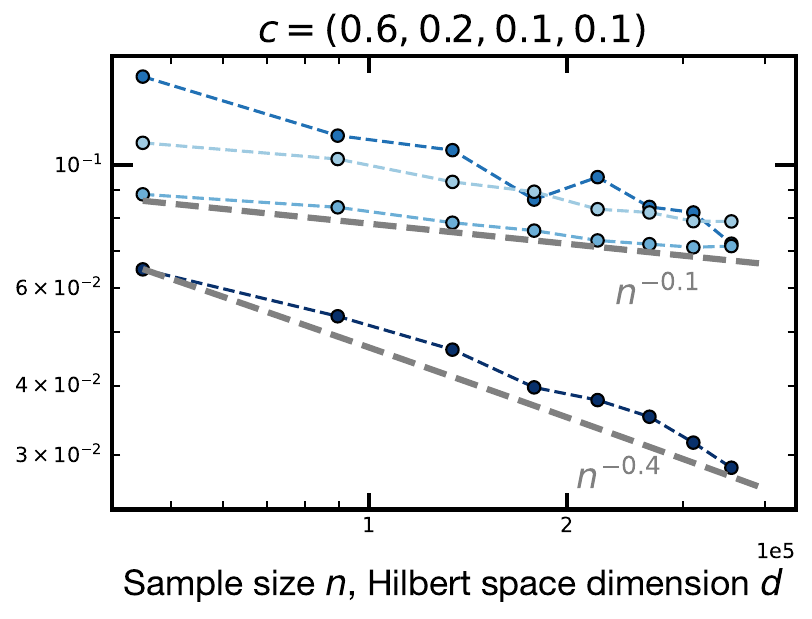}  
\caption{Sample complexity of 
the moment estimator in Regime C, with $k=4$. 
We arrange the elements of $\hat c$
and $c$ in decreasing order, and plot
the average errors $\bbE|\hat c_i - c_i|$
for $i=1,\dots,4$, with $i=1$
denoted by the darkest color, and $i=4$
by the lightest. Across 30
sample sizes $n=d$, we generate
500 replications from each model,
under multinomial sampling.
}
\label{fig:scaling_analysis_C}
\end{figure}

\subsubsection{Numerical Comparison in Regime C}

Although
the moment estimator is currently  our only practical proposal for addressing Regime C,
we close this Appendix by briefly
reporting a simulation study to 
illustrate its empirical sample
complexity, in Figure~\ref{fig:scaling_analysis_C}.
We consider two models: 
In the first, 
we draw $c$ from a flat Dirichlet
law on $\Delta_4$ at each
replication, thus placing $c$
close to the uniform
distribution, which is the least-favourable
parameter from the standpoint of sample
complexity. In the second 
model, we instead 
 fix $c = (0.6,0.2,0.1,0.1)$, 
 in which case the fidelity 0.6
 is appreciably larger than the remaining entries. 
In the first case, we nearly observe
the worst-case rate $n^{-1/8} = 0.125$ 
predicted by Proposition~\ref{prop:ub_blind}, 
across all entries of $\hat c$.  
In the latter case, we instead
see significant variation between
the convergence rate of the various
parameters. In particular, 
the estimated fidelity $\hat c_1$ 
nearly achieves the $n^{-1/2}$ rate
predicted by Proposition~\ref{prop:ub_blind_local}.

\section{Equivalent Statistical Models}
\label{app:model}
In this Appendix, 
we justify  the model approximations which we made in the exposition of the previous section:
We show that, when the dimension
$d$ is sufficiently large, 
our model is statistically equivalent
to a simpler model in which
the sample sizes $n$ and $m$ are Poissonized, and the entries of $\Pi$ are taken to be independent exponential
random variables without normalization. 
We will prove that these various
simplifications do not appreciably
alter the sample complexity of estimating $c$.
We will then adopt
the reduced model for much of the remainder of this manuscript,
without loss of generality.  

\subsection{Statistical Models}
Recall that the 
unsorted and sorted minimax estimation risks are
defined by
\begin{equation} 
\label{eq:baseline_minimax_risks}
\begin{aligned}
\calM(n,d,k,m) &= \inf_{\hat c} 
\sup_{c \in \Delta_k} 
\bbE_c \|\hat c(Z,W) - c\|, \\
\calM_<(n,d,k,m) &= \inf_{\hat c} 
\sup_{c \in \Delta_k} 
\bbE_c W\big(\hat c(Z,W),c\big),
\end{aligned}
\end{equation}
where the 
the infimum is taken over all Borel-measurable functions 
$\hat c:\bbR^n\times \bbR^{k\times m}\to \bbR^k$,
and the
expectation $\bbE_c$ is 
taken over the {\it marginal} distribution
of the random variables $Z=(Z_\ell)$ and $W=(W_{\ell i})$,
that is, over the probability law
$$(Z,W)\sim \bbE_\Pi\left[
\bigg( \sum_{i=1}^k c_i \Pi_{i\cdot}\bigg)^{\otimes n}
\otimes \bigg( \bigotimes_{i=1}^k 
\Pi_{i\cdot}\bigg)^{\otimes m}
\right],$$
where the expectation is 
to be interpreted as marginalization
 over
the law of $\Pi \in \bbR^{k\times d}$, assuming that its
rows are independently
distributed according to the 
flat Dirichlet law $\calD_d$. By abuse of notation, we identify the discrete density
$\Pi_{i\cdot}$ with the probability distribution that it induces. 

A set of sufficient statistics
for the  observations $Z,W$
is given by the following
histograms indexed by $j=1,\dots,d$, which we decorate with superscripts in this section only:
\begin{alignat*}{2}
\widetilde Y_j &= \sum_{\ell=1}^n I(Z_\ell=j),
\quad 
&&(\widetilde Y_1,\dots,\widetilde Y_d) \,\big|\, \Pi
    \sim \mathrm{Mult}(n;\Pi^\top c),  \\ 
\widetilde V_{ij} &= \sum_{\ell=1}^m I(W_{\ell i}=j),\quad 
&&(\widetilde V_{i1},\dots, \widetilde V_{id})\,\big|\,\Pi 
 \sim \mathrm{Mult}(m;\Pi_{i\cdot}),
 \quad i=1,\dots,k.
\end{alignat*}
The joint distribution 
of the random variables
$\widetilde Y=( \widetilde Y_j:1 \leq j \leq d)$
and $\widetilde V=(\widetilde V_{ij}:1 \leq i \leq k,1\leq j \leq d)$
is given by
\begin{equation} 
\label{eq:multinomial_model}
\widetilde \bQ_c = 
\bbE_\Pi \left[\mathrm{Mult}(n;\Pi^\top c)
\otimes \bigotimes_{i=1}^k 
\mathrm{Mult}(m;\Pi_{i\cdot})
\right].
\end{equation}
We will refer to the observation model~\eqref{eq:multinomial_model}
as the {\bf multinomial model}.
It follows by sufficiency of the histograms
$(\widetilde Y,\widetilde V)$ 
that the minimax risks in equation~\eqref{eq:baseline_minimax_risks}
can equivalently
be written as
\begin{equation} 
\label{eq:histo_minimax_risks}
\begin{aligned}
\calM(n,d,k,m) &= \inf_{\hat c} 
\sup_{c \in \Delta_k} 
\bbE_c \|\hat c(\widetilde Y,\widetilde V) - c\|, \\
\calM_<(n,d,k,m) &= \inf_{\hat c} 
\sup_{c \in \Delta_k} 
\bbE_c W\big(\hat c(\widetilde Y,\widetilde V),c\big),
\end{aligned}
\end{equation}
where the infimum is over
all Borel-measurable maps $\hat c$ from $\bbR^d \times \bbR^{k\times d}$ into $\bbR^k$. 


Let us now introduce two alternative 
sampling models which have comparable
minimax risks to the multinomial
model, which we already alluded to in
Section~\ref{sec:estimators}. 
The first of these models
will be referred to
as the {\bf normalized Poisson model}, 
under which the practitioner observes 
 random vectors $(\widebar Y,\widebar V)$ drawn from
the following joint distribution:
\begin{equation} 
\label{eq:normalized_poisson_model}
\widebar \bQ_c = 
\bbE_\Pi \left[\bigotimes_{j=1}^d \bigg(\mathrm{Poi}(n\Pi_{\cdot j}^\top c)
\otimes \bigotimes_{i=1}^k 
\mathrm{Poi}(m \pi_{ij})\bigg)
\right].
\end{equation}
Model~\eqref{eq:normalized_poisson_model}
can be viewed as a variant of model~\eqref{eq:multinomial_model}
in which the sample size $n$
is replaced by a Poisson random variable
$N\sim \mathrm{Poi}(n)$,
drawn independently
of all other random variables. 
The unordered and ordered
minimax risks under the normalized Poisson
model
are given by:
\begin{equation} 
\label{eq:reduced_minimax_risks}
\begin{aligned}
\calR'(n,d,k,m) &= \inf_{\hat c} 
\sup_{c \in \Delta_k} 
\bbE_c \|\hat c(\widebar Y,\widebar V) - c\|, \\
\calR_<'(n,d,k,m) &= \inf_{\hat c} 
\sup_{c \in \Delta_k} 
\bbE_c W\big(\hat c(\widebar Y,\widebar V),c\big),
\end{aligned}
\end{equation}
where the expectation is taken over
$(\widebar Y,\widebar V) \sim \widebar \bQ_c$. 

Our final model is the {\bf unnormalized
Poisson model}, in which the practitioner
observes random variables $(Y,V)$
drawn from the product measure $\bQ_c^{\otimes d}$, 
where
\begin{equation}
\label{eq:unnormalized_poisson_model}
\bQ_c  = \bbE_\varpi \left[\mathrm{Poi}(n\langle \varpi,c\rangle) \otimes \bigotimes_{i=1}^k \mathrm{Poi}(m\varpi_i) \right].
\end{equation}
Here, the expectation is taken over
a random variable $\varpi \sim \calE_d^k$
consisting of independent $\mathrm{Exp}(d)$
entries. The unnormalized Poisson model
can be viewed as a proxy of the normalized Poisson
model, in which the flat Dirichlet 
law of the rows of $\Pi$ are approximated
by the law of $\varpi$. 
The minimax risks under this model
are defined by 
\begin{equation} 
\label{eq:reduced_minimax_risks}
\begin{aligned}
\calR(n,d,k,m) &= \inf_{\hat c} 
\sup_{c \in \Delta_k} 
\bbE_c \|\hat c(Y,V) - c\|, \\
\calR_<(n,d,k,m) &= \inf_{\hat c} 
\sup_{c \in \Delta_k} 
\bbE_c W\big(\hat c(Y,V),c\big),
\end{aligned}
\end{equation}
where the expectations are now
taken over the law of a pair $(Y,V)$
drawn from $\bQ_c^{\otimes d}$.

\subsection{Equivalence of Minimax Risks}
Let us begin by showing that the 
minimax risks for the multinomial
and unnormalized Poisson models
are comparable. 
\begin{lemma}
\label{lem:poissonization}
There exists a universal constant $C > 0$ such that for all $n,m,d,k\geq 1$, 
\begin{multline*}
    \calM(2n,d,k,2m)\cdot \big(1- C e^{-n/C} - C\calI k e^{-m/C}\big) \\ 
\leq \calR'(n,d,k,m)\leq  \calM(n/2,d,k,m/2) +  C e^{-n/C} + C\calI k e^{-m/C},
\end{multline*}
where $\calI = I(m=0)$. An analogous assertions holds for the sorted minimax risks.
\end{lemma}
The proof appears in Appendix~\ref{app:pf_lem_poissonization}, and is based
on a well-known  Poissonization  technique~\citep{canonne2022,jiao2015minimax}. 
As a result of Lemma~\ref{lem:poissonization}, 
we find that the minimax risks $\calM$ and $\calR'$ are comparable, up to additive error terms
which decay exponentially with $n$ and $m$. 
The following result further relates 
the normalized and unnormalized Poisson models, but this time at the level of their distributions
rather than their induced risks.
\begin{lemma}
\label{lem:kl_poisson_models}
There exists a universal constant $C > 0$ such that for all  $n,d,m,k \geq 1$, 
and $c\in \Delta_k$, 
$$\TV(\widebar \bQ_c, \bQ_c^{\otimes d}) \leq C \left( \sqrt{\frac n d} + \sqrt{\frac{mk}{d}}\right).$$ 
\end{lemma}
The proof appears in Appendix~\ref{app:pf_lem_kl_poisson_models}. 
By definition of the total variation distance, one
can find a coupling between any random pairs $(\widebar Y,\widebar V)\sim \widebar \bQ_c$
and   $(Y,V)\sim \bQ_c^{\otimes d}$  such  that
the equality $(Y,V)=(\widebar Y,\widebar V)$ holds with probability at
least $1 - \TV(\widebar \bQ_c, \bQ_c^{\otimes d})$.
Due to the boundedness of the parameter space $\Delta_k$, it then follows
from Lemma~\ref{lem:kl_poisson_models} that
\begin{equation}\label{eq:risk_gap_norm_unnorm}
\calR(n,d,k,m) = \calR'(n,d,k,m) + O\left( \sqrt{\frac n d} + \sqrt{\frac{mk}{d}}\right).
\end{equation}
Together with Lemma~\ref{lem:poissonization}, the above bound will allow
us to reduce the problem of bounding the minimax risk $\calM$ to that of bounding
the risk $\calR$, at least whenever the scaling of these risks is of lower order
than $\sqrt{n/d} + \sqrt{mk/d}$. Furthermore, we will use Lemma~\ref{lem:kl_poisson_models}
directly in our development of lower bounds on the various minimax risks. 

\begin{remark} [Sharpness of Lemma~\ref{lem:kl_poisson_models}]
It is worth noting that the upper bound of Lemma~\ref{lem:kl_poisson_models}
can perhaps be improved quadratically, but not further. 
To elaborate, let us first note that, by Lemma~\ref{lem:dirichlet}, we may write
$$\bQ_c^{\otimes d} = \bbE_{\Pi,G}\left[ \bigotimes_{j=1}^d \mathrm{Poi}\big(n\textstyle \sum_{i=1}^k c_i G_i\pi_{ij}\big) \otimes 
\displaystyle \bigotimes_{i=1}^k \mathrm{Poi}(m\pi_{ij}) \right],$$
where $G = (G_1,\dots,G_k)^\top$ is a vector of independent $\mathrm{Gamma}(d,d)$-distributed random variables. 
Now, let us consider a proxy of this model and 
of the model $\widebar \bQ_c$ in which $k=1$,
and $m=0$. Let $Y$ and $\widebar Y$ be drawn from these corresponding models:
$$Y\sim \bbE_{\Pi,G}\left[ \bigotimes_{j=1}^d \mathrm{Poi}(nG\pi_{1j})\right],
\quad \widebar Y \sim \bbE_\Pi\left[ \bigotimes_{j=1}^d \mathrm{Poi}(n\pi_{1j})\right].$$
Define the random variables 
\begin{align*} 
S  &= \sum_{j=1}^d Y_j \sim \mathrm{NB}(d,(1+n/d)^{-1}),\quad \text{and,}\quad 
\widebar S  = \sum_{j=1}^d \widebar Y_j \sim \mathrm{Poi}(n)
\end{align*} 
where we obtained the law of $S$ 
by noting that $\sum_j Y_j \,|\,G \sim \mathrm{Poi}(nG)$, and 
any Poisson mixture with Gamma mixing measure has negative binomial distribution (cf. Lemma~\ref{lem:dirichlet}). 
One has 
$$\KL\big(\mathrm{Law}(\widebar Y)\, \big\|\, \mathrm{Law}(Y)\big) = \KL\big(\mathrm{Law}(\widebar S)\,\big\|\,\mathrm{Law}(S)\big) = \KL\big(\mathrm{NB}(d,(1+n/d)^{-1})\big) 
\,\big\|\, \mathrm{Poi}(n)\big).$$
Now, using the weak-lower semicontinuity
of the KL divergence, and
taking $d \to \infty$ such
that   $\tau = n/d$, one has the Gaussian approximation
\begin{align*}
\KL\big(\mathrm{NB}(d,(1+n/d)^{-1})\big) \,\big\|\, \mathrm{Poi}(n)\big)
 &\geq  \KL\big( N(d\tau, d\tau(1+\tau)), N(d\tau,d\tau)\big)  + o(1)\\
 &= \frac 1 2\log \frac 1 {1+\tau} + \frac{1+\tau}{2} - \frac 1 2 + o(1)\\
 &= \frac 1 2 \big( \tau - \log(1+\tau)\big) = O(\tau^2).
\end{align*}
This suggests that $\KL(\mathrm{Law}(\widebar Y)\, \big\|\, \mathrm{Law}(Y)) $, and hence
$\TV^2(\widebar \bQ_c\|\bQ_c^{\otimes d})$,    cannot scale faster than $\tau^2\asymp (n/d)^2$, 
  thus suggesting that, at least for $k=1$ and $m=0$, Lemma~\ref{lem:kl_poisson_models}
can only be improved quadratically, with no change in the trade-off between $n$ and $d$.
\end{remark} 

\begin{remark}[Mixtures of Products]
From a technical lens, our reduction
from model $\widetilde \bQ_c$
to model $\bQ_c^{\otimes d}$
will be fruitful since the former
is a mixture of product distributions,  
whereas the latter
is a product of mixture distributions,
which is significantly simpler
to handle. 
We refer to the recent manuscript~\cite{do2025moment}
for a more systematic comparison
of mixtures-of-products and products-of-mixtures
in discrete latent variable models.
\end{remark}
 
\subsection{Identifiability} 
Having reduced our problem to that of controlling the minimax
risk under the unnormalized Poisson model, we will now establish the identifiability
of that model, even in the absense of side information. 
\begin{lemma}
For all $c,\bar c \in\Delta_k$, the following assertions hold. 
\begin{enumerate}
    \item[(i)] Assume $m=0$. Then, $\bQ_c = \bQ_{\bar c}$ implies $W(c,\bar c) = 0$. 
    \item[(ii)] Assume $m\geq 1$. Then, $\bQ_c = \bQ_{\bar c}$ implies $c=\bar c$.
\end{enumerate}
\end{lemma}

\begin{proof} 
Mixtures
of products of Poisson measures
are identifiable in terms
of their mixing measures~\citep{teicher1967,barndorff-nielsen1965}.
To prove the claim, it
will therefore suffice to establish
the identifiability
of the collection $\{\mu_c:c \in\Delta_k\}$
of mixing measures with respect
to the parameter $c$, where
$$\mu_c := \mathrm{Law}(n\langle \varpi,c\rangle) \otimes \bigotimes_{i=1}^k \mathrm{Law}(m\varpi_i),
\quad c \in \Delta_k,$$
and $\varpi \sim \calE_d^{\otimes k}$.
The law of $\mu_c$ is uniquely characterized 
by its multivariate characteristic function, which is given by
\begin{align*}
\varphi_c(t,s)
 = \bbE_\varpi\left[ \exp\left( \bi t n\langle \varpi,c\rangle + \bi m \langle \varpi, s\rangle\right)
 \right],\quad t\in \bbR,s \in \bbR^k,
\end{align*}
with
$\bi = \sqrt{-1}$. 
One has
$$\varphi_c(t,s) = \prod_{i=1}^k \bbE\left[ \exp\big( \bi(nt c_i+ms_i)\varpi_i\big)\right]
 = \prod_{i=1}^k \frac{d}{d - \bi(ntc_i+ms_i)}.$$
The right-hand side of the above display defines
the reciprocal of a polynomial in $t$ and~$s$.
By setting $s=0$, this
polynomial depends only on 
the univariate parameter $t$, and is uniquely
characterized by its unordered collection
of roots, namely
$\{d / (\bi nc_i): 1 \leq i \leq k\}$. It
follows that $\varphi_c$, and hence $\mu_c$, 
is uniquely characterized by the 
unordered collection of entries
of $c$, and claim (i) readily follows from
this observation. Furthermore,
when $m \geq 1$, the equality
$\varphi_c(t,s) = \varphi_{\bar c}(t,s)$
can only hold uniformly over all
$(t,s) \in \bbR \times \bbR^k$
if $c =\bar c$, which proves part (ii).
\end{proof}   

With these various reductions in place, 
we are now in a position
to prove the main results
of this manuscript, beginning
with sample
complexity lower bounds.

\section{Proofs of Lower Bounds}  
\label{sec:lower_bounds}

The goal of this Appendix is to prove the lower bounds on the sample complexity
stated in Theorems~\ref{thm:main_unsorted}--\ref{thm:main_sorted}. 

\subsection{Preliminaries}
The bulk of our lower bound arguments will be contained in the following two Lemmas, which
characterize the divergence between elements of our model in terms of their parameter separation. 
In what follows,  for any given $k\geq 3$, $1 \leq s \leq k-2$, 
and $0 \leq \beta \leq  1/s$, we denote 
by $\Sigma_{k,s}(\beta)$ the set of 
  elements $c \in \Delta_k$ for which 
  exactly $s$ of the first $k-2$ entries of $c$ are equal to $\beta$, and 
  the last two entries of $c$ are given by $c_{k-1} = c_k = (1 - \beta s)/2$. 
Our first result deals with the regime where $m$ is polynomially smaller than $d$,
under the unnormalized Poisson model. 
 \begin{lemma}
 \label{lem:tv_bound}
 Let $n,m,d,k\geq 1$ satisfy the conditions $2\leq k\leq d$,
 $(n+m)^{1+\gamma}\leq d$, and $k^{1+\gamma} \leq d/m$, for an arbitrarily
 small constant $\gamma > 0$.  
Let $c,\bar c \in \Delta_{k}$ be two vectors such that
$$m_j(c) = m_j(\bar c),\quad \text{for all } j=0,\dots,k-1.$$
Furthermore,  let $q_1 := |m_k(\bar c) - m_k(c)|$ and $q_2 =  \|\bar c -c\|_2$. 
Assume $q_2 \geq 1/d$. 
Then, there exist constants $C_1 = C_1(\gamma,k) > 0$ and $C_2=C_2(\gamma) > 0$ such that
the following assertions~hold.
\begin{enumerate} 
\item We have,
\begin{align*}
\mathrm{TV}(\bQ_{\bar c}^{\otimes d}, \bQ_{c}^{\otimes d}) \leq C_1 n^{\frac k 2} d^{-\frac{k-1}{2}} q_1 
+ C_2 q_2\sqrt{\frac{nm}{d}}.
\end{align*}
\item Assume $W(c,\bar c) = 0$. 
Then,
\begin{align*}
\chi^2(\bQ_{\bar c}^{\otimes d}, \bQ_{c}^{\otimes d}) \leq  C_2 q_2^2 \frac{nm}{d}.
\end{align*}
\end{enumerate}
\end{lemma}
The proof of Lemma~\ref{lem:tv_bound} appears in Appendix~\ref{app:pf_lem_tv_bound}.
Our next result provides
an upper bound which
is effective in the regime $m \geq d$. In this case, we directly analyze the multinomial model. 
\begin{lemma}
\label{lem:multinomial_model_kl}
There exists a universal constant $C > 0$
such that if $3\leq k \leq d$, then 
for all $1 \leq s \leq k -2$, $0\leq \beta\leq s$,  $n,m \geq 1$, and   $c,\bar c \in \Sigma_{k,s}(\beta)$, 
$$\KL(\tilde \bQ_c\|\tilde \bQ_{\bar c}) \leq C n \|\bar c-c\|_2^2.$$
Furthermore, for all $k\geq 2$, there exists a constant $C_k > 0$ such that for all
$m,n,d\geq 1$ and $c,\bar c \in \Delta_k$ satisfying $\|c\|_\infty \vee \|\bar c\|_\infty \leq 3/4,$
$$\KL(\tilde \bQ_c\|\tilde \bQ_{\bar c}) \leq C_k n \|\bar c-c\|_2^2.$$
\end{lemma}
 The proof of Lemma~\ref{lem:multinomial_model_kl} appears in Appendix~\ref{app:pf_lem_multinomial_model_kl}. 
 A remarkable feature of Lemmas~\ref{lem:tv_bound}--\ref{lem:multinomial_model_kl}
 is the fact that they exhibit a quadratic scaling with respect to $q_2 = \|\bar c-c\|_2^2$,
 similarly to divergences between Gaussian
 location models. This reflects the fact that, despite the Poissonian
 nature of our problem, the heteroscedasticity of our observations is mild, due to
 assumption~\ref{assm:pt}.
 
Before proving these two results, let us show how they lead to our various minimax lower bounds.

\subsection{Minimax Lower Bound for the Sorted Loss Function}
Our aim in this section is to prove the following minimax lower bound.
\begin{proposition}
\label{prop:lower_bound_sorted}
Let $d,k,m,n\geq 1$ satisfy condition~\ref{assm:sample_size}. Then, 
there exists a constant $C_{k,\gamma} > 0$ such that 
$$  \calM_<(n,d,k,m)  \geq C_{k,\gamma}\left( \sqrt{\frac{d}{n(m+d^{1/k})}}
+ \frac 1 {\sqrt n}\right).$$
\end{proposition}

\noindent{\bf Proof of Proposition~\ref{prop:lower_bound_sorted}.}
By Le Cam's Lemma (cf.~\citep{polyanskiy2024}), it will suffice to show that there
exist parameters $c,\bar c \in \Delta_k$ such that 
\begin{align} 
\TV(\widetilde \bQ_c^{\otimes d} ,\widetilde \bQ_{\bar c}^{\otimes d})\leq 1/2, \quad \text{and}\quad 
W(c,\bar c) \gtrsim \sqrt{\frac{d}{n(m+d^{1/k})}} + \frac 1 {\sqrt n}.
\end{align}
Recall that, under condition~\ref{assm:sample_size}, we either have $m^{1+\gamma} \leq d$ or $m > d$.
Let us begin by proving the claim under the former condition. 
We will make use of the following Lemma  to obtain
least-favorable parameters $c,\bar c$. In what follows, recall that $m_p(u) = \frac 1 k \sum_{i=1}^k u_i^p$ for any 
$u\in \Delta_k$. 
\begin{lemma}
\label{lem:existence_two_point}
For any $k\geq 1$, there exists a constant $C_k > 0$ and vectors $u,v \in \bbR^k$
such that $\|u\|_\infty \vee \|v\|_\infty \leq 1$, and 
\begin{enumerate} 
\item[(i)] $m_1(u) = 0$. 
\item[(ii)] $m_p(u) = m_p(v)$ for all $p=1,\dots,k-1$. 
\item[(iii)]   $W(u,v)\wedge |m_k(u)-m_k(v)| \geq C_k$. 
\end{enumerate}
\end{lemma}
The proof of Lemma~\ref{lem:existence_two_point}
appears in Appendix~\ref{sec:pf_lem_existence_two_point}.
Now, given $\epsilon > 0$ and $u,v\in \bbR^k$ 
as in Lemma~\ref{lem:existence_two_point}, define
$$c = \left(\frac 1 k + \epsilon u_1,\dots,\frac 1 k + \epsilon u_k\right),\quad 
  \bar c = \left(\frac 1 k + \epsilon v_1,\dots,\frac 1 k + \epsilon v_k\right).$$
For all sufficiently small $\epsilon$, the conditions of Lemma~\ref{lem:existence_two_point} ensure
that $c,\bar c$ lie in the simplex~$\Delta_k$. Furthermore, we have 
$\|c-\bar c\|_2 \lesssim \epsilon$ 
and 
for all $p=1,\dots, k$, 
\begin{align*}
m_p(\bar c) - m_p(c) 
 &= \sum_{i=1}^k \left[\left(\frac 1 k + \epsilon u_i \right)^p - \left(\frac 1 k + \epsilon v_i\right)^p\right] \\
 &=  \sum_{i=1}^k \sum_{j=1}^p {p\choose j} k^{-(p-j)} \big(u_i^j - v_i^j\big)\epsilon^j  \\
 &=  \sum_{j=1}^p {p\choose j} k^{-(p-j)} \big(m_j(u)-m_j(v)\big)\epsilon^j  \\
 & = \epsilon^p\big(m_p(u) - m_p(v)\big),
\end{align*}
where we used property (ii) of Lemma~\ref{lem:existence_two_point}. 
Together with property (iii), we deduce that 
$$|m_p(u) - m_p(v)|\asymp \epsilon^p\cdot I(p=k),\quad \text{for all }p=1,\dots,k.$$ 
We may therefore apply Lemma~\ref{lem:tv_bound} with $q_1 \asymp \epsilon^k$ and $q_2 \asymp \epsilon$ to obtain
\begin{align*}
\mathrm{TV}(\bQ_{\bar c}^{\otimes d}, \bQ_{c}^{\otimes d}) \leq C_1 n^{\frac k 2} d^{-\frac{k-1}{2}} \epsilon^k
+ C_2 \epsilon\sqrt{\frac{nm}{d}},
\end{align*}
and thus, by Lemma~\ref{lem:kl_poisson_models}, we deduce the following bound for the normalized Poisson model:
\begin{align*}
\mathrm{TV}(\widebar \bQ_{\bar c}^{\otimes d}, \widebar \bQ_{c}^{\otimes d})
\lesssim n^{\frac k 2} d^{-\frac{k-1}{2}} \epsilon^k + \epsilon\sqrt{\frac{nm}{d}} + \sqrt{\frac n d} 
+ \sqrt{\frac {mk} d}. 
\end{align*}
Under condition~\ref{assm:sample_size}, 
the above quantity is bounded by $1/2$
if we  choose $\epsilon$ to be a sufficiently small multiple of $\sqrt{\frac{d}{n(d^{1/k} + m)}}$. 
On the other hand, Lemma~\ref{lem:existence_two_point} implies that 
$W(\bar c,c) = \epsilon W(u,v) \gtrsim \epsilon$. 
We thus deduce the following lower bound for the minimax risk in the normalized Poisson model, from Le Cam's Lemma:
$$\calR_<'(n,d,k,m) \gtrsim \sqrt{\frac{d}{n(d^{1/k} + m)}}.$$
Finally, let us deduce a lower bound for the multinomial model.
Notice that for any given $(n,d,k)$, the map $m\mapsto \calM_<(n,d,k,m)$
is monotonically decreasing, thus it suffices to prove the lower bound in the regime $m\geq d^{1/k}$. 
By combining the above display with Lemma~\ref{lem:poissonization}, we have
$$\calM_<(n,d,k,m) \gtrsim \calR_<'(2n,d,k,2m) - e^{-n/C} - e^{-m/C}
\gtrsim \sqrt{\frac{d}{n(d^{1/k} + m)}},$$
where the final inequality uses the fact that the terms $e^{-n/C}$
and $e^{-m/C}$ are both of low order when $m \geq d^{1/k}$.
This proves the claim in the regime $m^{1+\gamma} \leq d$. 
Finally, the claim for $m \geq d$ follows immediately from Le Cam's Lemma
together with Lemma~\ref{lem:multinomial_model_kl}.\qed

\subsection{Minimax Lower Bound for the Unsorted Loss Function}
\label{app:lb_unsorted_loss}
Our aim is now to derive lower bound  for the unsorted minimax risk. 
\begin{proposition}
\label{prop:lb_unsorted_fixed_k}
Assume that condition~\ref{assm:sample_size} holds with $k\geq 2$, and assume that $nm_d\geq d\log(k)$ with $m_d=\min\{m,d\}$.
Then, there exists a constant $C_\gamma > 0$ such that
$$\calM(n,d,k,m) \geq C_\gamma\cdot \min\left\{\left(\frac{d\log(k)}{nm_d}\right)^{\frac 1 4}, \left(\frac{dk}{nm_d}\right)^{\frac 1 2}\right\}. $$
\end{proposition} 
\begin{proof} 
Let us begin by proving the claim in the special case $k=2$. 
Since the unsorted minimax risk is bounded from below by the sorted minimax
risk, which in turn is always bounded from below by $1/\sqrt n$ (by Proposition~\ref{prop:lower_bound_sorted}),
it suffices to consider the regime $m^{1+\gamma} < d$, and to prove the lower bound
$$\calM(n,d,2,m) \gtrsim \sqrt{\frac{d}{nm}}.$$
By Lemma~\ref{lem:poissonization}, it further suffices to lower bound $\calR'(n,d,2,m)$, since
the conditions
$m^{1+\gamma}<d$
and $nm > d\log(k)$ imply that $n\wedge m \gtrsim d^\epsilon$ for some $\epsilon > 0$. 
To prove this lower bound, we can reason similarly as in the proof of Proposition~\ref{prop:lower_bound_sorted}.
Given $0\leq \epsilon,\delta \leq 1/4$, define the parameters
$$c  = \left(\frac 1 2 +\delta,\frac 1 2-\delta\right),\quad \bar c  = \left(\frac 1 2 + \delta + \epsilon, \frac 1 2 - \delta - \epsilon\right).$$
Notice that $\|c-\bar c\|_2 \asymp \epsilon$ and $|m_2(c)-m_2(\bar c)| \asymp \epsilon|\epsilon-\delta|$. 
Thus, by applying Lemmas~\ref{lem:kl_poisson_models} and~\ref{lem:tv_bound}
under condition~\ref{assm:sample_size}, as well as   Le Cam's Lemma, it suffices
to show that there exists a choice of $\epsilon,\delta$ such that
\begin{align*}
\frac n {\sqrt d} \epsilon|\epsilon-\delta|
+ \epsilon\sqrt{\frac{nm}{d}} = \epsilon \sqrt{\frac{n}{d}} \big( \sqrt n |\epsilon-\delta| + \sqrt m\big)\leq 1 / C,
\end{align*} 
for a sufficiently large constant $C > 0$. 
The above display is satisfied by choosing $\epsilon = c_0 \sqrt{d/nm}$
and $\delta = \epsilon + (1/4)\wedge \sqrt{m/n}$, 
for a sufficiently small constant $c_0 > 0$. This proves the claim for $k=2$.

Notice that the map $k\mapsto \calM(n,d,k,m)$  is monotonically increasing. Thus, in view of the preceding
lower bound for $k=2$, it suffices to assume that 
$k\geq 30$ in what follows. For this regime, we will  invoke Fano's Lemma~(cf.~\cite{polyanskiy2024}), 
a special case of which we recall next. 
\begin{lemma}[Fano's Lemma]\label{lem:fano}
Let $M\geq 2$ and let $c^{(1)},\dots,c^{(M)} \in \Delta_k$ be a collection of parameters satisfying 
$$\epsilon:=\min_{j\neq j'} \|c^{(j)}-c^{(j')}\|_2 > 0. $$
Let $J$ be a random variable uniformly-distributed over 
$\{1,\dots,M\}$, and let $\widetilde U$ be a random variable such that 
$ \widetilde U \,|\, J \sim \tilde \bQ_{c^{(J)}}$. 
Then, there exists a universal constant $C_1 > 0$ such that if
$I(\widetilde  U;J) \leq C_1\cdot \log M$, then 
$$\calM(n,d,k,m) \geq \epsilon / C_1.$$
\end{lemma}
In view of applying Fano's Lemma, let us begin by exhibiting
an $\epsilon$-packing of $\Delta_{k}$. Let $1 \leq s \leq k/10\leq k-2$ be an integer to be defined below. 
By the sparse Varshamov-Gilbert Lemma~(cf. Theorem 27.6 of~\cite{polyanskiy2024}), 
there exist an integer $M \geq 1$, a constant
$C > 0$, and bitstrings $\omega^{(1)},\dots,\omega^{(M)} \in \{0,1\}^{k-2}$ 
satisfying the following three properties (where $d_H$ denotes the Hamming distance):
\begin{enumerate}
    \item[(i)] $d_H(\omega^{(j)},\omega^{(j')}) \geq s/2,$
    for all $j\neq j'$. 
    \item[(ii)] $\log M \asymp s \log(k/s)$.
    \item[(iii)] $\|\omega^{(j)}\|_0 = s$ for all $j=1,\dots,M.$
\end{enumerate}
Given a constant $\beta \in [0,1/s]$ to be defined below, 
define for $j=1,\dots,M$ the vectors 
$$c_i^{(j)} = \beta \omega_i^{(j)},~ i=1,\dots,k-2,~~~\text{and}~~~ c_{k-1}^{(j)}=c_k^{(j)} = (1 - \beta s)/2,$$
which lie in the set $\Sigma_{k,s}(\beta)$
in view of condition (iii). 
With this choice, notice that
$$\epsilon:=\min_{j\neq j'} \|c^{(j)}-c^{(j')}\|_2
=\beta\cdot \min_{j\neq j'} \|\omega^{(j)} - \omega^{(j')}\|_2\asymp \beta \sqrt s, $$
by condition (i). 
Now, let $\widetilde U$ and $J$ be defined as in Lemma~\ref{lem:fano}.
We will bound their mutual information separately in the case
$m^{1+\gamma} \leq d$ and $m \geq d$, beginning with the former.
Define a random variable $U$ via $U\,|\, J \sim \bQ_{c^{(J)}}^{\otimes d}$. 
We will bound the mutual information $I(J;\widetilde U)$ by passing through
$I(J;U)$:
\begin{align*}
\big|  I(\widetilde U;J) - I(U;J)\big| \leq 
\big| H(J|U) - H(J|\widetilde U)\big| 
= \big| \bbE[H(P_{J|U}) - H(P_{J|\widetilde U})]\big| 
\leq (\log M)\cdot \TV(U,\widetilde U),
\end{align*}
where $P_{J|U}$ is the conditional law of $J$ given $U$, whose entropy is bounded
above by $\log M$.  Now, we simply have
$$ \TV(U,\tilde U)
\leq \frac {1} M \sum_{j=1}^M \TV(\bQ_{c^{(j)}}^{\otimes d}, \tilde \bQ_{c^{(j)}}^{\otimes d}) 
\lesssim \sqrt{n/d} + \sqrt{mk/d}=:\delta,$$
by Lemma~\ref{lem:kl_poisson_models}. Notice that $\delta$
vanishes under condition~\ref{assm:sample_size}. 
We now have
\begin{align*}
I(\widetilde U; J)
 \lesssim (\log M )\delta + I(U; J) 
 &\lesssim  
   (\log M )\delta + \frac 1 {M^2} \sum_{j,j'=1}^M \KL(\bQ_{c^{(j)}}^{\otimes d} \|  \bQ_{c^{(j')}}^{\otimes d}) \\ 
 &\lesssim (\log M) \delta +  \frac {nm} {dM^2} \sum_{j,j'=1}^M \|c^{(j)} - c^{(j')}\|_2^2,
\end{align*}
where the final inequality follows from Lemma~\ref{lem:tv_bound}(ii). 
On the other hand, when $m > d$, we may apply Lemma~\ref{lem:multinomial_model_kl} to directly
bound the mutual information by:
\begin{align*}
I(\widetilde U;J) \leq \frac 1 {M^2} \sum_{j,j'=1}^M \KL\big(\widetilde \bQ_{c^{(j)}} \| \widetilde \bQ _{c^{(j')}}\big)
\lesssim \frac{n}{M^2} \sum_{j,j'}^M \|c^{(j)} - c^{(j')}\|_2^2.  
\end{align*}
It follows that for all $m$ satisfying condition~\ref{assm:sample_size}, we have
\begin{align*}
I(\widetilde U;J) 
&\lesssim (\log M) \delta +  \frac{nm_d}{dM^2} \sum_{j,j'}^M \|c^{(j)} - c^{(j')}\|_2^2
\lesssim (\log M) \delta + \frac{nm_d\beta^2 s}{d}.  
\end{align*}
With this bound in place, let us apply Fano's Lemma, for which we
will use different choices of $\beta,s$ depending on the magnitude of $k$. 
If $k < \sqrt{nm_d/d}$, 
we may choose a small enough constant $C > 0$ such that if 
$s = \lfloor k/10\rfloor$ and $\beta = C \sqrt{d/nm_d}$, then $I(\widetilde U; J) \leq C_1 \log M$,
where $C_1$ is the constant appearing in the statement of Lemma~\ref{lem:fano}, and we
used property (ii) above. 
Thus, when $k < \sqrt{nm_d/d}$, we obtain the minimax lower bound
\begin{align}\label{eq:unsorted_calM_first} 
\calM(n,d,k,m) \gtrsim \beta \sqrt s = \sqrt{\frac{kd}{nm_d}}.
\end{align}  
Due to condition~\ref{assm:sample_size},
it remains to handle the
case $k \geq (\sqrt{nm_d/d})^{1+\gamma}$.
Pick $\beta = 1/s$, with 
$s$ the smallest integer satisfying
$$s\geq c_0\sqrt{\frac{nm_d}{d\log(dk/(nm_d))}},$$
for a sufficiently
small constant $c_0 > 0$ to be defined
below.
Then, 
$$I(\widetilde U;J) \lesssim \frac{nm_d}{ds}
\asymp s \log(dk/(nm_d))
\asymp s \log(k/s).
$$
Therefore, by property (ii) above, 
we have $I(\widetilde U;J) \leq C_1 \log M / 2$
provided $c_0$ is chosen sufficiently small.
Thus, by Fano's Lemma, one has
$$\calM(n,d,k,m) \gtrsim \beta \sqrt s \asymp 
\sqrt{\frac{d\log(k)}{n m_d}},$$
where we used the fact
that $\log(kd/(nm_d))\asymp \log(k)$
under the stated assumption on $k$.

Altogether, we have thus shown
that, under condition~\ref{assm:sample_size}
and $nm > d\cdot \log k$, we have
$$\calM(n,d,k,m) \gtrsim \min\left\{\left(\frac{d\log(k)}{nm_d}\right)^{\frac 1 4}, \left(\frac{dk}{nm_d}\right)^{\frac 1 2}\right\}. $$ 
This proves the claim.
\end{proof} 
   
\subsection{Proof of Lemma~\ref{lem:tv_bound}}\label{app:pf_lem_tv_bound}

Throughout the proof, let $\varpi = (\varpi_1,\dots,\varpi_k)$ denote a random vector with entries $$\varpi_i \overset{\mathrm{i.i.d.}}\sim \calE_d,
\quad i=1,\dots,k.$$ 
Recall that we write 
$$ \bQ_c = \bbE_\varpi\left[\mathrm{Poi} \left(n\langle c,\varpi\rangle\right)\otimes \bigotimes_{i=1}^k \mathrm{Poi}(m\varpi_i) \right],$$
and we denote the density of $\bQ_c$, defined over $I:= \bbN_0\times \bbN_0^k$, as
$$\bq_c(x,y) = \bbE_\varpi\left[ f(x;n\langle \varpi,c\rangle)
\prod_{i=1}^k f(y_i;m\varpi_i)\right],\quad (x,y)\in I,$$
where $f(\cdot;\lambda)$ is the $\mathrm{Poi}(\lambda)$ density. 
Fix once and for all the truncation parameter $t = (n+m)^{-\gamma_0} d^{\gamma_0-1}$
with $\gamma_0 = \gamma / (1+\gamma)$, and write
$$ \bQ_c^t = \bbE_\varpi\left[\mathrm{Poi} \left(n\langle c,\varpi^t\rangle\right)\otimes \bigotimes_{i=1}^k \mathrm{Poi}(m\varpi_i^t) \right],
~~ \text{where }
\varpi^t = \big(\varpi_1\wedge t, \dots, \varpi_k\wedge t\big),$$
with corresponding density denoted $\bq_c^t(x,y)$. 
The following Lemma reduces our problem to that of bounding
the $\chi^2$-divergence between the truncated measures.

\begin{lemma}
\label{lem:truncation_tv}
Let $1 \leq k \leq d$. Under assumption~\ref{assm:sample_size}, there exist constants 
$C_1,C_2,a > 0$ depending only on $\gamma$ such that the following assertions hold
for all $c,\bar c \in \Delta_k$.
\begin{enumerate}
\item We have, 
$$ \mathrm{TV}\big(\bQ_c^{\otimes d}, \bQ_{\bar c}^{\otimes d}\big) \leq C_1  
\Big( \sqrt{d \cdot \chi^2(\bQ_c^t,\bQ_{\bar c}^t)} + e^{-C_2 d^a }\Big).$$
\item 
If $W(c,\bar c) = 0$, then
$$ \KL\big(\bQ_c^{\otimes d}, \bQ_{\bar c}^{\otimes d}\big) \leq C_1  
\Big(  d \cdot \chi^2(\bQ_c^t,\bQ_{\bar c}^t) + e^{-C_2 d^a}\Big).$$
\end{enumerate}
\end{lemma}
The proof appears in Section~\ref{sec:pf_lem_truncation_tv}. 
Now, for the remainder of the proof, let $c,\bar c \in \Delta_k$ 
be any elements such that 
$$m_j(c) = m_j(\bar c),\quad j=1,\dots,k-1.$$
Notice that $q_1=0$ 
when $W(c,\bar c)=0$, 
thus, in view of Lemma~\ref{lem:truncation_tv}, both assertions of the claimed Lemma~\ref{lem:tv_bound} will
follow if we are able to prove the following upper bound
on the $\chi^2$-divergence between {\it truncated} distributions:
\begin{equation}\label{eq:goal_of_lb}
\chi^2(\bQ_c^t,\bQ_{\bar c}^t) 
\lesssim q_1^2\left(\frac n d\right)^k + \frac{nm q_2^2}{d^2}. 
\end{equation}
The remainder of the proof is devoted to deriving equation~\eqref{eq:goal_of_lb}. 
Our approach consists of expanding the $\chi^2$-divergence in terms of moment differences
of the mixing measures of $\bQ_c$ and $\bQ_{\bar c}$. 
Expansions of this type have been used for Gaussian mixture models
since the early work of~\cite{lepski1999,ingster2001}; 
for Poisson mixture models, related ideas have been used for instance by~\cite{han2018,wu2019}. 
Although these methods often scale
poorly for high-dimensional mixtures~\citep{schramm2022computational,han2024approximate},
our earlier truncation
step ensures that the relevant
mixing measures
have support near zero, and thus
have moments which decay exponentially in $k$. 

Define the following quantities:
\begin{alignat*}{2}
\widetilde U_c &= n\langle c,\varpi^t\rangle,\quad \widetilde V_i &&= m\varpi_i^t, 
\quad \lambda = \bbE[\varpi_1^t],\\
U_c &= \widetilde U_c - n\lambda, \quad V_i &&= \widetilde V_i - m\lambda,\quad i=1,\dots,k.
\end{alignat*}
Notice that the random variables $U_c$ and $V_i$ all have mean zero. 
We will make use of the following elementary bounds on $\lambda$.
\begin{lemma}
\label{lem:lambda_scaling}
It holds that
$$\frac 1 d \big(1 - td e^{-td}\big) \leq  \lambda  \leq \frac 1 d.$$
\end{lemma}
The proof appears in Section~\ref{app:pf_lem_lambda_scaling}.
In particular, by definition of $t$ and the fact that $(n+m)^{1+\gamma}\leq d$, we have $\lambda \asymp 1/d$. Now, notice that 
%
 \begin{align*}
\chi^2&(\bQ_{\bar c}^t, \bQ_{c}^t)
= \sum_{x,y_1,\dots,y_k=0}^\infty \frac{\left\{ \bbE_\varpi\left[\big(f(x; \widetilde U_{\bar c})- f(x; \widetilde U_{c})\big)\prod_{i=1}^k f(y_i; \widetilde V_i)  \right]\right\}^2}
                         {\bbE_\varpi \left[f(x; \widetilde U_{c})\prod_{i=1}^k f(y_i; \widetilde V_i) \right]}. 
\end{align*} 
We lower bound the denominator using the following. 
\begin{lemma}
\label{lem:chi2_lb} 
Assume condition~\ref{assm:sample_size}. Then, there exists a constant $C = C(\gamma) > 0$ such that for all~$(x,y) \in \bbN_0$,  
\begin{align*}
\bbE_\varpi & \left[f(x; \widetilde  U_{c})\prod_{i=1}^k f(y_i;\widetilde V_i) \right] \geq \frac 1 C f(x;n\lambda) \prod_{i=1}^k f(y_i;m\lambda).
\end{align*}
\end{lemma}
The proof appears in Appendix~\ref{app:pf_lem_chi2_lb}.
We may thus write, 
\begin{equation}
\label{eq:chi_sq_lb} 
\begin{aligned}
\chi^2&(\bQ_{\bar c}^t, \bQ_{c}^t)
\\ &\leq C \sum_{x,y_1,\dots,y_k=0}^\infty \frac{\left\{ \bbE_\varpi\left[\big(f(x; n\lambda + U_{\bar c})- f(x; n\lambda+  U_{c})\big)\prod_{i=1}^k f(y_i; m\lambda +  V_i)  \right]\right\}^2}
                         {f(x;n\lambda) \prod_{i=1}^k f(y_i;m\lambda)}. 
\end{aligned}
\end{equation} 
We will proceed by expanding the numerator of the above display in the basis of 
Charlier polynomials, whose definition and basic properties
are recalled in Appendix~\ref{app:charlier_polynomials}. Let 
$\{\varphi_\ell(\cdot;\lambda)\}_{\ell=0}^\infty$
denote the univariate family of Charlier polynomials with parameter $\lambda > 0$. 
Given $\blambda = (\lambda_0,\dots,\lambda_k) \in \bbR_+^{k+1}$, we define the following tensor-product family of Charlier polynomials
$$\varphi_{\alpha,\beta}(x,y;\blambda) = \varphi_\alpha(x; \lambda_0) \cdot \prod_{i=1}^k \varphi_{\beta_i}(y_i; \lambda_{i}),
\quad x\in \bbR, y\in \bbR^k,$$
for any multi-indices $(\alpha,\beta)\in I$, where $I := \bbN_0\times \bbN_0^{k}$.
In what follows, we show that they form an orthogonal basis with respect to the $L^2(g_{\blambda})$ inner product, 
where 
$$g_{\blambda}(x,y) = f(x;\lambda_0) \cdot \prod_{i=1}^k f(y_i;\lambda_i),\quad x\in \bbN_0,y\in \bbN_0^k.$$
\begin{lemma}
\label{lem:charlier_tensor}
The polynomial family $\{\varphi_{\alpha,\beta}(\cdot,\cdot;\blambda)\}_{(\alpha,\beta)\in I}^\infty$
with parameter $\blambda > 0$ is an orthogonal basis of $L^2(g_{\blambda})$, such that
$$\bbE_{(X,Y) \sim g_{\blambda}} \Big[\varphi_{\alpha,\beta}(X,Y;\blambda) \varphi_{\alpha',\beta'}(X,Y;\blambda)\Big]
 = \alpha! \lambda_0^\alpha \cdot \prod_{i=1}^k  \beta_i! \lambda_i^{\beta_i} 
 \cdot \bbI\big((\alpha,\beta)=(\alpha',\beta')\big),$$
 for any $(\alpha,\beta), (\alpha',\beta')\in I.$
Furthermore, one has the relation
$$g_{\blambda+\bu}(x,y) = g_{\blambda}(x,y) \sum_{(\alpha,\beta)\in I} \varphi_{\alpha,\beta}(x,y;\blambda) \frac{u_0^\alpha}{\alpha!\lambda_0^\alpha} \prod_{i=1}^k \frac{u_i^{\beta_i} }{\beta_i!\lambda_i^{\beta_i}},\quad
(x,y)\in I,$$
for any $\bu=(u_0,\dots,u_k)\in \bbR^{k+1}_+$ such that $\blambda+\bu$ has 
positive entries. 

\end{lemma}
The proof appears in Appendix~\ref{app:pf_lem_charlier_tensor}.  
Now, fixing $\blambda = (n\lambda,m\lambda,\dots,m\lambda)$, we deduce from Lemma~\ref{lem:charlier_tensor} that
\begin{align*}
  \bbE_\varpi&\left[\big(f(x; n\lambda + U_{\bar c})- f(x; n\lambda+  U_{c})\big)\prod_{i=1}^k f(y_i; m\lambda +  V_i) \right] \\
  &\hspace{1in} = g_{\blambda}(x,y) \sum_{(\alpha,\beta)\in I} \varphi_{\alpha,\beta}(x,y;\blambda) 
      \frac{\Delta_{\alpha,\beta}}{\alpha!\beta!(n\lambda)^\alpha (m\lambda)^{|\beta|}}, 
\end{align*}
where for any multi-indices  $(\alpha,\beta)\in I$, we write $|\beta|= \sum_i\beta_i$,
$\beta! = \beta_1!\cdots \beta_k!$, and 
$$\Delta_{\alpha,\beta} = \bbE_\varpi\Big[ (U_{\bar c}^\alpha-U_c^\alpha) V_1^{\beta_1} \cdots V_k^{\beta_k}\Big].$$
Thus, returning to equation~\eqref{eq:chi_sq_lb} , and 
using the orthogonality relation from Lemma~\ref{lem:charlier_tensor}, 
we arrive at 
\begin{align*}
\chi^2 (\bQ_{\bar c}, \bQ_{c})  
 &\leq C \sum_{(x,y)\in I} \left(\sum_{(\alpha,\beta)\in I} \varphi_{\alpha,\beta}(x,y;\blambda) 
      \frac{\Delta_{\alpha,\beta}}{\alpha!\beta!(n\lambda)^\alpha (m\lambda)^{|\beta|}}\right)^2 g_{\blambda}(x,y)  \\
 &=  C \sum_{(\alpha,\beta)\in I}  
      \frac{\Delta_{\alpha,\beta}^2}{\alpha!\beta!(n\lambda)^\alpha (m\lambda)^{|\beta|}} \\
       &=:C( S_1+S_2), 
\end{align*}
where
\begin{equation}
\label{eq:S1_S2}
S_1 = \sum_{\alpha=1}^\infty \frac{\Delta_\alpha^2}{\alpha!(n\lambda)^\alpha}, \quad 
  S_2 = \sum_{\alpha=1}^\infty \sum_{\substack{\beta \in \bbN_0^k \\ |\beta| \geq 1}}
    \frac{\Delta_{\alpha,\beta}^2}{\alpha!\beta! (n\lambda)^\alpha(m\lambda)^{|\beta|}},
\end{equation}    
    and where we abbreviate
$$\Delta_\alpha := \Delta_{\alpha,0} = \bbE_\varpi \big[ U_{\bar c}^\alpha - U_c^\alpha  \big].$$
We bound the terms $S_1$ and $S_2$ separately, beginning with the former. 
Recall that~$\kappa_\alpha(X)$ denotes the $\alpha$-th cumulant of a random variable $X$. 
By translational invariance of cumulants (except when $\alpha=1$), it holds that:
$$\kappa_\alpha(U_c) / n^\alpha = \kappa_\alpha( \langle c,\varpi^t\rangle) - \lambda \cdot I(\alpha=1),$$
for any $\alpha=1,2,\dots$.
Furthermore, since the elements of the vector~$\varpi$ are independent, we have
$$\kappa_\alpha(\langle c,\varpi^t\rangle)
= \sum_{i=1}^k c_i^\alpha \kappa_\alpha(\varpi_i^t)  =  \kappa_\alpha(\varpi_1^t) \cdot  m_\alpha(c),$$
where we recall that 
$m_\alpha(c)$ 
is the $\alpha$-th moment of the uniform distribution over $\{c_1,\dots,c_k\}$.
Since we assumed that $m_\alpha(c) = m_\alpha(\bar c)$ for all $1 \leq \alpha \leq k-1$, 
it must   follow from the preceding
two displays that
$\kappa_\alpha(U_c) = \kappa_\alpha(U_{\bar c})$ for all such $\alpha$, and we deduce that
\begin{align*}
S_1 \lesssim 
 \sum_{\alpha=k}^\infty \frac{\Delta_\alpha^2}{\alpha!(n\lambda)^\alpha}.
\end{align*}
Our aim is now to bound the remaining terms $\Delta_\alpha$, for $\alpha\geq k$.  
Notice that
\begin{align}
\label{eq:Delta_alpha_expansion} 
\nonumber 
\Delta_\alpha / n^\alpha 
 &= \bbE[(\langle \bar c,\varpi^t\rangle - \lambda)^\alpha] - \bbE[(\langle c,\varpi^t\rangle - \lambda)^\alpha]\\
 \nonumber 
 &= \sum_{j=k}^\alpha {\alpha \choose j} 
       \lambda^{\alpha-j}  \bbE\left[\langle \bar c,\varpi^t\rangle^j - \langle c,\varpi^t\rangle^j\right] \\
 &= \sum_{j=k}^\alpha {\alpha \choose j} 
       \lambda^{\alpha-j} (\bar \eta_j -\eta_j),
\end{align}
where we define the quantities
\begin{alignat*}{2}
\xi_j &= \kappa_j(\langle \varpi^t,c\rangle), &&\eta_j = \bbE[\langle \varpi^t,c\rangle^j], \\
\bar \xi_j &= \kappa_j(\langle \varpi^t,\bar c\rangle),\quad &&\bar \eta_j = \bbE[\langle \varpi^t,\bar c\rangle^j],
\quad j=1,2,\dots.
\end{alignat*}
By equation~\eqref{eq:bell_moments_cumulants} of Appendix~\ref{app:bell_polynomials}, the moments $\eta_j$ can be expressed in terms
of the cumulants $\xi_j$ via the following expansion in the Bell polynomial system:
\begin{align*}
\eta_j = \sum_{\ell=1}^j B_{j,\ell}\big(\xi_1,\dots,\xi_{j-\ell+1}\big) =
\sum_{\ell=1}^j  j! \sum_{(h_1,\dots,h_{j-\ell+1})\in \calH_{j,\ell}}
\prod_{i=1}^{j-\ell+1} \frac{\xi_i^{h_i}}{(i!)^{h_i} h_i!},
\end{align*}   
thus, for any $j=k,\dots,\alpha$, we have
\begin{align*}
\bar \eta_j-\eta_j
 &= \sum_{\ell=1}^j j! \sum_{(h_1,\dots,h_{j-\ell+1})\in \calH_{j,\ell}}
\bigg( \prod_{i=1}^{j-\ell+1} \frac{\xi_i^{h_i}}{(i!)^{h_i} h_i!} -
      \prod_{i=1}^{j-\ell+1} \frac{\bar\xi_i^{h_i}}{(i!)^{h_i} h_i!} \bigg).
\end{align*}
Let $K = K(j,\alpha) > 0$ denote a constant depending on $j,\alpha$, whose value may change from line to line. 
We have
\begin{align*}
|\bar \eta_j -\eta_j| \leq K \left|\sum_{\ell=1}^j  \sum_{(h_1,\dots,h_{j-\ell+1})\in \calH_{j,\ell}}
\bigg( \prod_{i=1}^{j-\ell+1}  \xi_i^{h_i}  -
      \prod_{i=1}^{j-\ell+1}  \bar\xi_i^{h_i}  \bigg)\right|.
\end{align*}
Now, recall that $m_i(c) = m_i(\bar c)$, and hence $\xi_i = \bar \xi_i$, for all $i=1,\dots,k-1$.
Thus 
\begin{align*}
|\bar\eta_j -\eta_j| &\leq K \left|\sum_{\ell=1}^j  \sum_{(h_1,\dots,h_{j-\ell+1})\in \calH_{j,\ell}}
\left(\prod_{i \leq k-1} \xi_i^{h_i}\right)
\left( \prod_{ i\geq k}\xi_i^{h_i} -
      \prod_{ i\geq k}\bar \xi_i^{h_i} \right)\right|,
\end{align*}
where all products are to be understood as ranging over all 
integers $1\leq i \leq j-\ell+1$ satisfying the stated conditions, 
with the convention that empty products equal~1. 
Now recalling that 
$\xi_i = \sum_{r=1}^k c_r^i \kappa_i(\varpi_1\wedge t)$ for all $i$,
due to the independence of the entries of $\varpi^t$, we 
have (cf.\,Appendix~\ref{app:bell_polynomials}):
\begin{align*}
\kappa_i(\varpi_1\wedge t)
 &\leq \sum_{\ell=1}^i (\ell-1)! B_{i,\ell}(\bbE[\varpi_1\wedge t],\dots,\bbE[(\varpi_1\wedge t)^{i-\ell+1}]) \\
 &\leq \sum_{\ell=1}^i (\ell-1)! B_{i,\ell}\left(\frac 1 d ,\dots ,\frac{(i-\ell+1)!}{d^{i-\ell+1}}\right) \\
 &\leq K\sum_{\ell=1}^i  B_{i,\ell}\left(\frac 1 d ,\dots ,\frac{1}{d^{i-\ell+1}}\right) \\
 &\leq K d^{-i}.
 \end{align*}
Thus, $\xi_i \leq Kd^{-i}$, and
similarly, $\bar \xi_i \leq K d^{-i}$. Using the definition of $\calH_{j,\ell}$, we
have $\sum_{i=1}^{j-\ell+1} ih_i = j$, thus
 \begin{align*}
|\bar\eta_j-\eta_j|
 &\leq K d^{-j}  \sum_{\ell=1}^j  \sum_{(h_1,\dots,h_{j-\ell+1})\in \calH_{j,\ell}}  
\left| \prod_{ i\geq k}(d^i\xi_i)^{h_i} -
      \prod_{ i \geq k}(d^i\bar \xi_i)^{h_i}\right|  \\
&\leq K d^{-j} \sum_{\ell=1}^j  \sum_{(h_1,\dots,h_{j-\ell+1})\in \calH_{j,\ell}}
\sum_{ i \geq k}\left|  (d^i\xi_i)^{h_i} -
       (d^i\bar \xi_i)^{h_i}\right|  \\
&\leq K d^{-j} \sum_{\ell=1}^j  \sum_{(h_1,\dots,h_{j-\ell+1})\in \calH_{j,\ell}}
\sum_{ i \geq k}d^i\left|  \xi_i  -
       \bar\xi_i \right|  \\
&\leq K  d^{-j}
\sum_{ i=k}^j d^i\left|  \xi_i  -
       \bar \xi_i \right|  \\
&= K   d^{-j}
\sum_{ i=k}^j d^i \left|  \kappa_i(\varpi_1\wedge t) \sum_{r=1}^k (c_r^i - \bar c_{r}^i)\right|  \leq  K d^{-j} q_1.
\end{align*} 
Returning to equation~\eqref{eq:Delta_alpha_expansion}, we deduce
that for any $\alpha \geq k$, there exists a constant $K_\alpha> 0$ such that
\begin{align*}
|\Delta_\alpha| 
 &\leq K_\alpha n^\alpha \sum_{j=k}^\alpha {\alpha \choose j} 
       \lambda^{\alpha-j} d^{-j}q_1  \leq K_\alpha q_1 \left(\frac n d\right)^\alpha. 
\end{align*} 
On the other hand, we also have the naive bound
\begin{align*}
|\Delta_\alpha| 
 &\leq n^\alpha  \sum_{j=k}^\alpha {\alpha \choose j} 
       \lambda^{\alpha-j} t^j \leq n^\alpha (t+\lambda)^\alpha \leq (2nt)^\alpha.
\end{align*}  
Thus, returning to equation~\eqref{eq:S1_S2}, 
we have shown that for all $p\geq k$, there exists $K_{ p} > 0$ such that
 \begin{align*}
S_1
 &\leq
K_{p}   q_1^2 \sum_{\alpha=k}^{p} \frac{(n/d)^{2\alpha}}{ \alpha!\lambda^\alpha} 
     + \sum_{\alpha=p+1}^\infty\frac{(2nt)^{2\alpha}}{ \alpha!\lambda^\alpha}  \\
 &\leq K_{p}  q_1^2   \sum_{\alpha=k}^{p} \frac{(n/d)^{\alpha}}{ \alpha!} 
     + \sum_{\alpha=p+1}^\infty\frac{(2n/d)^{\alpha/2}}{ \alpha!}  \\
 &\lesssim K_{p}   q_1^2 (2n/d)^{k} 
     +  (2n/d)^{p/2}.
\end{align*}
Under our conditions on $d$ and $n$, we can choose $p$ sufficiently large, 
as a function only of $\gamma$, such that
$(2n/d)^{p/2} \leq 1/d^5$, 
to obtain 
 \begin{align}
 \label{eq:final_S1_bound}
S_1 
 &\leq C q_1^2 ( n/d)^{k} + d^{-5},
\end{align}
for a constant $C = C(k,\gamma) >0$. 
We now turn to bounding the quantity 
\begin{align*}
  S_2 = \sum_{\alpha=1}^\infty \sum_{\substack{\beta \in \bbN_0^k \\ |\beta| \geq 1}}
    \frac{\Delta_{\alpha,\beta}^2}{\alpha!\beta! (n\lambda)^\alpha(m\lambda)^\beta }. 
\end{align*}
For any $\alpha,|\beta| \geq 1$, we have
\begin{align*}
|\Delta_{\alpha,\beta}|
 &\leq 
  n^\alpha m^{|\beta|} \bbE\left|\Big((\langle \varpi^t,\bar c\rangle-\lambda)^\alpha - (\langle \varpi^t,c\rangle - \lambda)^\alpha\Big)
 \prod_{i=1}^k (\varpi_i\wedge t-\lambda)^{\beta_i}\right|.
\end{align*}
By the mean value theorem, one has 
\begin{align*}
\big|(\langle &\varpi^t,\bar c\rangle-\lambda)^\alpha - 
 (\langle \varpi^t,c\rangle - \lambda)^\alpha \big| 
 \leq \alpha \big( |\langle \varpi^t,c\rangle-\lambda| + |\langle \varpi^t ,\bar c\rangle-\lambda|\big)^{\alpha-1}
 | \langle \bar c-c,\varpi^t\rangle|,
\end{align*} 
thus, 
$
|\Delta_{\alpha,\beta}|\leq  n^\alpha m^{|\beta|}
 \alpha   T_1^{1/2} (T_2T_3)^{1/4}, $
where 
\begin{align*}
T_1 &= \bbE |\langle \bar c-c,\varpi^t\rangle|^2 \\ 
T_2 &= \bbE \big( |\langle \varpi^t,c\rangle-\lambda| + |\langle \varpi^t ,\bar c\rangle-\lambda|\big)^{4(\alpha-1)} \\ 
T_3 &= \bbE \prod_{i=1}^k |\varpi_i\wedge t-\lambda|^{4\beta_i}.
\end{align*}
We have, 
\begin{align*}
T_1 = \Var[\langle \bar c-c,\varpi^t\rangle] = \sum_{i=1}^k (\bar c_i-c_i)^2 \Var[\varpi^t] \lesssim q_2^2 / d^2.
\end{align*}
To bound $T_2$, apply Jensen's inequality to obtain
\begin{align*}
\bbE |\langle \varpi^t ,\bar c\rangle-\lambda|^{4(\alpha-1)}
 = \bbE \bigg|\sum_{i=1}^k \bar c_i (\varpi_i^t - \lambda)\bigg|^{4(\alpha-1)} 
 \leq \sum_{i=1}^k \bar c_i \bbE |  \varpi_i^t - \lambda|^{4(\alpha-1)} 
 \leq  \frac{(4\alpha)!}{d^{4(\alpha-1)}} \wedge t^{4(\alpha-1)}.
\end{align*}
It follows that
$$T_2 \lesssim 2^{4(\alpha-1)} \left( \frac{(4\alpha)!}{d^{4(\alpha-1)}} \wedge t^{4(\alpha-1)}\right).$$
Next, we have, 
\begin{align*}
T_3 = \prod_{i=1}^k \bbE |\varpi_i\wedge t - \lambda|^{4\beta_i} 
 \lesssim \prod_{i=1}^k 2^{4\beta_i}\left(\frac{(4\beta_i)!}{d^{4\beta_i}} \wedge t^{4\beta_i}\right)
 \leq 2^{4|\beta|}  \left(  \frac{(4\beta)!}{d^{4|\beta|}}  \wedge t^{4|\beta|}\right)
\end{align*}
We thus obtain 
\begin{align*}
\frac{|\Delta_{\alpha,\beta}|}{ n^\alpha m^{|\beta|}}
 &\leq \alpha T_1^{1/2} (T_2T_3)^{1/4} \\ 
 &\leq \alpha \cdot \frac {q_2}{d} \cdot 2^{\alpha-1} \left( \frac{(4\alpha)!}{d^{(\alpha-1)}} \wedge t^{(\alpha-1)}\right) 
 \cdot 2^{|\beta|}  \left(  \frac{(4\beta)!}{d^{|\beta|}}  \wedge t^{|\beta|}\right) \\
 &\leq 2^{\alpha+|\beta|} q_2 \left( \frac{(4\alpha)!}{d^{\alpha}} \wedge \frac{t^{(\alpha-1)}}{d}\right) 
   \left(  \frac{(4\beta)!}{d^{|\beta|}}  \wedge t^{|\beta|}\right) \\
 &\leq 2^{\alpha+|\beta|} q_2 \left( \frac{(4\alpha)!}{d^{\alpha}} \wedge  t^{ \alpha } \right) 
   \left(  \frac{(4\beta)!}{d^{|\beta|}}  \wedge t^{|\beta|}\right) \\
 &\leq 2^{\alpha+|\beta|} q_2 \left( \frac{(4\alpha)!(4\beta)!}{d^{\alpha+|\beta|}} \wedge  t^{ \alpha +|\beta|} \right) .
\end{align*}  
It follows that for any fixed $\ell \geq 1$, there exists a constant $C_\ell > 0$ (which potentially grows
factorially in $\ell$) such that:
\begin{align*}
  S_2 &=  \sum_{\substack{\alpha\in \bbN_0,\beta \in \bbN_0^k \\ |\beta|\geq 1, \alpha+|\beta|< \ell}}
    \frac{\Delta_{\alpha,\beta}^2}{\alpha!\beta! (n\lambda)^\alpha(m\lambda)^{|\beta|}}+
    \sum_{\substack{\alpha\in \bbN_0,\beta \in \bbN_0^k \\ |\beta|\geq 1, \alpha+|\beta|\geq  \ell}}
    \frac{\Delta_{\alpha,\beta}^2}{\alpha!\beta! (n\lambda)^\alpha(m\lambda)^{|\beta|} } \\
 &\lesssim C_\ell  \sum_{\substack{\alpha\in \bbN_0,\beta \in \bbN_0^k \\ |\beta|\geq 1, \alpha+|\beta|< \ell}}
    \frac{(q_2 n^\alpha m^{|\beta|}/d^{\alpha+|\beta|})^2}{(n\lambda)^\alpha(m\lambda)^{|\beta|}}+
    \sum_{\substack{\alpha\in \bbN_0,\beta \in \bbN_0^k \\ |\beta|\geq 1, \alpha+|\beta|\geq \ell}}
    \frac{(q_2 n^\alpha m^{|\beta|}(2t)^{\alpha+|\beta|} )^2}{\alpha!\beta! (n\lambda)^\alpha(m\lambda)^{|\beta|} } \\
 &\lesssim  q_2^2\left\{ C_\ell  \frac{ n m}{d^2}+
    \sum_{\substack{\alpha\in \bbN_0,\beta \in \bbN_0^k \\ |\beta|\geq 1, \alpha+|\beta|\geq  \ell}}
    \frac{ (4t^2 d)^{(\alpha+|\beta|)} n^\alpha m^{|\beta|}}{\alpha!\beta!  } \right\} \\
 &\lesssim  q_2^2\left\{ C_\ell  \frac{ n m}{d^2}+
    \sum_{\alpha \geq \ell} \frac{(4t^2dn)^\alpha}{\alpha!} \cdot \bigg(\sum_{b  \geq  \ell}
    \frac{ (4t^2 dm)^{b} }{b!  }\bigg)^k \right\} \\
 &\leq  q_2^2\cdot C_\ell  \left\{ \frac{ n m}{d^2}+
    (4t^2dn)^\ell e^{4t^2dn} \cdot \big( (4t^2 dm)^{\ell} e^{4t^2dm}\big)^k\right\}.
 \end{align*}
 Notice that $t^2d(n+m) = \sqrt{(n+m)/d} =o(1)$, and 
 we have $e^{4t^2dm} \vee e^{4t^2dn} \leq C < \infty$, thus we obtain
 \begin{align*}
  S_2 &\lesssim  q_2^2\cdot C_\ell  \left\{ \frac{ n m}{d^2}+
    \left(\frac{4(n+m)}{d}\right)^{\ell/2}  \cdot \left(C \Big(\frac{4(n+m)}{d}\Big)^{\ell/2}  \right)^k\right\}.
 \end{align*}
 By choosing $\ell=5$, the second term in the above
 display is of lower order than the first for large enough $d$ (irrespective of the magnitude of $k$),
and we obtain $S_2 \lesssim q_2^2nm/d$. 
Combining this bound with equation~\eqref{eq:final_S1_bound}
and the fact that $q_2 \geq 1/d$, we have thus shown that
$$\chi^2(\bQ_{\bar c}^t, \bQ_{c}^t)  \leq C_1 q_1^2\left(\frac{n}{d}\right)^k + C_2\frac{q_2^2nm}{d^2},$$
where $C_1 = C_1(\gamma,k)$ and $C_2 = C_2(\gamma)$.
This proves equation~\eqref{eq:goal_of_lb}, and the claim follows.\qed  


\subsection{Proof of Lemma~\ref{lem:multinomial_model_kl}}
\label{app:pf_lem_multinomial_model_kl}
Let $c,\bar c \in \Delta_k$. 
Let $\Pi \in \bbR^{k\times d}$ be a random matrix whose rows are independently
drawn from the flat Dirichlet law $\calD_d$. Let 
$Y^c \in \bbR^d$ and $V\in \bbR^{k\times d}$ be drawn conditionally
independently according to
$$Y^c | \Pi \sim \mathrm{Mult}(n;\Pi^\top c), \quad 
  V_{i\cdot} | \Pi \sim \bigotimes_{i=1}^k \mathrm{Mult}(m;\Pi_{i\cdot}),\quad i=1,\dots,k,$$
so that $(Y^c,V) \sim \widetilde \bQ_c$. 
Given random variables $X,X'$ which are absolutely continuous
with respect to a common dominating measure, we denote by 
$p_{X,X'}$ their joint density, and by $p_X$
the marginal density of $X$. We also denote the Markov kernel of $X$ conditionally
on $X'$ by $p_{X|X'}$, whenever it exists. 

By the chain rule for the KL divergence, one has
\begin{align*}
\KL(\widetilde \bQ_c\|\widetilde \bQ_{\bar c})
 = \KL(f_{Y^c,V} \| f_{Y^{\bar c},V} ) =  \bbE_V\Big[\KL\big(f_{Y^c|V} \| f_{Y^{\bar c}|V}\big)\Big].
\end{align*}
Notice that the conditional law of $Y^{c}$ given $V$ is 
\begin{align}\label{eq:changing_law_of_Pi}
\nonumber 
f_{Y^c|V}(y|v)
 &= \frac 1 {f_V(v)} \int_{\bbR_+^k} f_{Y^c,V|\Pi}(y,v|\pi) d\calD_d^{\otimes k}(\pi) \\
\nonumber 
&= \frac 1 {f_V(v)} \int_{\bbR_+^k} f_{Y^c|\Pi}(y|\pi) \cdot f_{V|\Pi}(v|\pi) d\calD_d^{\otimes k}(\pi) \\ 
&=   \int_{\bbR_+^k} f_{Y^c|\Pi}(y|\pi) \cdot  f_{\Pi|V}(\pi|v)d\calD_d^{\otimes k}(\pi),
\end{align}
where we used Bayes' rule and the fact that the law of the rows of $\Pi$ are uniformly-distributed over the $d$-simplex. 
By conjugacy of the multinomial
 and Dirichlet distributions, notice that the conditional law of $\Pi$ given $V$ is given by
 $$\Pi_{i,\cdot}|V \sim \mathrm{Dirichlet}(1+V_{i1}, \dots, 1+V_{id}),\quad i=1,\dots,k.$$ 
Now, equation~\eqref{eq:changing_law_of_Pi}
shows 
that the conditional law of $Y^c$ given $V$ is simply given by the law 
$\bbE_\Pi[\mathrm{Poi}(n \Pi^\top c)\,|\,V]$, which we use as a shorthand
to denote the posterior distribution
of $\mathrm{Poi}(n \Lambda^\top c)$
when $\Lambda$ is drawn 
conditionally on $V$ from the Dirichlet 
law in the above display.
We thus have
\begin{align*}
\KL(\widetilde \bQ_c\|\widetilde \bQ_{\bar c})
&= \bbE_V\Big\{   \KL\Big(\bbE_\Pi[\mathrm{Poi}(n \Pi^\top c)\,|\,V] \,\|\, 
\bbE_\Pi[\mathrm{Poi}(n \Pi^\top \bar c)\,\big\|\,V] \Big) \Big\} \\
&\leq  \bbE_V\Big\{  \bbE_{\Pi} \big[\KL\big(\mathrm{Poi}(n \Pi^\top c) \,\|\, \mathrm{Poi}(n \Pi^\top \bar c) \big)\,\big|\, V \big]\Big\} \\
&=  \bbE_\Pi \Big\{ \KL\big(\mathrm{Poi}(n \Pi^\top c) \,\|\, \mathrm{Poi}(n \Pi^\top \bar c) \big)\Big\},
\end{align*}
where we used the convexity of the KL divergence in the second line.  Notice that
\begin{align*}
\bbE_\Pi\Big\{ \KL\big(f_{U^c_V|\Pi} \| f_{U^{\bar c}_V|\Pi} \big) \Big\} 
 &\lesssim  n \sum_{j=1}^d \bbE_\Pi\left[  \frac{(\Pi^\top (c-\bar c))_j^2}{(\Pi^\top \bar c)_j} \right]\\
 &\leq n \sum_{j=1}^d 
 \Big(\bbE_\Pi\left[   (\Pi^\top (c-\bar c))_j^6\right]\Big)^{\frac 1 3} 
 \bigg(\bbE\Big[  (\Pi^\top \bar c)_j^{-3/2} \Big]\bigg)^{\frac 2 3}. 
 \end{align*} 
 From here, we prove claims (i) and (ii) separately. 
 To prove claim (i), assume $c,\bar c \in \Sigma_{k,s}(\beta)$. 
Then, there exists two index sets $S_1,S_2 \subseteq \{1,\dots,k-2\}$
of cardinality $s$ such that 
$c_i = \beta\cdot I(i\in S_1)$ and $\bar c_i = \beta\cdot I(i\in S_2)$ for all $i=1,\dots,k-2$, and
$c_{k-1} = \bar c_{k-1} = c_k=\bar c_k = (1-\beta s)/2$. 
It follows that
$$\|c-\bar c\|_2^2 = \beta^2 \cdot \big(  |S_1\setminus S_2| + |S_2\setminus S_1|\big) = 2\beta^2 |S_1\setminus S_2|.$$ 
Now, notice that $(\Pi^\top \bar c)_j =  \beta  \sum_{i\in S_2} \pi_{ij} + \frac{1-s\beta}{2} (\pi_{kj}+\pi_{(k-1)j})$, 
where we recall that $\pi_{ij} \overset{iid}{\sim}\mathrm{Beta}(1,d-1)$
for all $i,j$, independently across $i$.
We will make use of the following.
\begin{lemma}\label{lem:sum_of_beta}
Let $L\geq 2$ and let $X_1,\dots,X_L \sim \mathrm{Beta}(1,d-1)$ be independent random variables. Then, 
there exists a universal constant $C > 0$ such that 
$$\bbE\Big[\big(\textstyle \sum_{i=1}^L X_i\big)^{-3/2}\Big] \leq \displaystyle C (d/L)^{3/2}.$$
\end{lemma}
The proof appears in Appendix~\ref{app:pf_lem_sum_of_beta}.
It follows that 
\begin{align*} 
\bbE\left[(\Pi^\top \bar c)_j^{-3/2}\right] 
 &\leq 
\min\left\{ \bbE\left[(\beta \textstyle  \sum_{i\in S_2} \pi_{ij})_j^{-3/2}\right], 
\bbE\left[(1-s\beta) (\pi_{kj} + \pi_{(k-1)j})/2\right]\right\} \\
 &\leq C 
d^{3/2} \cdot \min\left\{  (\beta s)^{-3/2} ,  (1-s\beta)/2^{5/2} \right\} \lesssim d^{3/2}.
\end{align*}

On the other hand, we have
\begin{align*}
\bbE_\Pi&\left[ (\Pi^\top (c-\bar c))_j^6\right] \\
  &= \beta^6\cdot \bbE\left[ \bigg(\sum_{i\in S_1\setminus S_2} \pi_{ij} 
  - \sum_{i\in S_2\setminus S_1} \pi_{ij} \bigg)^6
  \right] \\ 
&=  \beta^6\cdot \bbE\left[ \bigg(\sum_{i\in S_1\setminus S_2} (\pi_{ij} -\frac 1 d)
  - \sum_{i\in S_2\setminus S_1} (\pi_{ij}-\frac 1 d) \bigg)^6
  \right]  \\   
&\lesssim  \beta^6\cdot \bbE\left[ \bigg(\sum_{i\in S_1\setminus S_2} (\pi_{ij} -\frac 1 d)\bigg)^6\right] + 
 \beta^6\cdot \bbE\left[ \bigg(\sum_{i\in S_2\setminus S_1} (\pi_{ij} -\frac 1 d)\bigg)^6\right].
\end{align*}
Recall that $\bbE[\pi_{ij}] = 1/d$, and that the $\pi_{ij}$ are i.i.d. across
$i$, for any fixed $j$. 
Now, apply Rosenthal's inequalities~\citep{rosenthal1970,rosenthal1972} to obtain that for any $j=1,\dots,d$, 
\begin{align*} 
\bbE\left[ \bigg(\sum_{i\in S_1\setminus S_2} (\pi_{ij} -\frac 1 d)\bigg)^6\right] 
 &\lesssim |S_1\setminus S_2|^3 \cdot \Var^3[\pi_{11}]
+        |S_1\setminus S_2| \cdot \bbE|\pi_{11}|^6 \lesssim \frac{|S_1\setminus S_2|^3  }{d^6}.
\end{align*} 
After repeating a symmetric argument, we thus obtain 
\begin{align*}
\bbE_\Pi \left[ (\Pi^\top (c-\bar c))_j^6\right] 
 &\lesssim \beta^6 \frac{|S_1\setminus S_2|^3 + |S_2\setminus S_1|^3}{d^6} \asymp \|c-\bar c\|_2^6 / d^6.
\end{align*}
Altogether, we have thus shown
\begin{align*}
\bbE_\Pi\Big\{ \KL\big(f_{U^c_V|\Pi} \| f_{U^{\bar c}_V|\Pi} \big) \Big\} 
 &\leq n \sum_{j=1}^d 
  \Big(\bbE_\Pi\left[   (\Pi^\top (c-\bar c))_j^6\right]\Big)^{\frac 1 3} 
 \bigg(\bbE\Big[  (\Pi^\top \bar c)_j^{-3/2} \Big]\bigg)^{\frac 2 3}\\
 &\lesssim n d \cdot (\|c-\bar c\|_2^6 / d^6)^{-\frac 1 3}\cdot d \lesssim n\|c-\bar c\|_2^2.
 \end{align*} 
and the first claim follows.
To prove the second claim, notice that when $c,\bar c \in \Delta_k$ satisfy $\|c\|_\infty\vee \|\bar c\|_\infty \leq 3/4$, 
at least two entries of $\bar c$ are bounded from below by $3/(4k)$. Assuming without loss of generality
that these correspond to the first two entries of $\bar c$, one has 
\begin{align*}
 \bbE\Big[  (\Pi^\top \bar c)_j^{-3/2} \Big] 
 \gtrsim \bbE[(\pi_{11} + \pi_{21})^{-3/2}] \gtrsim d^{3/2},
 \end{align*}
 where the symbol `$\lesssim$' now hides constants depending on $k$, and 
 where we applied Lemma~\ref{lem:sum_of_beta}. Furthermore, one has
 $$\bbE\big[(\Pi^\top (c-\bar c)_j^6\big] 
 \lesssim  \bbE \big[\|\Pi_{\cdot j}\|_2^6\big] \cdot  \|c-\bar c\|_2^6
 \lesssim d^{-6}\|c-\bar c\|_2^2.$$
From here, the claim can be deduced as before.
\qed

\section{Proofs of Upper Bounds}
\label{sec:ub_proofs} 
The goal of this section is to prove Propositions~\ref{prop:collision} and~\ref{prop:ub_blind}.
 
\subsection{Proof of Proposition~\ref{prop:collision}}
\label{app:pf_prop_collision} 
We prove Proposition~\ref{prop:collision}
in two steps. We begin by showing that the
unregularized collision estimator 
$$
\hat c_i^\mathrm{coll} = \frac{d+1}{nm} \sum_{\ell=1}^n \sum_{r=1}^m I(Z_\ell=W_{ir})-1,\quad i=1,\dots,k,$$
already achieves the optimal convergence rate when $k < \sqrt{nm_d / d}$. 
\begin{lemma}\label{lem:collision}
Under the multinomial model, there exists
a universal constant $C > 0$ such that
for all $1 \leq n\leq d$
and all $m,k \geq 1$,
$$\bbE\|\hat c-c\|_2 \leq C\sqrt{\frac{dk}{nm_d}},$$
with $m_d = \min\{m,d\}$.
\end{lemma}
Second, we will prove that for   appropriate
$\lambda > 0$, the hard-thresholded collision estimator
achieves the optimal convergence rate when $k > \sqrt{nm_d / d}$. We will analyze
this estimator under the normalized
Poisson model, which will allow
us to deduce the following upper bound. 
\begin{lemma}
\label{lem:hard_thresh_collision}
Under the multinomial model, there exists
a universal constant $C > 0$ such that
for any $1 \leq n,k\leq d$ and $m \geq 1$,
$$
\calM(n,d,k,m) \leq C\left(\frac{d\log k}{nm_d}\right)^{\frac 1 4},$$
with $m_d = \min\{m,d\}$.
\end{lemma}
Let us now prove these Lemmas in turn.
\subsubsection{Proof of Lemma~\ref{lem:collision}} 
Throughout the proof, we repeatedly
make use of elementary moment bounds for Dirichlet random variables, as summarized
in Lemma~\ref{lem:dirichlet}. 
Abbreviate $\hat c^\mathrm{coll}$ by $\hat c$. 
It suffices to show that for each $i$, $\bbE|\hat c_i - c_i| \lesssim \sqrt{d/nm_d}$.
Without loss of generality, fix $i=1$. 
Recall that the collections of random variables $\{Z_\ell\}_\ell$ and $\{W_{1j}\}_j$ are 
each exchangeable, and  independent of each other conditionally on $\Pi$. Thus, 
$$\bbE[\hat c_1] = (d+1) \cdot \bbP(Z_1 = W_{11})-1.$$
Furthermore, 
\begin{align*}
\bbP(Z_1 = W_{11}) 
 &= \sum_{j=1}^d \bbE \big[\bbP(Z_1 = z_j,W_{11}=z_j|\Pi)\big] \\ 
 &= \sum_{j=1}^d \bbE \big[\bbP(Z_1 = z_j|\Pi)  \cdot  \bbP(W_{11}=z_j|\Pi)\big] \\ 
 &= \sum_{j=1}^d \bbE \left[ \sum_{i=1}^k c_i \pi_{ij} \pi_{1j}\right] \\
 &= \sum_{j=1}^d \left[ c_1 \frac{2}{d(d+1)} + \sum_{i=2}^k c_i  \frac{1}{d(d+1)}\right] \\
 &= \frac 1 {d(d+1)} \sum_{j=1}^d \left[ 1 + c_1\right] 
 = \frac {1 + c_1} {d+1}.
\end{align*} 
Thus, $\hat c_1$ is unbiased. 
To bound its variance,
notice that
\begin{align*}
\bbE &\left[\left(\sum_{\ell=1}^n \sum_{r=1}^m I(Z_\ell = W_{1r})\right)^2\right] \\
 &= \sum_{\ell,\ell'=1}^n \sum_{r,r'=1}^m \bbE \big[ I(Z_{\ell} = W_{1r})I(Z_{\ell'} = W_{1r'})\big] \\
 &= \sum_{\ell=1}^n \sum_{r=1}^m \bbE \big[ I(Z_{\ell} = Z_{1r})\big] 
 + \sum_{\ell=1}^n \sum_{r\neq r'} \bbE \big[ I(Z_{\ell} = W_{1r}= W_{1r'})\big] \\ 
 &\qquad + \sum_{\ell\neq \ell'} \sum_{r=1}^m \bbE \big[ I(Z_{\ell} = Z_{\ell'} = Z_{1r})\big] 
 + \sum_{\ell\neq \ell'} \sum_{r\neq r'} \bbE \big[ I(Z_{\ell} = W_{1r}) I(Z_{\ell'} = W_{1r'})\big] \\ 
 &=: (I)+(II)+(III)+(IV). 
\end{align*}
We compute these terms in turn. First, our earlier bias calculations imply that
$$(I) = nm\frac {1 + c_1} {d+1}\lesssim \frac{nm}{d}.$$
To compute term $(II)$, notice that 
\begin{align*}
\bbP(W_1 = Z_{11}= Z_{12})
 &= \sum_{j=1}^d \bbE\big[ \bbP(W_1=z_j|\Pi) \bbP(Z_{11}=z_j|\Pi) \bbP(Z_{12}=z_j|\Pi)\big]  \\
 &= \sum_{j=1}^d  \sum_{i=1}^k c_i \bbE\big[ \pi_{ij} \pi_{1j}^2\big]  
 = \sum_{j=1}^d \left\{c_1\bbE\big[ \pi_{1j}^3\big]+ \sum_{i=2}^k c_i  \bbE\big[\pi_{ij} \pi_{1j}^2\big] \right\} \lesssim 1/d^2,
\end{align*}
thus, 
$$(II) \lesssim \frac{nm(m-1)}{d^2} \lesssim \frac{nm^2}{d^2}.$$

A similar argument shows that 
$(III) \lesssim  mn^2/{d^2}.$
To bound $(IV)$, notice that  
\begin{align*}
\bbE&\big[ I(W_{1} = Z_{11}) I(W_{2} = Z_{12})\big]  \\
 &= \bbE\Big[ \bbP(W_1=Z_{11}|\Pi) \cdot \bbP(W_2 = Z_{12}|\Pi) \Big] \\
 &=\sum_{j,j'=1}^d \sum_{i,i'} c_ic_{i'} \bbE[ \pi_{ij} \pi_{1j} \pi_{i'j'} \pi_{1j'}]  
= (a) + (b),
\end{align*}
where
$$(a) = \sum_{j=1}^d \sum_{i,i'} c_ic_{i'} \bbE[\pi_{ij}\pi_{1j} \pi_{i'j}\pi_{1j}] \leq C_k' d^{-3},$$
for a sufficiently large constant $C' > 0$, 
and 
\begin{align*}
(b)&= \sum_{j\neq j'} \sum_{i,i'} c_ic_{i'} \bbE[ \pi_{ij} \pi_{1j} \pi_{i'j'} \pi_{1j'}]   \\
   &= \sum_{j\neq j'}\left(\sum_{i=1}^k c_i \bbE[\pi_{ij}\pi_{1j}]\right)^2 + 
      \sum_{j\neq j'} \sum_{i,i'} c_ic_{i'} \Cov( \pi_{ij} \pi_{1j}, \pi_{i'j'} \pi_{1j'}) \\ 
   &\leq \sum_{j\neq j'}\left(\sum_{i=1}^k c_i \bbE[\pi_{ij}\pi_{1j}]\right)^2 + C'' d^{-3}\\ 
    &= d(d-1) \left(\frac{1+c_1}{d(d+1)}\right)^2  + C'' d^{-3} \leq \left(\frac{1+c_1}{d+1}\right)^2 + C'' d^{-3},
\end{align*}
for another universal constant $C'' > 0$. 
We thus have \begin{align*}
(IV) = n(n-1)m(m-1)\big[(a) + (b)\big] \leq (nm)^2 \left[\frac {C'''} {d^3} +    \left(\frac{1+c_1}{d+1}\right)^2\right],
\end{align*}
for $C_k'''=C_k'+C_k''$. 
Altogether, we deduce that
\begin{align*}
\Var &\left[ \sum_{\ell=1}^n \sum_{j=1}^m I(W_\ell = Z_{1j}) \right] \\
 &= (I)+(II)+(III)+(IV)   -\textstyle  \left[\bbE  \Big(\sum_{\ell=1}^n \sum_{j=1}^m I(W_\ell = Z_{1j}) \Big)\right]^2\\
 &= (I)+(II)+(III)+(IV) - (nm)^2\left(\frac{1+c_1}{d+1}\right)^2 \\
 &\lesssim \frac{nm}{d} + \frac{nm^2}{d^2} + \frac{mn^2}{d^2} + \frac{(nm)^2}{d^3}.
\end{align*}
It follows that 
\begin{align*}
\Var[\hat c_1] &\lesssim \frac{d^2}{(nm)^2} \Var \left[\left(\sum_{\ell=1}^n \sum_{j=1}^m I(W_\ell = Z_{ij})\right)^2\right] 
 \lesssim \frac d {nm} + \frac{1}{n} + \frac 1 {m} + \frac 1 {d}.
\end{align*}
Since $n \leq d$, the claim follows.\qed 

\subsubsection{Proof of Lemma~\ref{lem:hard_thresh_collision}}
By condition~\ref{assm:sample_size}
and Lemma~\ref{lem:poissonization}, 
it will suffice to derive
an upper bound on the 
normalized Poisson minimax risk $\calR'(n,d,k,m)$.
Recall that, under
this model,  one
observes histograms $(Y,V)$ with 
entries that are conditionally independent
given $\Pi$. 
We will analyze the 
associated hard-thresholded
collision estimator, defined by
$$\tilde c_i = \hat c_i \cdot I(\hat c_i \geq \lambda), 
\quad \text{where } \hat c_i =  \frac d {nm} {(VY)_i} -1 =
 \frac d {nm}\sum_{j=1}^d V_{i j}Y_j-1,\quad i=1,\dots,k,$$
 where $\lambda  =\sqrt{a d \log(k) / (nm_d)}$, for a constant $a > 0$ 
 depending only on $\gamma$, to be specified below. 
Throughout the proof, $C > 0$ denotes
a universal constant whose value may change from one expression to the next.  
Our starting point is 
the following 
basic inequality for hard-thresholding estimators~\citep{donoho1994ideal}:
$$(\tilde c_i-c_i)^2 \leq 4(c_i \wedge \lambda)^2 + |\hat c_i-c_i|^2 I(|\hat c_i-c_i|>\lambda / 2) ,
\quad i=1,\dots,k.$$
Notice that
\begin{align*} 
\sum_{i=1}^k c_i^2\wedge \lambda^2
= \sum_{i: c_i \leq \lambda} c_i^2
 + \sum_{i:c_i > \lambda} \lambda^2
 \leq \lambda\|c\|_1
 + |\{i:c_i > \lambda\}| \lambda^2
 \lesssim \lambda\asymp \sqrt{\frac{da\log(k)}{nm_d}},
\end{align*} 
thus, to prove the claim, it will suffice to 
show that $\bbE[R] \lesssim \sqrt{da\log(k)/nm_d}$, where 
$$R = \max_{1 \leq i \leq k} R_i, \quad R_i =  k \cdot\bbE\Big[ |\hat c_i-c_i|^2
 \cdot I(|\hat c_i-c_i|>\lambda/2)\,\big|\, \Pi\Big].$$
Let us now derive
a concentration bound for $\hat c_i$ 
conditionally on $\Pi$. 
We will repeatedly
make use of the following
concentration bounds for  sub-Weibull random variables,
which are due to~\citet{kuchibhotla2022}. 
\begin{lemma}[Concentration of sub-Weibull Random Variables]
\label{lem:sub_weibull_conc}
Let $X_1,\dots,X_d$ be independent
mean-zero random variables 
in $\bbR$, such that for some $\alpha \in (0,1]$ and 
$\zeta > 0$, 
$$\max_{1 \leq j \leq d} \|X_j\|_{\psi_\alpha} \leq \zeta,
\quad \text{and define } \sigma^2 =  \sum_{j=1}^d \bbE[X_i^2].$$
Then, there exists a constant $C_\alpha > 0$
such that the following assertions
hold. 
\begin{enumerate}
    \item (Theorem 3.1,~\cite{kuchibhotla2022})
    For all $t > 0$, it holds with
    probability at least $1-2e^{-t}$
    that
    $$  \bigg| \sum_{j=1}^d X_j\bigg| 
    \leq  C_{\alpha} \zeta \big( \sqrt{dt}
    + t^{1/\alpha}\big).$$
    \item (Theorem 3.4,~\cite{kuchibhotla2022})
    For all $t > 0$, 
    it holds with probability at least 
    $1-3e^{-t}$ that
    $$\bigg| \sum_{j=1}^d X_j\bigg|
    \leq C_\alpha \left(\sqrt{\sigma^2 t} + \zeta (t\log d)^{1/\alpha}\right).$$
\end{enumerate}
\end{lemma}
We will also
make use of the following bound from~\citep{kuchibhotla2022}, which
characterizes the tail behavior
of products of sub-Weibull random
variables.
\begin{lemma}[\cite{kuchibhotla2022},~Proposition D.2]
\label{lem:holder_orlicz}
If $X_1,\dots,X_d$ are (possibly
dependent) random variables
satisfying $\|X_j\|_{\psi_{\alpha_j}} < \infty$ for some 
$\alpha_j > 0 $, then
$$\left\|\prod_{j=1}^d X_j\right\|_{\psi_\beta} \leq \prod_{j=1}^d \|X_j\|_{\psi_{\alpha_j}},
\quad \text{where } \frac 1 \beta  = \sum_{j=1}^d \frac 1 {\alpha_j}.$$
\end{lemma}
Now, using Proposition~6.5 of~\citep{leskela2025sub}, 
it can
be shown that the $\psi_1$-Orlicz
norm of a Poisson random variable 
$X \sim \mathrm{Poi}(\lambda)$
satisfies $\|X-\lambda\|_{\psi_1}\lesssim 1 \vee \sqrt \lambda$. Thus, denoting by $\|\cdot\|_{\psi_{\alpha,\Pi}}$
the Orlicz norm taken with
respect to the conditional law $\bbP(\cdot |\Pi)$,
one   has 
\begin{align*}
\big\|Y_j - \bbE[Y_j|\Pi]\big\|_{\psi_1,\Pi}
\lesssim 1 \vee \mu_{Y_j},\quad 
\big\|V_{ij}- \bbE[V_{ij}|\Pi]\big\|_{\psi_1,\Pi}
\lesssim 1 \vee \mu_{V_{ij}},
\end{align*}
with $\mu_{Y_j}=n\Pi_{\cdot j}^\top c$ 
and $\mu_{V_{ij}} = m \pi_{ij}$, 
  hence, 
\begin{align*} 
\big\|V_{ij}&Y_j - \bbE[V_{ij}Y_j|\Pi] \big\|_{\psi_{1/2,\Pi}} \\
 &\lesssim 
\big\|(V_{ij}-\bbE[V_{ij}|\Pi])(Y_j-\bbE[Y_j|\Pi])\big\|_{\psi_{1/2},\Pi} 
+ \mu_{Y_j} \|V_{ij}\|_{\psi_{1/2,\Pi}} 
+ \mu_{V_{ij}} \|Y_j\|_{\psi_{1/2,\Pi}} + 
\mu_{Y_j}\mu_{V_{ij}}\\
 &\lesssim  
 (1 \vee \mu_{Y_j})(1\vee \mu_{V_{ij}}),
\end{align*}
 where we   used Lemma~\ref{lem:holder_orlicz}
 on the final line. We write
 $$\zeta_{\Pi,i} = \max_{1 \leq j \leq d} (1\vee \mu_{Y_j})(1\vee \mu_{V_{ij}}),
 \quad i=1,\dots,k.$$
We may then apply the concentration 
inequality of Lemma~\ref{lem:sub_weibull_conc}(ii)
for $\psi_{1/2}$-random variables to 
obtain
\begin{align}
\label{eq:bernstein}
\bbP\left(\bigg|\sum_{j=1}^d \big(  
V_{ij}Y_j - \bbE(V_{ij}Y_j|\Pi)\big)\bigg|
> x \,\bigg|\,\Pi\right)
 \leq 2 \exp\left\{ -
 \frac 1 C \bigg(\frac{x^2}{\sigma_{\Pi,i}^2}
 \wedge  \frac{ 
 \sqrt {x/\zeta_{\Pi,i}} } 
 {\log d}\bigg) \right\},
\end{align}
for all $x > 0$, where 
$$\sigma_{\Pi,i}^2 = 
\sum_{j=1}^d\Var[V_{ij} Y_j|\Pi].$$
It follows that
\begin{align}
\bbP\big(
|\hat c_i - \bbE[\hat c_i\,|\,\Pi]|
> x \,|\,\Pi\big)
 \leq 2 \exp\left\{ -
 \frac 1 C \bigg(\frac{(nm)^2x^2}{d^2\sigma_{\Pi,i}^2}
 \wedge  \frac{ 
 1 } 
 {\log d}\sqrt{\frac{nm x}{d\zeta_{\Pi,i}}}\bigg) \right\}.
\end{align}
In particular, denoting by $\beta_{\Pi,i} = |\bbE[\hat c_i|\Pi] - c_i|$ the conditional bias
of $\hat c_i$, we deduce that
\begin{align}\label{eq:bernstein_second}
\bbP\big(
|\hat c_i - c_i  |
> x \,|\,\Pi\big)
 \leq 2 \exp\left\{ -
 \frac 1 C \bigg(\frac{(nm)^2(x-\beta_{\Pi,i})_+^2}{d^2\sigma_{\Pi,i}^2}
 \wedge  \frac{ 
 1 } 
 {\log d}\sqrt {\frac{nm (x-\beta_{\Pi,i})_+}{d\zeta_{\Pi,i}}}\bigg) \right\}.
\end{align}
where the transition between the sub-Gaussian and sub-Weibull tail occurs when the point $x$
exceeds 
$\tau_{\Pi,i} := (d/nm)(\sigma_{\Pi,i}^2 / \sqrt{\zeta_{\Pi,i}})^{2/3}
+\beta_{\Pi,i}$. 
With these preliminaries in place, we obtain 
\begin{align*}
R_i 
 &= \int_{\lambda^2/4}^\infty \bbP\big( |\hat c_i - c_i|^2 > u\,|\,\Pi\big)du \\ 
&\lesssim  k\cdot  
 \int_{\lambda^2/4}^\infty 
  \exp\left\{ -
 \frac 1 {C_0} \bigg(\frac{(nm)^2(\sqrt x-\beta_{\Pi,i})_+^2}{d^2\sigma_{\Pi,i}^2}
 \wedge  \frac{ 
 1 } 
 {\log d}\sqrt {\frac{nm (\sqrt x-\beta_{\Pi,i})_+}{d\zeta_{\Pi,i}}}\bigg) \right\} dx,
\end{align*}
for a sufficiently large constant $C_0 > 0$. 
Now, define 
for all $i=1,\dots,k$ the quantities 
\begin{align*} 
\lambda_{\Pi,i} &= (\lambda / 2-\beta_{\Pi,i})_+ \\
K_{\Pi,i} &=  k 
 \exp\left\{ - \frac{(nm)^2\lambda_{\Pi,i}^2}{C_0 d^2\sigma_{\Pi,i}^2}\right\}
 + k \exp\left\{ - \frac{ 
 1 } 
 {C_0 \log d}\sqrt {\frac{nm \lambda_{\Pi,i}}{d\zeta_{\Pi,i}}} \right\} .
 \end{align*}
In particular, notice that $K_{\Pi,i}$ solves
\begin{align*}
k \leq  K_{\Pi,i} \cdot 
 \exp\left\{ -
 \frac 1 {C_0} \bigg(\frac{(nm)^2(\sqrt x-\beta_{\Pi,i})_+^2}{d^2\sigma_{\Pi,i}^2}
 \wedge  \frac{ 
 1 } 
 {\log d}\sqrt {\frac{nm (\sqrt x-\beta_{\Pi,i})_+}{d\zeta_{\Pi,i}}}\bigg) \right\}
 \quad \text{ for all } x > \lambda^2/4,
\end{align*} 
thus we obtain
\begin{align*}
R_i
 &\leq K_{\Pi,i} 
 \int_{\lambda^2/4}^{\tau_{\Pi,i}^2}
\exp\left\{ - \frac{(nm)^2(\sqrt x - \beta_{\Pi,i})_+^2}{C_0 d^2\sigma_{\Pi,i}^2}\right\}
 dx
 + K_{\Pi,i}
 \int_{\tau_{\Pi,i}^2}^{\infty} 
   \exp\left\{ - \frac{ 
 1 } 
 {C_0}\sqrt {\frac{nm (\sqrt x - \beta_{\Pi,i})_+}{d (\log d)^2\zeta_\Pi}} \right\} dx \\
  &\leq K_{\Pi,i} \left\{
  (\beta_{\Pi,i}^2-\frac{\lambda^2}{4})_+
  + \frac{d^2\sigma_{\Pi,i}^2}{(nm)^2}
  +  \beta_{\Pi,i} \frac{d\sigma_{\Pi,i}}{nm}
 \right\}  + K_{\Pi,i}\left\{ 
 \frac{d^2\,(\log d)^4\,\zeta_{\Pi,i}^2}{(nm)^2}+
\beta_{\Pi,i}\,
\frac{d\,(\log d)^2\,\zeta_{\Pi,i}}{n m}
\right\}
\end{align*}
where we used the following elementary integral identities, which hold for all $a,b,f > 0$, 
\begin{align*}
\int_a^\infty e^{-b(\sqrt x-f)_+^2}dx &\lesssim (f^2-a)_+ + \frac 1 {b} + \frac f {\sqrt b}, \\
\int_a^\infty 
e^{-b (\sqrt x - f)_+^{1/2}}dx 
&\lesssim (f^2-a)_+ + \frac 1 {b^4} + 
\frac f {b^2}.
\end{align*}
Thus, if we define the event 
\begin{equation} 
\label{eq:sets_A} 
A = \bigcap_{i=1}^k \bigg\{\sigma_{\Pi,i}^2 \leq C_1 \frac{nm}{d} + C_2 \frac{nm^2}{d^2}
+ C_3 \frac{n^2m}{d^2}\bigg\}
\cap \big\{\beta_{\Pi,i} \leq \frac \lambda 4\big\} \cap 
\{\zeta_{\Pi,i} \leq 1 + \frac m d \} \cap 
\big\{ K_{\Pi,i} \leq 1\big\},
\end{equation}
for  large enough constants $C_1 ,C_2,C_3 > 0$ to be defined
below, 
then, under condition~\ref{assm:sample_size},
we obtain
\begin{align}\label{eq:goal_of_thresholding_proof}
\nonumber 
\bbE[R] 
 &\lesssim \bbE[R\cdot I(A^\cp)] \\
\nonumber 
 & + \bbE\left\{ \max_{1 \leq i \leq k} 
 K_{\Pi,i} \left[
  (\beta_{\Pi,i}^2-\frac{\lambda^2}{4})_+
  + \frac{d^2\sigma_{\Pi,i}^2}{(nm)^2}
  +  \beta_{\Pi,i} \frac{d\sigma_{\Pi,i}}{nm}
 +   \left(\frac{d(\log d)^2 \zeta_{\Pi,i}}{nm}\right)^2\right]\right\} \\
 &\lesssim d\cdot \bbP(A^\cp) + \lambda,
 \end{align}
where we used the fact that $\hat c_i \leq d$, 
thus $R \lesssim d$. 
To complete the claim, it will thus suffice
to show that the event $A$ occurs with sufficiently high probability.
To this end, we will provide high-probability bounds on the quantities
$\sigma_{\Pi,i}^2$, $\beta_{\Pi,i}$,
$\zeta_{\Pi,i}$ and $K_{\Pi,i}$, in turn.\\

\noindent {\bf Bounding term $\sigma_{\Pi,i}^2$.}
Recall that the entries 
of the matrix $\Pi$ take the form 
$\pi_{ij} = \varpi_{ij} / S_i,$
where $S_i = \sum_{\ell=1}^d \varpi_{i\ell}$, 
and $\varpi_{ij} \sim \calE_d$ are i.i.d. exponential
random variables. Furthermore, 
write $X = (\varpi_{ij}:1 \leq i \leq k ,1 \leq j \leq d) \in \bbR^{k\times d}.$
We  have 
\begin{align*}
\sigma_{\Pi,i}^2 
 &=  \sum_{j=1}^d 
 \Var[V_{ij}Y_j|\Pi] \\ 
 &= \sum_{j=1}^d \Big(\Var[V_{ij}|\Pi] \Var[Y_j|\Pi]  + 
 \Var[Y_j|\Pi] \big(\bbE[V_{ij}|\Pi]\big)^2
 + \Var[V_{ij}|\Pi] \big(\bbE[Y_j|\Pi]\big)^2\Big) \\ 
 &= \sum_{j=1}^d 
 \Big(m\pi_{ij} n(\Pi^\top c)_j 
  + n(\Pi^\top c)_j (m\pi_{ij})^2
  + m \pi_{ij}(n(\Pi^\top c)_j)^2\Big)  \\ 
 &= \sum_{j=1}^d \sum_{s=1}^k c_s
 \Big(
       mn  \pi_{ij} \pi_{sj} 
   +   nm^2   \pi_{sj} \pi_{ij}^2
    +  mn^2   \sum_{s'=1}^k c_{s'} \pi_{s'j}\pi_{ij}\pi_{sj} 
 \Big)  \\
 &= mn\sum_{j=1}^d \sum_{s=1}^k c_s
         \frac{\varpi_{ij} \varpi_{sj} }{S_iS_s}
  +   nm^2 \sum_{j=1}^d \sum_{s=1}^k c_s
   \frac{ \pi_{ij}^2\pi_{sj}}{S_i^2S_s}
   +  mn^2 \sum_{j=1}^d\sum_{s,s'=1}^k c_s   c_{s'} \frac{\pi_{ij}\pi_{sj} \pi_{s'j}}{S_iS_{s}S_{s'}}
 \\
 &=:T_{1i} + T_{2i} + T_{3i}.
\end{align*}
We bound the terms $T_{1i}$, $T_{2i}$, and $T_{3i}$ in turn.
By a sub-exponential tail bound, notice that
one has
\begin{equation} 
\label{eq:concentration_norm} 
\bbP\big(\max_{1 \leq i \leq k} |S_i - 1| > x\big) \lesssim k\cdot 
\exp(-Cd(x\wedge x^2)),
\quad \text{for all } x > 0.
\end{equation}
Under condition~\ref{assm:sample_size}, we deduce that
the event
$\calA_1 = \bigcap_{i=1}^k \{ |S_i - 1| \leq 1/2 \}$
satisfies $\bbP(\calA_1) \geq 1-Ce^{-d/C}$.
Over the event $\calA_1$, we deduce that 
for all $i=1,\dots,k$, 
\begin{align} \label{eq:Fij_analysis}
\frac {T_{1i}}{nm} \asymp 
\sum_{j=1}^d    \sum_{s=1}^k  c_s  \varpi_{ij} \varpi_{sj}
  =\sum_{j=1}^d F_{ij},\quad\text{where } F_{ij} =  \sum_{s=1}^k  c_s  \varpi_{ij} \varpi_{sj}.
\end{align} 
Recalling the Orlicz norms $\|\cdot\|_{\psi_\alpha}$ defined
in equation~\eqref{eq:orlicz}, notice that for all $j=1,\dots,d$,
$$\|F_{ij}\|_{\psi_{1/2}} \leq 
\sum_{s=1}^k c_s \|\varpi_{ij}  \varpi_{sj}\|_{\psi_{1/2}}
\leq \sum_{s=1}^k c_s \|\varpi_{ij}\|_{\psi_{1}} 
\|\varpi_{rj}\|_{\psi_1} 
\lesssim \frac 1 {d^2},$$
where the penultimate inequality
follows from Lemma~\ref{lem:holder_orlicz}, and the final inequality 
is a simple consequence of the fact
that the random variables $d\varpi_{ij}$
are (sub-)exponential with fixed modulus.
We deduce that, up to rescaling, $F_{ij} - \bbE[F_{ij}]$
are $(1/2)$-sub-Weibull random variables, which are independent across $j=1,\dots,d$.
Applying Lemma~\ref{lem:sub_weibull_conc}(i),  
we deduce that for all $x > 0$, 
\begin{equation}\label{eq:Fig_bound}
\bbP\left( \max_{1 \leq i \leq k} \bigg|\sum_{j=1}^d \big(F_{ij} - \bbE[F_{ij}]\big)\bigg| 
\geq (C /d^2)\big(\sqrt{d x} + x^2)  \right)\lesssim ke^{-x}.
\end{equation}
Furthermore, we readily have $\bbE[F_{ij}]\lesssim d^{-2}$, thus 
the above display implies that, for a large enough constant $C > 0$,
the event 
$ \calA_2:= \{\max_{1 \leq i \leq k}|\sum_{j=1}^d F_{ij}| \leq C / d \}$ 
has probability content  $\bbP(\calA_2^\cp) \leq k e^{-\sqrt d} \lesssim e^{-\sqrt d/C}.$
We deduce that over $\calA_1\cap \calA_2$, it holds that
$$\max_{1 \leq i \leq k} T_{1i} \leq C_1 \frac{nm}{d}.$$

To bound $T_{2i}$, we adopt a similar proof. We again have, over the event $\calA_1$, 
\begin{align} 
\frac {T_{2i}}{nm^2} \asymp 
\sum_{j=1}^d    \sum_{s=1}^k  c_s  \varpi_{ij}^2 \varpi_{sj}
  =\sum_{j=1}^d L_{ij},\quad\text{where } L_{ij} =  \sum_{s=1}^k  c_s  \varpi_{ij}^2 \varpi_{sj}.
\end{align} 
The random variables $d^3 L_{ij}$ are $(1/3)$-sub-Weibull, since, reasoning as before, one has
$$\|L_{ij}\|_{\psi_{1/3}}  
\leq \sum_{s=1}^k c_s \|\varpi_{ij}\|_{\psi_{1}}^2
\|\varpi_{rj}\|_{\psi_1} 
\lesssim \frac 1 {d^3}.$$
Furthermore, one has $\bbE[L_{ij}] \lesssim d^{-3}$, 
thus by again applying the sub-Weibull tail bound from Lemma~\ref{lem:sub_weibull_conc}(i), we arrive at 
$$\bbP\bigg( \max_{1 \leq i \leq k}  \sum_{j=1}^d  L_{ij} 
\geq C / d^2 + (C /d^3)\big(\sqrt{d x} + x^3)  \bigg)\lesssim ke^{-x},
\quad x > 0,$$
which implies that the event 
$\calA_3:=  \{\max_{1 \leq i \leq k}|\sum_{j=1}^d L_{ij}| \leq C / d^2\}$
satisfies $\bbP(\calA_3) \geq 1 - e^{-d^{1/3}/C}$. 
Over the event $\calA_1 \cap \calA_3$, we thus obtain 
$$\max_{1 \leq i \leq k} T_{2i}\leq C_2\frac {nm^2} {d^2}.$$ 
for a large enough choice of the constant $C_2 > 0$.
An analogous proof
can be used to show that, for an 
event $\calA_4$
satisfying $\bbP(\calA_4)\geq 1-e^{-d^{1/3}/C}$, 
$$\max_{1 \leq i \leq k} T_{3i}\leq C_3\frac {n^2m} {d^2}.$$
Altogether, we have thus shown that
\begin{align}\label{eq:pf_thresh_step1}
\max_{1\leq i \leq k}\sigma_{\Pi,i}^2  \leq C_1 \frac{nm}{d} + C_2 \frac{nm^2}{d^2}
+ C_3 \frac{n^2m}{d^2}.
\end{align}
over the event $\calA_1\cap\calA_2\cap\calA_3
\cap\calA_4$. This completes our upper bound of $\sigma_{\Pi,i}^2$.
\\

\noindent {\bf Bounding term $\beta_{\Pi,i}$.}
Next, we  provide a high-probability bound on 
the conditional bias term $\beta_{\Pi,i}$. We will make use of the following simple Lemma.
\begin{lemma}\label{lem:Bi}
Define the random variables $B_i = (d+1)\sum_{j=1}^d \sum_{r=1}^k c_r \varpi_{ij}\varpi_{rj}-1$. 
Then, there exists a constant $C_3 > 0$ such that the event 
$$\calA_5 = \bigcap_{i=1}^k \Big\{ \Big| \big(\bbE[\hat c_i|\Pi] - c_i\big) - \big(B_i - \bbE[B_i]\big)\Big|
\leq C_3 / d\Big\}$$
satisfies $\bbP(\calA_5) \geq 1-C_3 e^{-\sqrt d / C_3}$.
\end{lemma}
The proof appears in Appendix~\ref{app:pf_lem_Bi}. 
In view of Lemma~\ref{lem:Bi}, we can bound $\beta_{\Pi,i}$ in the same way as $\sigma_{\Pi,i}^2$. 
Indeed, notice that $B_i +1 \asymp d \sum_{j=1}^d F_{ij}$, thus, using equation~\eqref{eq:Fig_bound}, 
we have for all $x > 0$, 
$$\bbP\left(\max_{1 \leq i \leq k} \big| B_i - \bbE[B_i]\big| > (C/d)(\sqrt {dx}+x^2)\right)
\lesssim ke^{-x}.$$
Notice that $\lambda \geq n^{-1/2} \geq d^{-\frac 1 {2(1+\gamma)}}$ 
under condition~\ref{assm:sample_size}, thus, by choosing $x = d^{\gamma / (2(1+\gamma))-\epsilon}$
for any fixed $\epsilon > 0$, we 
deduce that 
the event $\calA_6=\{\max_{1 \leq i \leq k} |B_i-\bbE[B_i]| > \lambda/8 \}$
satisfies $\bbP(\calA_6^\cp) \lesssim e^{- d^{b} / C}$
for a fixed constant $b = b(\gamma)> 0$. 
It thus follows that, over the event $\calA_5\cap \calA_6$, 
\begin{align}\label{eq:pf_thresh_step2}
\max_{1 \leq i \leq k} \beta_{\Pi,i} \leq \lambda/4.
\end{align}  

\noindent {\bf Bounding term $\zeta_{\Pi,i}$.}
By repeating the same
arguments as in the previous
steps,   there exists an event $\calA_7$
of probability content at least $1-e^{-\sqrt d / C}$
over which it holds that
$$\mu_{Y_j} = n\sum_{i=1}^k c_i \pi_{ij}  
\leq 2n/d,\quad \text{and,}
\quad \mu_{V_{ij}} = m\pi_{ij}\leq 2m/d,$$
uniformly in $i,j$. 
Under condition~\ref{assm:sample_size},
it follows that $\zeta_{\Pi,i} \leq C(1 + \sqrt{m/d})$. 
\\

\noindent {\bf Bounding term $K_{\Pi,i}$.}
The preceding steps readily lead to an upper bound on $K_{\Pi,i}$. 
Indeed, under condition~\ref{assm:sample_size}, over the event $\bigcap_{s=1}^7 \calA_s$, 
we have for all $i=1,\dots,k$ that
$\lambda_{\Pi,i} \geq \lambda/2$ and 
\begin{align}\label{eq:pf_thresh_step3}
\nonumber  
K_{\Pi,i} &=  k 
 \exp\left\{ - \frac{(nm)^2\lambda_{\Pi,i}^2}{C_0 d^2\sigma_{\Pi,i}^2}\right\}
+ o(1)\\
 \nonumber 
 &\lesssim k \exp\left\{ -\frac 1 {C} \frac{nm^2 a\log(k)}{dm_d (nm/d + nm^2/d^2 + mn^2/d^2)}\right\}
 + o(1)\\
\nonumber  &\lesssim k \exp\left\{ -\frac {a\log(k)} {C}   \right\}
 + o(1) \\ 
 &\lesssim k \cdot k^{-a/C}+ o(1) < 1,
 \end{align}
 for a sufficiently large choice of $a$, depending only on $\gamma$. \\ 

\noindent {\bf Concluding the proof.}
By combining equations~\eqref{eq:pf_thresh_step1}, \eqref{eq:pf_thresh_step2}, and~\eqref{eq:pf_thresh_step3},
we deduce that the set $A$ defined in equation~\eqref{eq:sets_A} satisfies
$$\bbP(A^\cp) \lesssim \sum_{i=1}^5 \bbP(\calA_i^\cp) \lesssim e^{-d^b},$$
for a sufficiently small exponent $b> 0$ depending only on $\gamma$.
Therefore, returning to equation~\eqref{eq:goal_of_thresholding_proof}, we arrive at
$$\bbE[R] \lesssim de^{-d^b} + \lambda \lesssim \lambda.$$
The claim follows from here.\qed

\subsection{Proof of Propositions~\ref{prop:ub_blind}--\ref{prop:ub_blind_local}}
\label{app:pf_prop_ub_blind}

We will prove the three
claims by first analyzing the risk of the cumulant estimators $\hat\xi_p$,
under the unnormalized Poisson model.
\begin{lemma}
\label{lem:xip_rate}
Assume condition~\ref{assm:sample_size}
and let $1 \leq p \leq k$. 
Then, under the unnormalized Poisson model,  there exists a constant $C_{p, \gamma} > 0$ such that
$$\bbE|\hat\xi_p-\xi_p| \leq   \frac {C_{p,\gamma} }{\sqrt{n^p d^{p+1}}}.$$
\end{lemma}
\begin{proof} We will make
use of the following fact.
\begin{lemma}\label{lem:risk_V}
\label{app:pf_lem_risk_V}
Under the unnormalized Poisson model, for all $i=1,\dots,k$,  there exists a constant $C_i > 0$ such that
$$\bbE[ T_{1,i}] = \eta_i,\quad \Var[T_{1,i}] \leq  \frac {C_i} {\min\{n,d\}^i d^{i}}.$$
\end{lemma}
The proof appears in Appendix~\ref{app:pf_lem_risk_V}. 
It follows from Lemma~\ref{lem:risk_V}
that
\begin{align}\label{eq:W_unbiased}
\bbE[W_\bh]
 &= \bbE\left[\prod_{i=1}^{p-\ell+1} \prod_{j\in S_i} T_{j,i}\right]  
 =\prod_{i=1}^{p-\ell+1} \prod_{j\in S_i} \bbE\left[T_{j,i}\right] = \prod_{i=1}^{p-\ell+1} \prod_{s=1}^{h_i} \eta_i
  = \prod_{i=1}^{p-\ell+1} \eta_i^{h_i},
\end{align}
thus it is clear that    $\bbE[\hat\xi_p] = \xi_p$. Let us now
compute the variance. 
 Notice that:
$$\Var[\hat\xi_p] \lesssim \sum_{\bh\in \calH_{p,\ell}}\Var[W_\bh],$$
where the implicit constants depend
on $k,p$. 
Now, let us rewrite $W_\bh$ as 
the $p$-th order U-Statistic:
$$W_{\bh} = \frac 1 {{d\choose p}} 
\sum_{1 \leq j_1 < \dots < j_p
\leq d} 
\zeta_\bh(j_1,\dots,j_p),$$
where
$$\zeta_\bh(j_1,\dots,j_{p}) = 
\frac {{d\choose p}} {{d\choose \ell}{\ell\choose \bh}} 
\sum_{(A_1,\dots,A_{p-\ell+1})}
\prod_{i=1}^{p-\ell+1} \prod_{j\in A_i}
T_{j,i},
$$
and where the summation is taken
over all partitions 
$A_1,\dots,A_{p-\ell+1}$ of
$\{j_1,\dots,j_p\}$
such that $|A_i| = h_i$
for all $i$.
By~\cite{lee1990}, we have for any   $\bh\in \calH_{p,\ell}$
and any large enough $d$ that
\begin{align*}
\Var[W_{\bh}]
 &\lesssim \frac 1 d \Var\left[\zeta_\bh(1,\dots,p)
 \right]
 \lesssim  \Var\left[\prod_{i=1}^{p-\ell+1} \prod_{j\in A_i} 
 T_{j,i}\right],
\end{align*}
for any fixed partition $(A_1,\dots,A_{p-\ell+1})$
of $\{1,\dots,j\}$ with $|A_i| = h_i$. 
The random variables appearing in the above product are independent, thus,
together with Lemma~\ref{lem:risk_V},
we have
\begin{align*}
\Var\left[\prod_{i=1}^{p-\ell+1} \prod_{j\in A_i}  T_{j,i}\right]
 &\leq \prod_{i=1}^{p-\ell+1} \prod_{j\in A_i} 
 \big(  \Var\left[T_{j,i}\right] + \bbE\left[T_{j,i}\right]^2\big) \\
&\lesssim \prod_{i=1}^{p-\ell+1} \prod_{j\in A_i} 
\left(\frac 1 {n^i d^{i}} + \frac 1 {d^{2i}}\right) \\
&\lesssim \prod_{i=1}^{p-\ell+1}  
\left(\frac 1 {n^i d^{i}}\right)^{h_i}=   (nd)^{-\sum_i ih_i}  
=   n^{-p} d^{-p}.
\end{align*}
The claim follows.
\end{proof}  

The following is now an immediate
consequence of Lemma~\ref{lem:xip_rate}
and the definition 
of the  
estimator
$\hat m = (\hat m_1,\dots,\hat m_p)^\top$
defined in equation~\eqref{eq:unbiased_moment_estimate_p}.

\begin{lemma}
\label{lem:unbiased_moment_rate}
Assume condition~\ref{assm:sample_size}. Then, under the unnormalized Poisson model, 
there exists a constant $C=C(k,\gamma) > 0$ such that
\begin{equation}\label{eq:mombound}
\bbE \|\hat m-m(c)\|_2 \leq C \sqrt{\frac{d^{k-1}}{n^k}}.
\end{equation}
\end{lemma} 

With Lemma~\ref{lem:unbiased_moment_rate}
in hand, we are now ready to prove
Propositions~\ref{prop:ub_blind}--\ref{prop:ub_blind_local}.
For the proofs of Propositions~\ref{prop:ub_blind}--\ref{prop:ub_blind_cluster}, notice that 
$W(\hat c,c) \leq W(\tilde c,c)$
(cf.\,Lemma~48 of~\cite{hundrieser2025}),
thus  
it suffices to prove upper bounds
for the possibly complex-valued estimator $\tilde c$. 
By Newton's identity~\eqref{eq:newton_identity}, 
the moments of $\tilde c$ are given by
$$m_p(\tilde c) = \hat m_p,\quad p=1,\dots,k.$$
Therefore, we may
apply Lemma~\ref{lem:stability_moments}
to obtain
\begin{align*}
\bbE W(\tilde c,c)  \lesssim 
\bbE  \big\|m(\tilde c) - m(c)\big\|^{\frac 1 k} \lesssim \sqrt{\frac{d^{1 - \frac 1 k}}{n}},
\end{align*} 
where the implicit constants 
depend only on $k,\gamma$. 
This proves Proposition~\ref{prop:ub_blind}. 
Proposition~\ref{prop:ub_blind_cluster} follows
similarly, by now replacing
the $1/k$-modulus of continuity
in the above display 
by $1/(k-k_0+1)$, as a result of Lemma~\ref{lem:refined_stability_moments}. 

Finally, to prove Proposition~\ref{prop:ub_blind_local},  
invoke Proposition~\ref{prop:ub_blind}
to deduce that  
the loss functions $\calD_{c^\star}(\hat c,c)$
and $\widebar \calD_{c^\star}(\hat c,c)$ 
(defined in Appendix~\ref{app:elementary_symmetric_polynomials}) coincide
for all large enough $d$. 
By again applying Lemma~48 of~\cite{hundrieser2025}, we thus obtain
$$ \calD_{c^\star}(\hat c,c)
 = \widebar \calD_{c^\star}(\hat c,c)
\leq \widebar \calD_{c^\star}(\tilde c,c).$$
The claim now follows as before, 
invoking  Lemma~\ref{lem:refined_stability_moments}
to bound $\calD_{c^\star}$ in terms of the moment
differences. The claim follows.\qed  

\subsection{Moment Estimator in the Multinomial Model}
\label{app:mom_corr}
We now show that Propositions~\ref{prop:ub_blind}--\ref{prop:ub_blind_local} 
lead to upper bounds
on the sorted minimax risk
$\calM_<(n,d,k,0)$ for the original
multinomial sampling model, as
well as for the local minimax risk:
$$\calM_<(n,d,k;c^\star,\epsilon)
 = \inf_{\hat c} \sup_{\substack{c\in \Delta_k \\ W(c,c^\star)\leq \epsilon}}
 \bbE_c \Big[\calD_{c^\star}(\hat c(Y,V),c)\Big],$$
 which is defined for any 
 $c^\star \in \Delta_k$ and $\epsilon \geq 0$, where the expectation
 is taken over a realization
 $(Y,V)$ from  the multinomial
 model with parameter $c$.
\begin{corollary}
\label{cor:pois_to_mult}
Assume condition~\ref{assm:sample_size},
and $n > d^{1-\frac 1 k}$.
Then, the following assertions hold.
\begin{enumerate}
\item[(i)]
There exists a constant $C = C(k,\gamma) > 0$ such that 
$$\calM_<(n,d,k,0) \leq C \sqrt{\frac{d^{1 - \frac 1 k}}{n}}.$$
\item[(ii)] For any $1 \leq k_0 \leq k$,
$\delta > 0$, and $c^\star \in \Delta_{k,k_0}(\delta)$, there exist
constants $C,\epsilon > 0$ depending
on $k,\gamma,\delta$ such that
$$\calM_<(n,d,k;c^\star,\epsilon)\leq 
C \sqrt{\frac{d^{k-1}}{n^k}}.$$
\end{enumerate}
\end{corollary}
To prove Corollary~\ref{cor:pois_to_mult}(i),
notice that
\begin{align*}
\calM_<(n,d,k,0)
\lesssim 
\calR_<(n,d,k,0) + e^{-n/C_1}
\lesssim \calR'_<(n,d,k,0) + e^{-n/C_1}
 + \sqrt{\frac n d},
\end{align*}
for a large enough constant $C_1 > 0$,
due to Lemmas~\ref{lem:poissonization}--\ref{lem:kl_poisson_models}. 
By Proposition~\ref{prop:ub_blind}, 
we deduce 
\begin{align*}
\calM_<(n,d,k,0)
\lesssim 
\sqrt{\frac{d^{1-\frac 1 k}}{n}} + e^{-n/C_1}  + \sqrt{\frac n d}.
\end{align*}
Under the condition $n \geq d^{1-\frac 1 k}$, 
the first term on the right-hand
side of the above display is dominant, 
which proves Corollary~\ref{cor:pois_to_mult}(i). 
The second claim can be proven analogously,
by again using the fact that
$\sqrt{d^{k-1}/n^k}$ dominates
$\sqrt{n/d}$ in the regime
$n\geq d^{1/k}$.\qed




\section{Proofs of Main Results}
\label{app:pf_main_theorems}

Our main results---namely
Theorems~\ref{thm:main_unsorted}--\ref{thm:main_sorted}
and Proposition~\ref{prop:main_fidelity}---now 
follow from the lower
and upper bounds developed in the preceding
two appendices. 
Concretely, 
Theorem~\ref{thm:main_unsorted}
follows from the  lower bound
 in Proposition~\ref{prop:lb_unsorted_fixed_k}
 and the upper bound in Proposition~\ref{prop:collision},
 while Theorem~\ref{thm:main_sorted}
 follows from the lower bound in Proposition~\ref{prop:lower_bound_sorted}
 and the upper bounds in Proposition~\ref{prop:ub_blind} (where
 we recall that $W \leq \|\cdot\|$),
 and Corollary~\ref{cor:pois_to_mult}(i).
 Finally,    
 Proposition~\ref{prop:main_fidelity}
 is a direct consequence of 
 Corollary~\ref{cor:pois_to_mult}(ii)
 with $k_0=2$ and $r_1=1$.

\section{Proofs Deferred from Appendices~\ref{app:model}--\ref{sec:ub_proofs}}
\label{app:deferred_pf}

\subsection{Proofs Deferred from Appendix~\ref{app:model}}
\subsubsection{Proof of Lemma~\ref{lem:poissonization}}
\label{app:pf_lem_poissonization}
Our proof follows
a standard Poissonization argument~\citep{canonne2022}.
We prove the claim for the unordered
minimax risk, and a similar
proof can then be used for the ordered
risk. Furthermore, we focus on the case
$m > 0$; adaptations to the
special case $m=0$ are straightforward.

Let $(Y^{(n)},V^{(m)})$ be random variables
drawn from the multinomial model $\widetilde Q_c$
with sample sizes $n$ and $m$. 
Let $N\sim \mathrm{Poi}(n)$ and $M\sim \mathrm{Poi}(m)$
be independent
of all other random variables, and notice that
$(Y^{(N)},V^{(M)})$ is distributed 
according to $\widebar Q_c$.
In this notation, the unordered
minimax risks are given by 
\begin{align*} 
\calM(n,d,k,m) &:= \inf_{\hat c} \sup_{c\in \Delta_k} \bbE_c \|\hat c( Y^{(n)},  V^{(m)}) - c\| \\
\calR'(n,d,k,m) &:= \inf_{\hat c} \sup_{c\in \Delta_k} \bbE_c \|\hat c( Y^{(N)},  V^{(M)}) - c\|.
\end{align*}
Now, given $\epsilon > 0$, let $\hat c_\epsilon$ be
a near-optimal estimator satisfying
$$\sup_{c \in \Delta_k} \bbE_c\|\hat c_\epsilon( Y^{(n)}, V^{(m)})-c\|
\leq \calM(n,d,k,m)+\epsilon.$$
Writing $E_{n'm'} = \{N=n',M=m'\}$, we thus have, 
\begin{align*}
\sup_{c\in \Delta_k}& \bbE_c \|\hat c_\epsilon(Y^{(N)}, V^{(M)})-c \| \\
 &=\sup_{c\in \Delta_k} \bbE_\Pi \left\{ \bbE_c\left[ 
    \|\hat c_\epsilon(  Y^{(N)},  V^{(M)})-c\| \,\big|\, \Pi 
    \right]\right\} \\ 
 &=\sup_{c\in \Delta_k} \sum_{n',m'=0}^\infty
 \bbE_\Pi\left\{ \bbE_c\left[ 
    \|\hat c_\epsilon( Y^{(N)},  V^{(M)})-c\| \,\big|\, \Pi  
    \right]
     \bbP(E_{n'm'}| \Pi)\right\}   \\ 
 &=\sup_{c\in \Delta_k} \sum_{n',m'=0}^\infty 
 \bbE_\Pi\left\{ \bbE_c\left[ 
    \|\hat c_\epsilon( Y^{(n')}, V^{(m')})-c\| \,\big|\, \Pi 
    \right]\right\} 
     \bbP(E_{n'm'}) \\
 &\leq \sup_{c\in \Delta_k}\sum_{n',m'=0}^\infty
    \calM(n',d,k,m') 
     \bbP(E_{n'm'})  + \epsilon.
 \end{align*} 
Since  the risk
 function $\calM(n,d,k,m)$ is monotonically
 decreasing in $n$ and $m$, we deduce that 
\begin{align*}
\sup_{c\in \Delta_k}& \bbE_c \|\hat c_\epsilon( Y^{(N)},  V^{(M)})-c \| \\
 &\leq \calM(n/2,d,k,m/2) + \bbP(N < n/2) + \bbP(M < m/2) + \epsilon \\
 &\leq \calM(n/2,d,k,m/2) +  C e^{-n/C} + C e^{-m/C} + \epsilon,  
 \end{align*} 
where the final inequality
holds for a sufficiently large universal
constant $C > 0$ by
standard Chernoff bounds for the Poisson distribution
(e.g.\,equation (C.1) of~\cite{canonne2022}).
Since $\epsilon$ was arbitrary, it follows that
\begin{align} 
\label{eq:widebarM_to_M}
 \calR'(n,d,k,m) \leq \calM(n/2,d,k,m/2) +  C e^{-n/C} + Ck e^{-m/C}. 
\end{align} 
To prove a converse bound,
we reason similarly as in Lemma~1 of~\cite{wu2019},
and use the fact that the worst-case
Bayes risk provides a lower bound on the minimax risk:
$$\calR'(n,d,k,m)
\geq \sup_{\rho} \inf_{\hat c} \bbE_{c\sim \rho}\big\{\bbE_c \|\hat c( Y^{(N)}, V^{(M)}) - c\|\big\},$$
where the supremum is taken over all probability
distributions on $\Delta_k$. 
Reasoning similarly as before, we have 
for any prior $\rho$,
\begin{align*} 
\inf_{\hat c} & \,\bbE_{c\sim \rho}\big\{\bbE_c \|\hat c( Y^{(N)},  V^{(M)}) - c\|\big\}\\
 &= \inf_{\hat c} 
 \sum_{n',m'=0}^\infty
 \bbE_{c\sim \rho}\big\{\bbE_c \|\hat c( Y^{(n')}, V^{(m')}) - c\|\big\} \cdot 
 \bbP(N=n',M=m')\\
 &\geq 
 \sum_{n'=0}^{2n} \sum_{m'=0}^{2m} \inf_{\hat c} 
 \bbE_{c\sim \rho}\big\{\bbE_c \|\hat c( Y^{(n')}, V^{(m')}) - c\|\big\} \cdot 
 \bbP(N=n',M=m')\\ 
 &\geq 
 \bbE_{c\sim \rho}\big\{\bbE_c \|\hat c( Y^{(2n)}, V^{(2m)}) - c\|\big\}\cdot
\bbP(N \geq 2n,M\geq 2m),
\end{align*}
where we again used
the fact that the map 
$(n',m')\mapsto \inf_{\hat c} 
 \bbE_{c\sim \rho}\big\{\bbE_c \|\hat c( Y^{(n')}, V^{(m')}) - c\|\big\} $
 is decreasing in both of its coordinates.
 Now, by again applying a Poisson Chernoff bound, 
 we obtain $\bbP(N\leq 2n,M\leq 2m) \leq 
 1-C e^{-n/C} - Ce^{-m/C}$. Taking the supremum
 over $\rho$ on both sides of the above display,
 we thus obtain 
\begin{align} 
\label{eq:widebarM_to_M_second}
 \calR'(n,d,k,m) \geq 
\calM(2n,d,k,2m)\cdot \big(1- C e^{-n/C} - Ck e^{-m/C}\big). \end{align}
This proves the claim.\qed

\subsubsection{Proof of Lemma~\ref{lem:kl_poisson_models}}
\label{app:pf_lem_kl_poisson_models}
Let  $C > 0$ be a universal constant, whose value 
may change from one display to the next. 

Let $\Pi \in \bbR^{k\times d}$ be a random 
matrix with rows independently drawn from the flat Dirichlet law
$\calD_d$. Let $G_1,\dots,G_k \sim \mathrm{Gamma}(d,d)$ be independent
Gamma-distributed random variables, which are independent of $\Pi$, and notice that the matrix
$X:= \diag(G) \Pi$ consists of i.i.d. Exp$(d)$-distributed entries (cf.\,Lemma~\ref{lem:dirichlet}). 
Thus, we can write:
\begin{alignat*}{2}
\bQ_c^{\otimes d}
 &=  \bbE_{\Pi,G}[\bQ_{c|\Pi,G}],\quad &&\text{with }  \bQ_{c|\Pi,G } = 
\bigotimes_{j=1}^d 
\left(\mathrm{Poi}(n \displaystyle \sum_{i=1}^k G_i c_i \pi_{ij} )
\otimes \bigotimes_{i=1}^k 
\mathrm{Poi}(mG_i \pi_{ij})\right) \\ 
\widebar \bQ_c &= \bbE_\Pi[ \widebar \bQ_{c|\Pi}],\quad &&\text{with }\bQ_{c|\Pi} = 
\bigotimes_{j=1}^d 
\left(\mathrm{Poi}(n \sum_{i=1}^k c_i \pi_{ij}  )
\otimes \bigotimes_{i=1}^k 
\mathrm{Poi}(m\pi_{ij}  )\right).
\end{alignat*} 
By convexity of the TV distance, one has
\begin{align*}
\TV(\widebar \bQ_c, \bQ_c^{\otimes d})
 &\leq \bbE_{\Pi,G} \Big[ 
\TV(\widebar \bQ_{c|\Pi},\bQ_{c|\Pi,G})\Big] \\ 
 &\leq  \bbE_{\Pi,G} \Big[ 
\TV(\widebar \bQ_{c|\Pi},\bQ_{c|\Pi,G})\cdot I(A)\Big] + \bbP(A^\cp),
\end{align*}
where $A$ denotes the event that $\pi_{ij} \leq 1$  and
$|G_i-1| \leq d^{-1/4}$ for all $i=1,\dots,k$ and $j=1,\dots,d$. 
Notice that $\pi_{1i} \sim \mathrm{Beta}(1,d-1)$, and is
therefore sub-Gaussian with variance proxy $1/4(d+1)$ by~\citep{marchal2017sub}. Furthermore, 
$G_1\sim \mathrm{Gamma}(d,d)$ can be expressed as $G_1 = \frac 1 d \sum_{j=1}^d X_j$ where
$X_j \sim \mathrm{Exp}(1)$ are i.i.d sub-exponential random variables. 
By a sub-Gaussian and sub-exponential tail bound, one readily obtains
$$\bbP(A^\cp) \leq kd \cdot \bbP(\pi_{11} > 1) + k\cdot \bbP(|G_1-1| > d^{-1/4})
\lesssim kd\cdot e^{-d/C} + k\cdot e^{-\sqrt d/C} \lesssim e^{-\sqrt d / C}.$$
It thus remains to bound the mean value of $\TV(\widebar \bQ_{c|\Pi},\bQ_{c|\Pi,G})$ over the event $A$.
By Pinsker's inequality, it will suffice to bound the $\KL$ divergence
$\KL^2(\widebar \bQ_{c|\Pi},\bQ_{c|\Pi,G})$, which, over the event $A$, is bounded from above as
follows (cf. Lemma~\ref{lem:kl_poisson}):
\begin{align*} 
\KL&(\widebar \bQ_{c|\Pi},\bQ_{c|\Pi,G}) \\
 &\lesssim \sum_{j=1}^d  \left\{ \KL\Big(\mathrm{Poi}\big(n\textstyle \sum_i c_i G_i \pi_{ij}\big) ,
 \mathrm{Poi}\big(\sum_i c_i \pi_{ij}\big)\Big) 
 + \displaystyle \sum_{i=1}^k \KL\Big( \mathrm{Poi}(mG_i\pi_{ij}), \mathrm{Poi}(m\pi_{ij})\Big)\right\}  \\
 &\lesssim 
 n \sum_{j=1}^d \frac{(\sum_{i=1}^k c_i \pi_{ij}( G_i-1))^2}{\sum_{i=1}^k c_i \pi_{ij}}
 + m \sum_{j=1}^d \sum_{i=1}^k \frac{((G_i-1)\pi_{ij})^2}{\pi_{ij}}.
\end{align*}
We thus have,
\begin{align*}
\bbE_{\Pi,G} &\Big[ \TV^2(\widebar \bQ_{c|\Pi},\bQ_{c|\Pi,G})\Big] \\ 
 &\lesssim e^{-\sqrt d/C} + nd\cdot  \bbE_{\Pi,G}\left[\frac{(\sum_{i=1}^k c_i \pi_{ij}( G_i-1))^2}{\sum_{i=1}^k c_i \pi_{ij}} \right] + mdk \bbE_{\Pi,G} [(G_1-1)^2 \pi_{11}] \\
 &\leq  e^{-\sqrt d/C} + n \cdot \bbE_\Pi\left[\frac{\Var_G\big[\sum_{i=1}^k c_i \pi_{ij}G_i\,\big|\,\Pi\big]}{\sum_{i=1}^k c_i \pi_{ij}} \right] + mdk \Var[G_1] \cdot \bbE[ \pi_{11}] \\
 &=  e^{-\sqrt d/C} + \frac n d \bbE_\Pi\left[ \frac{\sum_{i=1}^k c_i^2 \pi_{ij}^2}{\sum_{r=1}^k c_r \pi_{rj}} \right]
 +\frac{mk}{d}\\
 &\leq  e^{-\sqrt d/C} + \frac n d \bbE_\Pi\left[  \sum_{i=1}^k c_i  \pi_{ij}  \right]  +\frac{mk}{d^2} =   
 e^{-\sqrt d/C} + \frac{n}{d^2} +\frac{mk}{d^2}.
 \end{align*}
 The claim follows.\qed

\subsection{Proofs Deferred from Appendix~\ref{sec:lower_bounds}}

\subsubsection{Proof of Lemma~\ref{lem:existence_two_point}}
\label{sec:pf_lem_existence_two_point}
We will construct the vectors $u,v$ using a procedure inspired by~\cite{hundrieser2025}. 
Let   $u$ be any fixed vector with mean zero, and with entries satisfying
\begin{equation}
\label{eq:u_gap}
-1/2 < u_1 < u_k < 1/2, \quad u_{i+1} > u_i + \frac 1 {4k},~~i=1,\dots,k-1.
\end{equation}
%
Define the polynomial
$$f_u(z) = \prod_{i=1}^{k} (z-u_i), \quad z \in \bbC.$$
Now, consider the perturbed polynomial $f_v(z) = f_u(z) + (1/4k)^k$. 
In view of the separation condition~\eqref{eq:u_gap},  Lemma~\ref{lem:real_roots} 
ensures that the polynomial $f_v$ has $k$ real roots  $v_1,\dots,v_{k}$ contained in the interval $[-1,1]$.
Since $f_v$ is monic, it takes the form
$$f_v(z) = \prod_{i=1}^{k}(z-v_i).$$
Now, let us apply Vieta's formula (cf.\,Appendix~\ref{app:elementary_symmetric_polynomials}) 
to obtain
\begin{align*}
f_u(z) &= z^{k} + \sum_{j=1}^{k} (-1)^j e_j(u_1,\dots,u_{k})z^{k-j},~~ f_v(z) =
z^{k} + \sum_{j=1}^{k} (-1)^j e_j(v_1,\dots,v_{k})z^{k-j},
\end{align*}
where $e_j$ denote the elementary symmetric polynomials.
Since $f_u$ and $f_v$ only differ in their zeroth-order coefficient, we deduce that
$$e_j(u_1,\dots,u_{k}) = e_j(v_1,\dots,v_{k}),\quad j=1,\dots,k-1.$$
By Newton's identities (equation~\eqref{eq:newton_identity}), it follows from here that $\tilde m_j(u) = \tilde m_j(v)$ for all $j=1,\dots,k-1$.
Furthermore, by again using Newton's identities, we have for all $z \in \bbC$:
\begin{align*}
(1/4k)^k &= f_v(z) - f_u(z) \\
 &= (-1)^{k} \Big[ e_{k}(v_1,\dots,v_k) - e_{k}(u_1,\dots,u_{k})\Big] \\
 &= (-1)^{k} \sum_{j=1}^{k} (-1)^{j-1} \Big[ e_{k-j}(v_1,\dots,v_{k})  m_j(v)-
 e_{k-j}(u_1,\dots,u_{k})  m_j(u)\Big] \\
 &=   m_{k}(v)- m_{k}(u) \lesssim W(u,v),
\end{align*}
where the implicit constant depends on $k$. The claim readily follows from here. 
\qed

\subsubsection{Proof of Lemma~\ref{lem:risk_V}}
\label{app:pf_lem_risk_V}
Let $\theta = \langle c,\varpi\rangle$,
and $Y\sim \mathrm{Poi}(n\theta)$. 
To prove the claim, it suffices
to show that 
$$\hat U_\ell = \frac{Y!}{n^\ell(Y-\ell)!}$$
satisfies
$ \bbE[\hat U_\ell]=\eta_\ell$,
and $\Var[\hat U_\ell] \lesssim \frac 1 {n^\ell d^\ell} + 
\frac 1 {d^{2\ell}}.$
For any $\ell \geq 1$,  one has
$$\bbE[\hat U_\ell\,|\,\varpi] = \theta^\ell ,\quad \Var[\hat U_\ell\,|\,\varpi] 
\lesssim \frac {1} {n^{2\ell}}  (n\theta)^\ell \big((n\theta+\ell)^\ell - (n\theta)^\ell\big),$$
by Lemma~\ref{lem:acharya}. 
To deduce the unconditional bound, notice that
$$\Var[\hat U_\ell]
 = \bbE\big\{ \Var[T_{1,\ell} | \varpi]\big\} + 
   \Var\big\{ \bbE[T_{1,\ell} | \varpi]\big\}.$$
The first term satisfies
\begin{align*}
\bbE\big\{ \Var[\hat U_\ell| \varpi]\big\}
 &\lesssim \frac 1 {d^2 n^{2\ell}}
 \sum_{j=1}^d\bbE\Big[ (n\theta)^\ell +
 (n\theta)^{2\ell-1}\Big]  \\
 &\lesssim \frac 1 {n^{2\ell}} \Big( (n/d)^\ell + (n/d)^{2\ell-1}\Big)
 \lesssim \frac 1 {n^\ell d^{\ell}} + 
 \frac 1 {nd^{2\ell-1}},
\end{align*}
where we used elementary
bounds on the moments of exponential
random variables (cf.\,Lemma~\ref{lem:dirichlet}) whereas the second satisfies
\begin{align*}
\Var\big\{ \bbE[\hat U_\ell|\varpi]\big\}
 \leq \Var[\theta^\ell] \lesssim \frac 1 {d^{2\ell}}.
\end{align*}
The claim follows.\qed 

\subsubsection{Proof of Lemma~\ref{lem:truncation_tv}}
\label{sec:pf_lem_truncation_tv}
We begin by noting the following simple bound.
\begin{lemma}
\label{lem:kl_trunc}
There exist  constants $C,a > 0$ depending on $\gamma$ such that 
\begin{align*}
\sup_{c\in \Delta_k}   \KL(\bQ_c^t \| \bQ_c)  \leq C\cdot e^{-d^a/2}.
\end{align*}
\end{lemma}
The proof appears below.
By Lemma~\ref{lem:kl_trunc}, together with Pinsker's inequality and the tensorization
property of the KL divergence, one has for all $c,\bar c \in \Delta_k$, 
\begin{align*} 
\nonumber
\mathrm{TV}(\bQ_{\bar c}^{\otimes d}, \bQ_{c}^{\otimes d})
 &\leq 
\mathrm{TV}\big(\bQ_{\bar c}^{\otimes d}, (\bQ_{\bar c}^t)^{\otimes d})+
\mathrm{TV}\big((\bQ_{\bar c}^t)^{\otimes d}, (\bQ_{c}^t)^{\otimes d})+
\mathrm{TV}\big((\bQ_{c}^t)^{\otimes d}, \bQ_{c}^{\otimes d}) \\ 
 &\lesssim 
\sqrt{d \cdot \KL (\bQ_{\bar c}^t \,\|\, \bQ_{c}^t)}
 + 
   \sup_{\tilde c \in \Delta_k}\sqrt{d \cdot  \KL (\bQ_{\tilde  c}^t \,\|\, \bQ_{\tilde c}) } \\ 
& \lesssim \sqrt{d \cdot \chi^2 \big(\bQ_{\bar c}^t, \bQ_{c}^t)}+\sqrt d \cdot e^{-d^a/2} \\
& \lesssim \sqrt{d \cdot \chi^2 \big(\bQ_{\bar c}^t, \bQ_{c}^t)}+ e^{-d^a/4}.
\end{align*} 
This proves the first claim. To prove the second claim, we make use of the following 
observation:
\begin{lemma}
\label{lem:hellinger_chi2}
There exist constants $C_1,C_2 > 0$ depending on $\gamma$ such that
for all $c,\bar c \in \Delta_{k}$ satisfying $W(c,\bar c)=0$, 
$$\chi^2(\bQ_c\|\bQ_{\bar c}) \leq C_1 \Big( H^2(\bQ_c\| \bQ_{\bar c}) + e^{-C_2 d^a}\Big).$$ 
\end{lemma}
The proof appears below. We deduce from Lemma~\ref{lem:hellinger_chi2} that
\begin{align*}
\KL&(\bQ_{c}^{\otimes d} \| \bQ_{\bar c}^{\otimes d})  \\
  &\lesssim d \cdot \chi^2 (\bQ_{c}^{\otimes d} \| \bQ_{\bar c}^{\otimes d})  \\
  &\lesssim d\Big( H^2(\bQ_c\|\bQ_{\bar c}) + e^{-C_2d}\Big) \\ 
  &\lesssim d \Big( H^2(\bQ_c \| \bQ_c^t) + H^2(\bQ_c^t\|\bQ_{\bar c}^t) + H^2(\bQ_{\bar c}^t \| \bQ_{\bar c}) + e^{-C_2 d}\Big) \\
  &\lesssim d \Big( \KL(\bQ_c \| \bQ_c^t) + \chi^2(\bQ_c^t\|\bQ_{\bar c}^t) + \KL(\bQ_{\bar c}^t \| \bQ_{\bar c}) + e^{-C_2 d}\Big) \\
  &\lesssim d \cdot  \chi^2(\bQ_c^t\|\bQ_{\bar c}^t)   + e^{-C_2 d^a},
  \end{align*}
for a possibly smaller constant $C_2 > 0$, where we used Lemma~\ref{lem:kl_trunc}
in the final inequality. The claim follows.\qed 

It thus remains to prove Lemmas~\ref{lem:kl_trunc}--\ref{lem:hellinger_chi2}. \newline 

\noindent {\bf Proof of Lemma~\ref{lem:kl_trunc}.}
By   convexity and tensorization of the KL divergence, one has
\begin{align}\label{eq:kl_simple_decomp}
\nonumber 
\KL\big(  \bQ_{c}^t\,\|\, \bQ_{c}\big)  
 &\leq \bbE_{\varpi}\left[ \KL\Big( \mathrm{Po}(n\langle \varpi^t,c\rangle) 
 \,\big\|\, \mathrm{Poi}(n\langle \varpi,c\rangle) \Big) \right]\\
 &\qquad +  k\cdot  \bbE_\varpi \Big[ \KL(\mathrm{Po}( m\varpi_1^t)\,\big\|\, \mathrm{Poi}(m\varpi_1))\Big],
 \end{align}
 where the means are taken over $\varpi\sim \calE_d^{\otimes k}$.
 To bound the first term, notice that the inequality  
 $|\langle c,\pi-\pi^t\rangle| \leq \langle c,\pi\rangle$ always holds, thus we
 may 
 apply Lemma~\ref{lem:kl_poisson} to obtain
 \begin{align*}
\bbE_\varpi\Big[\KL\Big( \mathrm{Poi}(n\langle \varpi,c^t\rangle) 
 \,\|\, \mathrm{Poi}(n\langle \varpi, c\rangle) \Big)\Big]  
 &\leq \bbE_\varpi\left[ \frac{ \langle c, \varpi-\varpi^t\rangle^2}{\langle c,\varpi\rangle} \right] \\  
  &\leq \bbE_\varpi\left[ \frac{ \Big(\sum_{i=1}^k c_i \varpi_i \cdot I(\varpi_i > t)\Big)^2}{\sum_{i=1}^k c_i \varpi_i } \right] \\  &\leq \bbE_\varpi\left[   \sum_{i=1}^k c_i \varpi_i  \,\Big|\, \max_i \varpi_i > t \right] 
  \cdot \bbP\Big(  \max_i \varpi_i > t\Big) \\
 &\leq \sum_{i=1}^k c_i  \bbE_\varpi\left[   \varpi_i \,|\, \varpi_i > t \right]
 \cdot \bbP\Big(  \max_i \varpi_i > t\Big). 
 \end{align*}
 By the memoryless property of the exponential distribution, one has 
 $$\bbE[\varpi_i \,|\,\varpi_i > t] = t + \bbE[\varpi_i] = t + \frac 1 d \leq 2t.$$
 We thus have, 
 \begin{align*}
\bbE_\varpi\Big[\KL&\Big( \mathrm{Poi}(n\langle \varpi,c^t\rangle) 
 \,\|\, \mathrm{Poi}(n\langle \varpi, c\rangle) \Big)\Big]  
 \leq 2t \cdot \bbP\Big(  \max_i \varpi_i > t\Big) \lesssim tk \cdot e^{-dt} \lesssim e^{-d^a/2},
 \end{align*}
 where the final inequality holds for a sufficiently small constant $a=a(\gamma) > 0$
 by definition of $t$, and using the fact that $k\leq d$. An analogous upper bound
 can be obtained on the second term in equation~\eqref{eq:kl_simple_decomp}, and 
 the claim then follows.\qed 
\newline 

\noindent {\bf Proof of Lemma~\ref{lem:hellinger_chi2}.}
Let $c,\bar c$ satisfy $W(c,\bar c)=0$.
Under this condition, we will begin by showing that the 
ratio of the densities
$\bq_c(x,y)$ and $\bq_{\bar c}(x,y)$ is bounded from above and below
by positive constants over a large range of values $(x,y)$. 
Indeed, we have for all $(x,y)\in I$ 
\begin{align*}
\bq_c(x,y)
 &= \frac 1 {x!y!}\bbE\left[f(x;n\langle \varpi,c\rangle)\prod_{i=1}^k f(y_i;m\varpi_i)\right] \\
 &= \frac 1 {x!y!}\sum_{j:|j|=x} {x \choose {j}} \prod_{i=1}^k \bbE\left[(n\varpi_ic_i)^{j_i}e^{-n\varpi_ic_i} (m\varpi_i)^{y_i}e^{-m\varpi_i}\right],
 \end{align*}
 where the summation is taken over all 
 $j = (j_1,\dots,j_k) \in \bbN_0^k$ such that
 $\sum_i j_i = x$, and we write ${x \choose j} = x! / j_1! \dots j_k!$.  Thus, 
\begin{align*}
\bq_c(x,y)
 &= \frac 1 {x!y!}\sum_{j:|j|=x} {x \choose {j}} \prod_{i=1}^k \int_0^\infty  (nu c_i)^{j_i}  (m u)^{y_i} d e^{-(n c_i + d+m)u} du\\
 &= \frac 1 {x!y!}\sum_{j:|j|=x} {x \choose {j}} \prod_{i=1}^k (n c_i)^{j_i} m^{y_i} \frac d {nc_i+d+m}\int_0^\infty u^{j_i+y_i} (nc_i+d+m) e^{-(n c_i + d+m)u} du \\
 &= \frac 1 {x!y!}\sum_{j:|j|=x} {x \choose {j}} \prod_{i=1}^k (n c_i)^{j_i} m^{y_i} \frac{d(j_i+y_i)!}{(nc_i+d+m)^{j_i+y_i+1}}.
 \end{align*}
Writing $\zeta(j,y) = \prod_{i=1}^k (y_i+j_i)! / j_i!$, we thus have
\begin{align*}
\bq_c(x,y)
 &=  \frac{m^{|y|}}{y!}\sum_{j:|j|=x} \zeta(j,y) \prod_{i=1}^k 
 \frac{d(n c_i)^{j_i}}{(nc_i+d+m)^{j_i+y_i+1}}.
 \end{align*}
 On the one hand, this implies
\begin{align*}
\bq_c(x,y) 
 &\leq \frac{m^{|y|}}{y!}\sum_{j:|j|=x} \zeta(j,y) \prod_{i=1}^k 
 \frac{d(n c_i)^{j_i}}{(d+m)^{j_i+y_i+1}} \\
 &\leq \frac{m^{|y|}}{y!}\frac {d^k} {(d+m)^{x+|y|+k}} \sum_{j:|j|=x} \zeta(j,y) 
 \prod_{i=1}^k 
  (n c_i)^{j_i} 
 =: \varphi_c(x,y),
 \end{align*}
Notice that $\varphi_c(x,y)$ only depends on the sorted vector $c$. 
On the other hand, 
\begin{align*}
\bq_c(x,y)
 &= \frac{m^{|y|}}{y!}\sum_{j:|j|=x} \chi(j,y) \prod_{i=1}^k 
\ \frac{d(n c_i)^{j_i}}{(d+m)^{j_i+y_i+1}} 
\left(1 + \frac{nc_i}{d+m}\right)^{-(j_i+y_i+1)} \\
 &\geq \left(1 + \frac{n}{d}\right)^{-(x+|y|+k)}\varphi_c(x,y),
 \end{align*}
We thus   have  for all $c$, 
$$\left(1 + \frac{n}{d}\right)^{-(x+|y|+k)}\varphi(x,y) \leq \bq_c(x,y) \leq \varphi_c(x,y).$$
Since $\varphi_c = \varphi_{\bar c}$
whenever $W(c,\bar c)=0$ for $c,\bar c \in \Delta_k$, we have for all 
such $c,\bar c$ that
$$\left(1 + \frac{n}{d}\right)^{-(x+|y|+k)}  \leq \frac{\bq_c(x,y)}{\bq_{\bar c}(x,y)} \leq \left(1 + \frac{n}{d}\right)^{x+|y|+k}.$$
This implies that  
\begin{equation}
\label{eq:bounded_density_ratio}
1/2 \leq \bq_c(x,y)/ \bq_{\bar c}(x,y)\leq 2,\quad \text{for all } x +|y|\leq  M:= \frac {\log 2}{\log(1 + n/d)} - k.
\end{equation}
Under condition~\ref{assm:sample_size}, we note that $M \leq C_0 d/n$. 
Now, we form the decomposition
\begin{align*}
\chi^2(\bQ_c\| \bQ_{\bar c})
 &= \underbrace{\sum_{\substack{(x,y)\in I \\ x+|y| \leq M}} \frac{\bq_c^2(x,y)}{\bq_{\bar c}(x,y)}  - 1}_{A} 
  + \underbrace{\sum_{\substack{(x,y)\in I \\ x+|y| > M}} \frac{\bq_c^2(x,y)}{\bq_{\bar c}(x,y)}}_{B}. 
\end{align*}
To bound $A$, notice first that for all $z > 0$, 
\begin{align}
\label{eq:tails_trunc_orig}
\nonumber 
\sum_{x+|y| > z} \bq_c(x,y) 
 &= \bbP(n\langle \varpi,c\rangle + m \|\varpi\|_1 \geq z) \\
 &\leq  \bbP((n+ m) \|\varpi\|_1 \geq C_0z) \leq k \exp\left(-\frac{dz}{k (n+m)}\right),
\end{align}
and in particular, since $M \asymp d/n$, we obtain
\begin{align}
\label{eq:tails_trunc} 
\sum_{x+|y| > M} \bq_c(x,y)  \lesssim  \exp\left(-C_2d\right) .
\end{align}
We thus have, 
\begin{align*}
A
 &= \sum_{\substack{(x,y)\in I \\ x+|y| \leq M}} \frac{\bq_c^2(x,y)}{\bq_{\bar c}(x,y)}  - 1 \\
 &=  \sum_{\substack{(x,y)\in I \\ x+|y| \leq M}} \frac{(\bq_c^2(x,y)-\bq_{\bar c}(x,y))^2}{\bq_{\bar c}(x,y)}  
  + 2 \sum_{\substack{(x,y)\in I \\ x+|y| > M}} \bq_c(x,y) - \sum_{\substack{(x,y)\in I \\ x+|y| > M}} \bq_{\bar c}(x,y) \\
 &\lesssim  \sum_{\substack{(x,y)\in I \\ x+|y| \leq M}} \left(\frac{\bq_c^2(x,y)-\bq_{\bar c}(x,y)}{\sqrt{\bq_{\bar c}(x,y)} +\sqrt{\bq_{\bar c}(x,y)} }\right)^2  + 
 e^{-C_2d},
\end{align*}
where the final display follows from equations~\eqref{eq:bounded_density_ratio}--\eqref{eq:tails_trunc}. 
We have thus shown that
\begin{align*}
A
 &\lesssim  H^2(\bQ_{c}, \bQ_{\bar c}) +
 e^{-C_2d},
\end{align*}
Furthermore, 
\begin{align*}
B &= \sum_{\substack{(x,y)\in I \\ x+|y| > M}} \frac{\bq_c^2(x,y)}{\bq_{\bar c}(x,y)} \\
 &\lesssim  \sum_{\substack{(x,y)\in I \\ x+|y| > M}} \bq_c(x,y) \left(1 + \frac n d\right)^{x+|y|+k} \\
 &\lesssim  \sum_{\substack{(x,y)\in I \\ x+|y| > M}} ke^{-\frac{d}{k(n+m)} (x+|y|+k)} \left(1 + \frac n d\right)^{x+|y|+k} 
 \lesssim e^{-C_2 d},
\end{align*}
for a possibly larger constant $C_2 > 0$. 
The claim follows from here.\qed

\subsubsection{Proof of Lemma~\ref{lem:lambda_scaling}}
\label{app:pf_lem_lambda_scaling}
The upper bound is clear. For the lower bound, notice that  
\begin{align*}
 \bbE[\langle c,\varpi^t\rangle] 
 = \sum_{i=1}^k c_i  \bbE[ \varpi_{i}\wedge t]
 \geq \int_0^t x d e^{-dx}dx = \frac 1 d \big( 1 - dte^{-dt}\big),
\end{align*}
as desired.\qed 

\subsubsection{Proof of Lemma~\ref{lem:chi2_lb}}
\label{app:pf_lem_chi2_lb}

We have, 
\begin{align*}
\bbE_\varpi & \left[f(x; \widetilde U_{c})\prod_{i=1}^k f(y_i;\widetilde V_i) \right] \\
 &= \bbE_\varpi \left[\frac{e^{-(\widetilde U_c + \widetilde V_1 + \dots + \widetilde V_k)} (\widetilde U_c \widetilde V_1^{y_1}\dots \widetilde V_k^{y_k})}{x!y_1!\dots y_k!}\right] \\
 &\geq  \frac{e^{-(n + km)t}}{x!y_1!\dots y_k!} \bbE_\varpi \left[\widetilde U_c^x \widetilde V_1^{y_1}\dots \widetilde V_k^{y_k}\right]  \\
 &=  \frac{e^{-(n + km)t}}{x!y_1!\dots y_k!} \left(n^x m^{\sum_i y_i}  \right) 
  \bbE_\varpi \left[\left(\sum_{i=1}^k c_i(\varpi_i\wedge t)\right)^x  \prod_{i=1}^k (\varpi_i\wedge t)^{y_i}\right]  \\
 &=  \frac{e^{-(n + km)t}}{x!y_1!\dots y_k!} \left(n^x m^{\sum_i y_i} \right) 
  \bbE_\varpi \left[ \sum_{\substack{0\leq j_1,\dots,j_k\leq x \\ j_1+\dots + j_k=x}}
                        {x\choose {j_1,\dots,j_k}} \prod_{i=1}^k c_i^{j_i} (\varpi_i\wedge t)^{j_i + y_i}\right]  \\
 &=  \frac{e^{-(n + km)t}}{x!y_1!\dots y_k!} \left(n^x m^{\sum_i y_i} \right) 
   \sum_{\substack{0\leq j_1,\dots,j_k\leq x \\ j_1+\dots + j_k=x}}
                        {x\choose {j_1,\dots,j_k}} \prod_{i=1}^k c_i^{j_i} \bbE_\varpi \left[(\varpi_i\wedge t)^{j_i + y_i}\right]  \\
 &\geq   \frac{e^{-(n + km)t}}{x!y_1!\dots y_k!} \left(n^x m^{\sum_i y_i} \right) 
   \sum_{\substack{0\leq j_1,\dots,j_k\leq x \\ j_1+\dots + j_k=x}}
                        {x\choose {j_1,\dots,j_k}} \prod_{i=1}^k c_i^{j_i} \lambda^{j_i + y_i} \\
 &=   \frac{e^{-(n + km)t}}{x!y_1!\dots y_k!} \left(n^x m^{\sum_i y_i} \right) \lambda^{x+\sum_i y_i}  \\
 &= e^{-(n + km)t} \cdot f(x;n\lambda) \prod_{i=1}^k f(y_i;m\lambda). 
\end{align*}
The claim now follows from the fact that $nt \leq (n/d)^{1-\gamma_0}\leq 1$
and $mtk \leq k(m/d)^{1-\gamma_0} = k(m/d)^{\frac 1 {1+\gamma}}\leq 1$,
by assumption on $k$.
\qed


\subsubsection{Proof of Lemma~\ref{lem:charlier_tensor}}
\label{app:pf_lem_charlier_tensor}

The collection $\{\varphi_{\alpha,\beta}(\cdot;\blambda)\}_{\alpha,\beta}$ is dense in $L^2(g_{\blambda})$, 
and satisfies the orthogonality property
 \begin{align*}
&\bbE_{(X,Y)\sim g_{\blambda}}\Big[ \varphi_{\alpha,\beta}(X,Y;\blambda) \varphi_{\alpha',\beta'}(X,Y;\blambda)\Big] \\
  &\quad = \sum_{(x,y)\in I} \Big(\varphi_\alpha(x; \lambda_0)  \varphi_{\alpha'}(x; \lambda_0)  f(x;\lambda_0)\Big)
   \cdot \prod_{i=1}^k \Big(\varphi_\beta(y_i; \lambda_{i})\varphi_{\beta'}(y_i; \lambda_{i}) f(y_i;\lambda_i)\Big) \\
  &\quad = \left(\sum_{x=0}^\infty \varphi_\alpha(x; \lambda_0) \varphi_{\alpha'}(x; \lambda_0)  f(x;\lambda_0)\right)\cdot 
  \prod_{i=1}^k  \sum_{y_i=0}^\infty \Big(\varphi_\beta(y_i; \lambda_{i})\varphi_{\beta'}(y_i; \lambda_{i}) f(y_i;\lambda_i)\Big) \\
  &\quad =  \alpha! \lambda_0^\alpha I(\alpha=\alpha')\cdot 
  \prod_{i=1}^k  \beta_i! \lambda_i^{\beta_i} I(\beta_i=\beta_i'),
\end{align*}
where we used the orthogonality of the univariate Charlier basis (cf.\,equation~\eqref{eq:charlier_ortho}). 
We deduce that $\{\varphi_{\alpha,\beta}(\cdot;\blambda)\}_{\alpha,\beta}$ forms an orthogonal basis
of $L^2(g_{\blambda})$. To prove the second identity, recall from equation~\eqref{eq:charlier_generating_fn}
that the generating function of the Charlier polynomials with parameter $\lambda_0$ is given by $e^{-u_0} (1 + u_0/\lambda)^x$ for all 
$u_0 > -\lambda_0$, thus we have
\begin{align*}
\frac{f_{\lambda_0 + u_0}(x)}{f_{\lambda_0}(x)} = e^{-u_0} \left(1 + \frac {u_0}{ \lambda_0}\right)^x = \sum_{\ell=0}^\infty \varphi_\ell(x;\lambda_0) \frac{(u_0/\lambda_0)^\ell}{\ell!},\quad x=0,1,\dots.
\end{align*}
Re-applying a similar identity, we obtain for all $(x,y)\in I$ and all $\bu \in \bbR^{k+1}$ such that $u_j \geq -\lambda_j$,
$j=0,\dots,k$,  
\begin{align*}
\frac{g_{\blambda+\bu}(x,y)}{ g_{\blambda}(x,y)}
 &= \left(\sum_{\alpha=0}^\infty \varphi_\alpha(x;\lambda_0)
 \frac{(u_0/\lambda_0)^\alpha}{ \alpha!}\right)
 \prod_{i=1}^k \left( \sum_{\beta_i=0}^\infty \varphi_{\beta_i}(x_i;\lambda_i) \frac{(u_i/\lambda_i)^{\beta_i}}{\beta_i!}\right) \\
 &=  \sum_{(\alpha,\beta)\in I}^\infty 
 \varphi_{\alpha,\beta}(x,y;\blambda)
 \frac{(u_0/\lambda_0)^\alpha}{ \alpha!}
 \prod_{i=1}^k  \frac{(u_i/\lambda_i)^{\beta_i}}{\beta_i!}.
\end{align*}
and the claim then follows.
\qed

\subsubsection{Proof of Lemma~\ref{lem:sum_of_beta}}
\label{app:pf_lem_sum_of_beta}
Let $G_1,\dots,G_L \overset{iid}{\sim} \mathrm{Exp}(1)$
and let $G = \sum_{i=1}^L G_i\sim \mathrm{Gamma}(s,1)$. Let $X = d\sum_{i=1}^L X_i$. 
Notice that  the lower tail of the rescaled Beta density $f_{dX_1}$ is dominated by the exponential density
$f_{G_1}$; indeed, one has for all $x \in [0,1/2]$, 
$$f_{dX_1}(x) = \frac {d-1}{d} (1-x/d)^{d-2}
 \leq e^{-x} (1-x/d)^{-2} \lesssim e^{-x} = f_{G_1}(x).$$
We thus have
\begin{align*} 
\bbE\left[\left(d \sum_{i=1}^L X_i\right)^{-3/2}\right] 
 \lesssim \bbE[1/G^{3/2}]  + \bbE\left[   {X^{-3/2}}\Big|\min_{i\in S_2} d\pi_{i1} > 1/2\right]
 \lesssim \bbE[1/G^{3/2}]  +  L^{-{3/2}}.
 \end{align*} 
 The remaining expectation can be computed
 in closed form as 
 $$\bbE[1/G^{3/2}] = \Gamma(L-3/2) / \Gamma(L) \lesssim L^{-3/2},$$
where $\Gamma$ denotes the Gamma function. The claim follows.\qed

\subsection{Proofs Deferred from Appendix~\ref{sec:ub_proofs}}

\subsubsection{Proof of Lemma~\ref{lem:Bi}}
\label{app:pf_lem_Bi}
Notice first that 
$$\bbE[B_i] =(d+1) \sum_{j=1}^d \left(\frac{2c_i}{d^2} + \sum_{r\neq i} \frac{c_r}{d^2}\right) 
= (d+1)d \left(\frac{c_i}{d^2} + \frac 1 {d^2}\right) = c_i +O(d^{-1}).$$
Second, notice that by equation~\eqref{eq:concentration_norm}, one has
with probability at least $1-Ce^{-\sqrt d/C}$,
\begin{align*} 
\big|\bbE[\hat c_i|\Pi] - B_i\big|
 \leq
 (d+1)\sum_{j=1}^d \sum_{r=1}^k c_r\varpi_{ij}\varpi_{rj}\big| (1 - S_r^{-1})(1-S_i^{-1})\big|  
 \lesssim  d\sum_{j=1}^d \sum_{r=1}^k c_r\varpi_{ij}\varpi_{rj}.
\end{align*}
The claim now follows by re-applying the same
argument as under equation~\eqref{eq:Fij_analysis}.\qed

\section{Description of Synthetic Data Analysis in Section~\ref{sec:numerical}}
\label{appendix:numerical}
In this appendix, we provide the details of our analysis on the synthetic data presented in the main text.

\subsection{Time-dependent models}
To generate the synthetic dataset that mimics possible increasing error rates in the real experiments, we consider a one-dimensional array of $L$ qubits and construct a circuit of $L$ layers. 
Each layer consists of a set of two-qubit random unitaries applied to neighboring qubits on all even bonds, and a set of following random unitaries on odd bonds. 
After each layer, we introduce single-qubit Pauli errors ($X$, $Y$, or $Z$) on every qubit. 
We study two types of error models: 
\begin{enumerate}
\item {\bf Experiment-mimicking model}: the error rate at each layer is drawn uniformly from $[0.25\epsilon, 0.75\epsilon]$, where $\epsilon$ increases linearly from $2.5\times 10^{-4}$ (first layer) to $10^{-3}$ (last layer). 
\item {\bf Null model}: $\epsilon$ is fixed at $\sim 6\times 10^{-4}$. 
\end{enumerate}
For both error models, the many-body fidelity is $F\approx 0.5$. 

In Regime A (with side information from classical simulation), we analyze the synthetic data using MLE, where each column of $\Pi$ corresponds to one of the Pauli errors. 
From this estimator, we extract the error rate for each spacetime position (Fig.~\ref{fig:timedep_error}a) and the average error rate $\epsilon_{\mathrm{est}}$ for each layer (Fig.~\ref{fig:timedep_error}b).
Fitting $\epsilon_{\mathrm{est}}$ as a linear function of depth, $\beta \cdot \mathrm{depth} + \epsilon_0$, yields an error growth rate $\beta$. 
To validate the time-dependence in the error rate, we should test whether the extracted $\beta$ significantly differs from zero. 
For this purpose, we simulate $500$ instances of the null (time-independent) model and perform the above analysis, constructing a histogram of the extracted $\beta$ values to determine confidence intervals and $p$-values (Fig.~\ref{fig:timedep_error}c).
In regime B, we repeat the same analysis using variational EM. 

\subsection{Correlated error models}
To generate the synthetic dataset that mimics spatially correlated errors possibly existing in the real experiments, we consider a $5\times 4$ two-dimensional array of qubits and construct a five layer circuit. 
Each layer consists of four sets of two-qubit random unitaries, consecutively applied to neighboring qubits cycling among pairs in the four different orientations. 
After each layer, we introduce incoherent errors, which include all single-qubit Pauli errors as well as select correlated errors. 
All error rates are assumed to be the same across different layers. 
The single-site Pauli error rates are drawn from a uniform random distribution $ [10^{-3}, 3\times 10^{-3}]$.
We consider two different models of correlated errors: 
\begin{enumerate}
\item {\bf Two-body correlated error}: correlated-$XX$ errors that may exist for any pair of qubits. Here, we consider the situation where error rates are negligibly small except for one ``bad'' pair with error rate $\sim 10^{-3}$. Our goal is to identify this pair by applying our algorithm to synthetic data. 
\item {\bf Multi-body correlated error}: correlated multi-$X$ errors can exist along any column or row of qubits. Here, we consider the situation where all these error rates are negligibly small except one ``bad'' row and one ``bad'' column with error rate $\sim 10^{-3}$. Our goal is to identify such a row and column from analyzing the synthetic data. 
\end{enumerate}
Error rates are chosen such that the many-body fidelity is $F\approx 0.5$ in both models. 
In Fig.~\ref{fig:correlated_error}, we focus on regime A (i.e.~with classically computed $\pi_i$'s) and use the MLE estimator. 

\section{Description of Real Data Analysis in Section~\ref{sec:google}}
\label{app:google}
In this Appendix, we describe the error model and the numerical methods used to analyze data from the experiment in Ref.~\cite{arute2019quantum}.
\subsection{Error Model}

\begin{table}[]
    \centering
    \begin{tabular}{c c  c c}
    \toprule
    Error & Kraus op.~$K_i^{(a)}$ & Coef. $w^{(a)}_i$ & Fid. contribution $f_i$.\\
    \midrule
    State prep. & $X$ & +1 & 0\\
     1q dephasing    & $Z$& +1 & 0\\
 2q dephasing    & $\begin{pmatrix}
 1&0&0&0\\
 0&1&0&0\\
 0&0&1&0\\
 0&0&0&-1
 \end{pmatrix}$& +1 & $+1/4$\\
 2q flip-flop & $\begin{pmatrix}
 1&0&0&0\\
 0&0&1&0\\
 0&1&0&0\\
 0&0&0&1
 \end{pmatrix}$& +1 & $+1/4$\\
\\
     \multirow{2}{*}{$1\rightarrow0$ readout error}    & $|0\rangle\langle1|$& +1 & \multirow{2}{*}{-1/2}\\
& $|1\rangle\langle1|$ & -1 & \\
    \multirow{2}{*}{$0\rightarrow1$ readout error} & $|1\rangle\langle0|$& +1 & \multirow{2}{*}{-1/2}\\
& $|0\rangle\langle0|$ & -1 & \\
\\
\multirow{4}{*}{
$1\rightarrow0, 1\rightarrow0$ double readout error}    & $|00\rangle\langle11|$& +1 & \multirow{4}{*}{+1/4}\\
& $|01\rangle\langle11|$ & -1 & \\
& $|10\rangle\langle11|$ & -1 & \\
& $|11\rangle\langle11|$ & +1 & \\
\bottomrule
    \end{tabular}
    \caption{Error processes modeled in the analysis of RCS data from Ref.~\cite{arute2019quantum} (Fig.~\ref{fig:google_data}). The readout error sources have multiple terms $K_j^{(a)}$ with coefficients $w_j^{(a)}$ [Eq.~\eqref{eq:quantum_channel}], derived below. The last column indicates the fidelity contribution of the error source, necessary to obtain the many-body fidelity (App.~\ref{app:converting_rates}).}
    \label{tab:error_sources}
\end{table}

Our data analysis involves a model for the output probability distribution of the form
\begin{align} 
p_c(z|\Pi) = c_1\pi_1(z) + \sum_{i>1} c_i \pi_i(z) + c_{-1} 1/d,
\label{eq:k_component_Google}
\end{align}
where the index $i$ contains information both about the error type and spacetime location. 

As stated in Eq.~\eqref{eq:pi_Kraus}, we generate the $\Pi$ matrix from a physical model of a noisy quantum state, parameterized by noise coefficients we wish to learn. In general, these terms, proportional to the unknown coefficient $c_i$, will be of the form 
\begin{equation}
    R_i(\rho) = \sum_a w_i^{(a)} K_i^{(a)}\rho K_i^{(a)\dagger},
    \label{eq:quantum_channel}
\end{equation}
where the sum over $a$ indicates multiple terms which may be associated with the same error source. As examples, when the error corresponds to a unitary operator (e.g.~a Pauli error), there is only one term, and $K_i$ is said unitary. However, more complicated error channels such as asymmetric readout errors require multiple terms, e.g.~$R_{1\rightarrow 0,j}^{\text{readout}}(\rho) = (|0\rangle\langle 1|)_j \rho (|1\rangle\langle 0|)_j -(|1\rangle\langle 1|)_j \rho (|1\rangle\langle 1|)_j  $ describes the contribution of $1\rightarrow 0$ readout errors on qubit $j$, see Table~\ref{tab:error_sources} and discussion below.

To benchmark a realistic error model, we include several classes of errors which have been reported in the literature. They are: 1. State preparation errors (simply as a bit-flip $X_j$ on the initial state $|0\rangle^{\otimes N}$), 2. single-qubit errors, where for our analysis we focus simply on Pauli $Z_j$ dephasing errors, and 3. two-qubit errors representing (a) dephasing on the $11$ state and (b) flip-flop exchange between neighboring pairs of qubits. Finally, we include 4. asymmetric readout errors with different rates of $0\rightarrow1$ and $1\rightarrow 0$ errors, as well as double $1\rightarrow 0, 1\rightarrow 0$ readout errors, which occur at non-negligible rates because of the larger $1\rightarrow 0$ error rates.
Each of these have error channels of the form Eq.~\ref{eq:quantum_channel}, with parameters summarized in Table~\ref{tab:error_sources}.

The two-qubit errors may arise from processes such as coupling to higher transmon levels~\cite{arute2019quantum,andersen2025thermalization}, which may appear as stochastic errors in the control angles $\phi$ and $\theta$ of the FSIM class of gates:
\begin{equation}
\text{FSIM}(\theta,\phi) =
\begin{pmatrix}
1 & 0 & 0 & 0\\
0 & \cos \theta/2 & -i \sin \theta/2 & 0\\
0 & -i \sin \theta/2 & \cos \theta/2 & 0\\
0 & 0 & 0 & \exp(i\phi)
\end{pmatrix}
\end{equation}
Integrating over Gaussian fluctuations of $\theta$ and $\phi$ gives a more complicated channel (of Lindblad form) proportional to the fluctuations $\Delta\theta^2, \Delta \phi^2$. However, in this work we do not assume a precise model for these two-qubit errors and we instead use a simpler unitary error channel (Table~\ref{tab:error_sources}), taking these as representative of dephasing processes on the $11$ state, or flip-flop between $01$ and $10$ states.

While symmetric readout errors can simply be modeled as Pauli $X_j$ errors on qubits $j$, asymmetric readout errors, which capture the strongly biased readout errors reported in Ref.~\cite{arute2019quantum}, are more involved.
The simplest way to obtain the relevant operators for asymmetric readout is to linearize the amplitude damping channel (of strength $\gamma$, acting on qubit $j$)
\begin{align}
    &R[\rho] = \begin{pmatrix}
    1 & 0 \\
    0 & \sqrt{1-\gamma }
    \end{pmatrix}_j \rho \begin{pmatrix}
    1 & 0 \\
    0 & \sqrt{1-\gamma }
    \end{pmatrix}_j + \begin{pmatrix}
    0 & \sqrt{\gamma} \\
    0 & 0
    \end{pmatrix}_j \rho \begin{pmatrix}
    0 &  0 \\
    \sqrt{\gamma} & 0
    \end{pmatrix}_j \label{eq:amp_damp}\\
    & = \rho + \gamma \left[ - |1\rangle\langle 1|_j \rho |1\rangle\langle 1|_j -\frac{1}{2} |0\rangle\langle 0|_j \rho |1\rangle\langle 1|_j -\frac{1}{2} |1\rangle\langle 1|_j \rho |0\rangle\langle 0|_j  + |0\rangle\langle1|_j \rho |1\rangle\langle0|_j \right] + O(\gamma^2). \nonumber
\end{align}
The middle two terms (proportional to $\gamma$) correspond to dephasing induced by amplitude damping and can be neglected for readout errors, since the system is immediately measured in the $Z$ basis. However, if one wanted to model an amplitude damping channel in the middle of the circuit, all terms above should be kept). This gives the operators for $1\rightarrow 0$ readout in Table~\ref{tab:error_sources}.

The effect on the classical probability distribution can be understood as follows. Assume for simplicity that a $1\rightarrow 0$ readout error occurs on the first bit. We write the distribution $\pi_1(z) = (\pi_1^{(0)}(z'), \pi_1^{(1)}(z'))$ in terms of the distributions on the substrings $z' = z_2 z_3 \cdots z_N$, conditioned on $z_1 = \{0,1\}$. The readout error acts on the distribution as:
\begin{equation}
   \pi_1(z) \mapsto \pi_1(z) + \gamma \left(\pi_1^{(1)}(z'), - \pi^{(1)}(z')\right) \equiv \pi_1(z) + \gamma \pi^{\text{readout}}_{1\rightarrow0,1}(z),
   \label{eq:1_0_asym_readout}
\end{equation}
that is, it shifts probability mass from $z_1=1$ onto $z_1=0$. These terms precisely correspond to application of the operators in Table~\ref{tab:error_sources}.

Double readout errors are modeled in a similar fashion: applying the amplitude damping channel with rates $\gamma_i, \gamma_j$ on qubits $i$ and $j$, and keeping terms proportional to $\gamma_i\gamma_j$, we obtain
\begin{align}
   \pi_1(z) \mapsto  &\pi_1(z) + \gamma_i \pi^{\text{readout}}_{1\rightarrow0,i}(z) + \gamma_j \pi^{\text{readout}}_{1\rightarrow0,j}(z) \nonumber\\
   &+ \gamma_i\gamma_j \left(\pi_1^{(11)}(z'), - \pi_1^{(11)}(z'),- \pi_1^{(11)}(z'), \pi_1^{(11)}(z') \right), \label{eq:2q_readout}
\end{align}
where the length-four vector now runs over $(z_i,z_j) \in \{00,01,10,11\}$, and $z' = z \backslash (z_i,z_j)$. Surprisingly, the second order contribution \textit{adds} probability mass to the $11$ state.

Finally, in Eq.\eqref{eq:k_component_Google} we additionally include a ``white noise" term proportional to $1/d$, not assumed in our theoretical analysis, but which models the aggregate weight of errors outside of our model, which we expect to sum to such a featureless distribution~\cite{dalzell2024random}. 
Assuming a Markovian error model with a total rate of $\gamma$ local error events per unit time, a simple estimate gives a many-body fidelity of $e^{-\gamma t}$, $\gamma t e^{-\gamma t}$ ``single" error events, and $[(\gamma t)^2/2]e^{-\gamma t}$ ``double" error events where two independent local errors occur. In this work, we neglect the vast majority of such double and higher-order errors, only including double readout events. Therefore, a simple estimate for the weight $c_{-1}$ is $1 - \hat{F}- \hat{F}\log(1/\hat{F})$, where $\hat{F}$ is the estimated many-body fidelity (see Appendix~\ref{app:converting_rates}). For the $N=18$ dataset, this gives $0.41$ for $\hat{F} = 0.24$: we additionally estimate double readout errors to constitute $0.05$ of the signal, leading to reasonable agreement with the estimated $\hat{c}_{-1} = 0.32(1)$.

After obtaining $\pi_1(z)$ and $\pi_i(z)$ via classical simulation of the RCS circuits using the \verb|Cirq| package, we construct a matrix $\Pi$ with entries 
$\pi_{ij} := \pi_i(z_j)$. This matrix generally has negative entries, even though $\Pi^\top c$ is guaranteed to have nonnegative entries for any physically sensible error vector $c$. 
We perform our fitting procedure under the modeling assumption that the bitstring histogram
$Y$ has entries drawn independently according to the distribution:
\begin{equation}
\label{eq:poisson_model_google}  
Y_j\sim \mathrm{Poi}(n\Pi_{\cdot j}^\top c),
\quad j=1,\dots,d,
\end{equation}
for some $c \in \{x \in \bbR_+^k: \Pi_{\cdot j}^\top x \geq 0,~j=1,\dots,d\}$.  
As we saw in Section~\ref{app:model},
this Poissonian model is statistically indistinguishable from the multinomial sampling model when the shot noise is large,  the entries of $\Pi$ are nonnegative, and $c\in \Delta_k$.
In our more general setting here, 
where $c$ may not lie in the simplex, the multinomial model is not well-defined, which is the reason we adopt the above more general Poissonian model. 
We use a Poisson MLE estimator to fit the $c$ coefficients,
\begin{align} \label{eq:mleapp}
\hat{c}^\mathrm{MLE} = \argmax_{x \in \bbR_+^k} \sum_{j=1}^d \bigg(Y_j \log(\Pi_{\cdot j
}^\top x) - \Pi_{\cdot j
}^\top x\bigg),
\end{align}
which should be contrasted to the multinomial
MLE presented in equation~\eqref{eq:main_mle} of the main text. 
In practice, we find the matrices $\Pi$ to be poorly conditioned, and we add 
a ridge regularization penalty $10^{-8}\|x\|_2^2$
to the objective function~\eqref{eq:mleapp}.  
 
\subsection{Converting learned error rates into physical quantities}
\label{app:converting_rates}

\subsubsection{Many-body fidelity}
\label{app:fidelity_conversion}
As alluded to in Table~\ref{tab:error_sources}, we obtain the many-body fidelity estimate by a weighted sum of the learned $c_i$:
\begin{align}
    \hat F = \hat{c}_1 + \sum_{i>1} f_i \hat{c}_i,
    \label{eq:fid_correction}
\end{align}
where coefficients $f_i$ for various sources of error are given in Table~\ref{tab:error_sources}. 

The reason why this is required is because the sources of error we consider need not result in output states orthogonal to the target output state, and hence output distributions $\pi_i$ orthogonal to $\pi_1$.

For a Pauli error channel, with high probability in a RUC, the quantum state associated with each error trajectory has exponentially small overlap with the target quantum state and hence $f_i=0$ in these cases.
For a more general error channel, however, this is not the case. $f_i$ can be computed by a simple analytical theory: one simply assumes the state is Haar random at the point the error is applied. For a single error, a Haar-average~\cite{Collins2022} reveals that the many-body fidelity is 
\begin{equation}
f_i \approx  \mathbb{E}_{\psi\sim \text{ Haar}}[ \langle \psi| \sum_i w_i^{(a)} K_i^{(a)} |\psi\rangle \langle \psi| K_i^{(a)\dagger} |\psi\rangle ] =\sum_i w_i^{(a)} \frac{|\text{tr}(K_i^{(a)})|^2 + \text{tr}(K_i^{(a)}K_i^{(a)\dagger})}{d(d+1)},
\end{equation}
where the trace should be taken as over the entire $N-$qubit Hilbert space. For 2-qubit dephasing or flip-flop error (Table~\ref{tab:error_sources}), $|\text{tr}(K_i)|^2 = d^2/4$, while $\text{tr}(K_i^{}K_i^{\dagger}) = d$ (the second term is always sub-leading), leading to the coefficients $f_i = +1/4$. That is to say, acting with a controlled-Z or flip-flop unitary ``error" produces a state which has, on-average, an fidelity of 1/4 with the target state. This fidelity should be added back to the fidelity estimate $\hat{F}$, in addition to $\hat{c}_1$.

As another example, for asymmetric readout errors, one in fact has to use all the terms of the linearized amplitude damping channel Eq.~\eqref{eq:amp_damp} (including the off-diagonal terms), this gives $f_i = -1/2$. As intuition, if there were \textit{only} $1\rightarrow0$ readout errors, i.e. with the model
\begin{equation}
p(z) = c_1 \pi_1(z) + \sum_{j} c_j \pi^\text{readout}_{1\rightarrow0,j} (z),
\end{equation}
our algorithm would learn a coefficient $\hat{c}_1 = 1$, since for all the other terms $\sum_z  \pi^{\text{readout}}_j (z) = 0$, but the sampled distribution is by definition normalized. However, the actual many-body fidelity is smaller, precisely $c_1 - \sum_j c_j /2$, with the factor of $1/2$ arising from the probability of the bit being in the $1$ state. As an independent check, one can verify that the XEB fidelity between $\pi_1$ and $\pi^\text{readout}_{1\rightarrow0,j}$ is 1/2. A similar calculation yields $f_i = +1/4$ for double asymmetric readout errors: we summarize these results in Table~\ref{tab:error_sources}.

The simple behavior of Pauli errors discussed above does not hold true near the start and end of the RUC. Near the start, the circuit depth is too low for the Haar-random assumption made above to hold, and a local error does not orthogonalize the state. Meanwhile, near the end of the RUC, a local operator \textit{does} orthogonalize the state. However, a dephasing error does not sufficiently scramble before measurement in order to change the XEB: this is the ``lag time" in the XEB that had been previously noted~\cite{mark2023benchmarking}. Both effects are evident in the correlation matrix $\Pi^T \Pi$ for 1q dephasing, 2q dephasing and 2q flip-flop errors: furthermore, these effects are largely confined to the first and last three layers.
While the latter effect does not contribute to the many-body fidelity, we omit both these boundary circuit layers in order to cleanly test our fidelity coefficients $f_j$ for errors deep in the circuit: doing so reveals close quantitative agreement between the XEB and our estimate $\hat{F}$ [Eq.~\eqref{eq:fid_correction}] in Fig.~\ref{fig:google_data}.

A more refined theory of these fidelity contributions $f_j$ that incorporates the space-time positions of the errors would enable our method to include such boundary errors without comprimising the fidelity estimate.

\subsubsection{Correction of double readout errors on single readout error rates}
A similar effect happens between single and double readout errors: these have non-trivial overlaps, and after our fitting procedure, one must correct the estimate of the readout error rate on qubit $j$ as
\begin{equation}
    \hat{c}^{\text{readout}}_{1\rightarrow 0, j} \mapsto \hat{c}^{\text{readout}}_{1\rightarrow 0, j} - \frac{3}{14} \sum_{k\neq j}\hat{c}^{\text{doub.~readout}}_{(1\rightarrow 0)^2, jk}
    \label{eq:double_readout_corr}
\end{equation}
where the sum is taken over qubits $k\neq j$. The coefficient of $3/14$ comes from the following considerations: the quantum fidelity for mixed states (such as the contributions of single- and double- readout error) is less straightforward to analyze. Therefore, we use as a proxy a heuristic analysis based on the XEB: we seek to ``orthogonalize" the $\Pi$ matrix rows $\pi^\text{readout}_{1\rightarrow0,i}$ and $\pi^{\text{doub.~readout}}_{(1\rightarrow 0)^2, jk}$. We define orthogonalization with respect to the dot product:
\begin{equation}
    \langle \pi_1, \pi_2\rangle \equiv d \sum_z \pi_1(z) \pi_2(z) 
\end{equation}
It is also convenient to subtract the identity component such that all vectors we consider sum to zero, that is work with $\pi_1 - 1/d$ instead of $\pi_1$. This has the feature that the XEB fidelity of a distribution $p(z)$ can be understood as the inner product $\langle p, \pi_1-1/d\rangle$.

This orthogonalization procedure correctly reproduces the fidelity contributions $f_j$: $\pi^\text{readout}_{1\rightarrow0,i} +(\pi_1-1/d)/2$ and $\pi^{\text{doub.~readout}}_{(1\rightarrow 0)^2, jk} - (\pi_1-1/d)/4$. Our estimation problem is equivalent to fitting to a modified model:
\begin{equation}
    p = 1/d +  F(\pi_1-1/d) + \cdots + \sum _i c_{1\rightarrow0,i} [\pi^\text{readout}_{1\rightarrow0,i} + (\pi_1-1/d)/2] + \cdots,
\end{equation}
where the coefficient $F$ is precisely the many-body fidelity (more precisely, this prescription ensures that the learned coefficient agrees with the XEB fidelity). To estimate the overlap between the double and single readout errors, we simply consider their inner product, which can be calculated to be (in our setting)
\begin{equation}
    \langle \pi^\text{readout}_{1\rightarrow0,i} + (\pi_1-1/d)/2, \pi^{\text{doub.~readout}}_{(1\rightarrow 0)^2, jk} - (\pi_1-1/d)/4\rangle = \frac{3}{14} \text{ if $i=j$ or $i=k$},
\end{equation}
As a reminder, we assume that the double readout errors cannot happen on the same qubit and therefore $j\neq k$. Orthogonalization the double readout term against the single readout term, and re-parametrizing the model gives the desired correction Eq.~\eqref{eq:double_readout_corr}, which we used in Fig.~\ref{fig:google_data}(e).

\subsubsection{Proportion of error sources}
Combining these results allows us to determine the proportions of each error source to the overall measurement, as plotted in Fig.~\ref{fig:google_data}(a). 

Therefore, we assign their proportions as:
\begin{itemize}
    \item Fidelity: estimated as in Eq.~\eqref{eq:fid_correction}.
    \item State preparation and 1q dephasing errors: no change to $\hat{c}_i$.
    \item 2q dephasing and flip-flop errors: $(3/4) \hat{c}_i$
    \item Single qubit readout errors: $(1/2)\hat{c}^{\text{readout}}_{1\rightarrow 0, j}-(3/14)\sum_k \hat{c}^{\text{doub.~readout}}_{(1\rightarrow 0)^2, jk}$  [Eq.~\eqref{eq:double_readout_corr}].
    \item Double qubit readout errors:
     contribution given by $(3/7-1/4) \hat{c}^{\text{doub.~readout}}_{(1\rightarrow 0)^2, jk}$.
\end{itemize}

In our problem, we have considered error sources where $\sum_z \pi_i(z) = 1$ or $\sum_z \pi_i(z) = 0$. The sum of the $c_i$'s of the former type will be 1, in order for $p(z)$ to be normalized, while the $c_i$'s of the latter type do not have such a constraint. One can verify that with the above prescription, the contributions over all error sources will sum to 1, as desired.

\subsubsection{Physical error rates}
Finally, we can construct estimators for the physical error rates  $\Gamma_i$ from the fitted coefficients $c_i$ as well as the fidelity $\hat F$. In our case, where errors correspond to the application of \textit{only one} non-trivial Kraus operator, and under the assumption that the errors are independent, these are related by:
\begin{equation}
    \hat{\Gamma}_i = \frac{\hat{c}_i}{\hat{F} + \hat{c_i}}.
    \label{eq:c_to_gamma}
\end{equation}

This relation arises as the coefficient $c_i$ describes the probability of a specific, single event, which is the \textit{product} of the physical error rates:
\begin{equation}
    c_i \approx \Gamma_i \prod_{j\neq i} (1-\Gamma_j) \approx \frac{\Gamma_i}{1-\Gamma_i} F~,
\end{equation}
where the second equality is because the many-body fidelity is given by $F \approx  \prod_{j} (1-\Gamma_j)$. Eq.~\eqref{eq:c_to_gamma} is necessary to extract the physical error rates, as plotted in Fig.~\ref{fig:google_data}(c,d,e,f). In particular, this rescaling by $\hat{F}$ is necessary for proper comparison between single- and double- readout error rates to detect correlated readout errors in Fig.~\ref{fig:google_data}(f).

\subsection{Goodness-of-Fit}
In this section, we 
conduct a goodness-of-fit analysis for  model~\eqref{eq:poisson_model_google}  
which we adopted in our real data analysis. 
Concretely, our goal is to
test the null hypothesis 
$$H_0: Y \sim \widebar \bQ_c = \otimes_{j=1}^d \mathrm{Poi}(n\Pi_{\cdot j}^\top c), 
\quad \text{for some } c 
\in \bbR_+^k\text{ such that } \Pi^\top_{\cdot j} c \geq 0 \text{ for all } j.$$
We construct a heuristic
test for this composite null hypothesis, using
the following 
$\chi^2$
statistic~\citep{wasserman2013}:
\begin{align}
    \chi^2 = \sum_{j} \frac{\big(Y_j - n\Pi^\top \hat c^\mathrm{MLE})^2} {n\Pi^\top  \hat c^\mathrm{MLE}},
\end{align}
where $\hat c^{\mathrm{MLE}}$ is
defined in equation~\eqref{eq:mleapp}.
We note that, under the null hypothesis, the typical magnitude of $\chi^2$ is on the order of $d$. 

We calibrate the $\chi^2$ statistic
heuristically, using the parametric bootstrap. 
Concretely, we compare the observed value of  the $\chi^2$ statistic, denoted
$\chi^2_{\text{obs}}$, to the distribution of $\chi^2$ 
that would be expected if $\widebar \bQ_{\hat{c}^{\mathrm{MLE}}}$ were the true data-generating distribution.
We approximate this distribution
by simulating
$1,000$ synthetic datasets of $n=500,000$ bitstring samples from this distribution. For each dataset, we refit the model coefficients $c$ and compute the corresponding $\chi^2$ values. This forms an empirical estimate of the sampling distribution of $\chi^2$
when $Y\sim \widebar \bQ_{\hat c^{\mathrm{MLE}}}$, which
we use to compute the probability of observing a $\chi^2$ value more extreme
than $\chi^2_{\mathrm{obs}}$,
under the hypothesis that the measurements are drawn from $\widebar \bQ_{\hat c^\mathrm{MLE}}$. 

For $N=18$ and one random circuit instance, we obtain $\chi^2_{\text{obs}} = 281,858$. Meanwhile, our simulated $\chi^2$ distribution has mean $\mu = 261,818$ and standard deviation $\sigma = 770$. Our observed $\chi^2_{\text{obs}}$ is thus $26\sigma$ away from the mean, with $p$-value $< 10^{-3}$. 

We compare this to a goodness-of-fit test using Google's two-component error model~\eqref{eq:white_noise_model} as the null hypothesis---an analysis which was also conducted by~\citep{rinott2022statistical}. We find an observed value $\chi^2_{\text{obs}} = 290,342$, while the simulated $\chi^2$ distribution has mean $\mu = 262,134$ and standard deviation $\sigma = 740$. The observed $\chi^2_{\text{obs}}$ value is $38\sigma$ away from the mean in this case, also with negligible $p$-value $< 10^{-3}$. 

Although our analysis suggests considerable room for improvement in modeling the data of Ref.~\cite{arute2019quantum}, our $k$-component model has a marked improvement of $12\sigma$ over the simplest two-component white-noise model. The overall goodness-of-fit is still poor, indicating that the dataset contains information about a host of error processes not currently in our model. Nevertheless, our fitting procedure turns out to be remarkably robust, yielding fitted error rates for quantities of interest that are 
comparable with all available estimates from alternative benchmarking methods (see Fig.~\ref{fig:google_data}). 
As in the spirit of cross-entropy benchmarking, we expect this robustness to be due to the fact that the remaining error sources not captured by our model
lie in spaces that are roughly orthogonal to the row space of $\Pi$
(up to centering).  

\section{Justifying the Independent Porter-Thomas Assumption} 
In this section, we provide some quantitative evidence for the validity of assumption~\ref{assm:pt}, which we leveraged throughout our theoretical study. This assumption requires the various bitstring error distributions $\Pi_{i\cdot}$\,to be mutually independent, and distributed according to the Porter-Thomas law. In what follows, we will show that this assumption holds to second order: under mild assumptions on the errors of the circuit, we find that the rows of $\Pi_{i\cdot}$ are approximately uncorrelated, and have marginal moments which are consistent with the Porter-Thomas law.

Specifically, we shall explicitly calculate 
$\bbE[\Pi \Pi^T]$, where the rows $\Pi_{i\cdot}$  are probability distributions arising from the presence of a single error in a local brickwork random unitary circuit. For simplicity of analysis, we will assume here that our errors are single-qubit Pauli terms. Furthermore, we also consider errors that are within the bulk of circuit, so that we can replace the circuit before and after the signal with global Haar random unitaries $R_1,R_2$. More specifically, for diagonal terms, we take
\begin{align}\label{eq:haarPi}
    \pi_0(z) &= |\langle z|R_2 R_1|0\rangle|^2 \\
    \pi_i(z) &= |\langle z|R_2 P_i R_1|0\rangle|^2,
\end{align}
where $P_i$ is a Pauli error and $R_1, R_2 \sim \text{Haar}(2^N)$ with $N$ the system size. 
We treat the vectors 
$\pi_0(z)$ and $\pi_i(z)$ as forming the rows of $\Pi$.

Now, using Weingarten calculus~\cite{Collins2006, Collins2022}, we can compute
\begin{align}
    \bbE_{R_1, R_2} \left[\sum_z \pi_0(z)^2 \right]= \frac{2}{d} + O\left(\frac{1}{d^2}\right), \\
\bbE_{R_1, R_2} \left[\sum_z \pi_i(z)^2\right] = \frac{2}{d} + O\left(\frac{1}{d^2}\right).
\end{align}
The leading order terms exactly match results for Porter-Thomas distributions, where $d = 2^N$. Thus, differences from Porter-Thomas are exponentially small.

For off-diagonal terms involving $\pi_0(z)$ and an error distribution $\pi_i(z)$, we can similarly use Eq.~\ref{eq:haarPi}. This yields
\begin{align}
\bbE_{R_1, R_2}\left[ \sum_z \pi_i(z) \pi_0(z)\right] = \frac{1}{d} + O\left(\frac{1}{d^3}\right).
\end{align}
The leading order term is again exact if assuming i.i.d.~Porter-Thomas distributions.

Finally, for off-diagonal terms involving two different signals, we also consider the portion of the underlying brickwork circuit $\calR$ that lies between the spacetime locations of these two signals. To be specific, suppose the signal $P_i$ and $P_j$ are $t$ layers apart in the original circuit. Then, $\calR$ would comprise of exactly these $t$ layers of local random unitaries that separate the two signals.
This leads to modified expressions of the form
\begin{align}
    \pi_i(z) &= |\langle z|R_2 \calR P_i R_1|0\rangle|^2, \\
    \pi_j(z) &= |\langle z|R_2 P_j \calR R_1|0\rangle|^2
\end{align}

for computing 
\begin{align}
    \bbE_{R_1, R_2, \calR} \left[\sum_z \pi_i(z) \pi_j(z)\right].
\end{align}

Now,
\begin{align}
    \bbE_{R_1, R_2}\left[ \sum_z \pi_i(z) \pi_j(z) \right]= \frac{1 + d^{-2}\text{Tr}( P_i(t) P^*_{j}) \text{Tr}( P^*_i(t) P_{j})}{d} + O\bigg(\frac{1}{d^2}\bigg),
\end{align}
where we have defined
\begin{align}
    P_i(t) \equiv \calR^\dagger P_i \calR.
\end{align}

The leading order term above differs only from the i.i.d. Porter-Thomas case by $d^{-2}\text{Tr}( P_i(t) P^*_{j}) \text{Tr}( P^*_i(t) P_{j})$, where we have included $d^{-2}$ to normalize the trace, which is $O(d^2)$. 

To proceed, we can compute
\begin{align}
    \bbE_{\calR} \bigg(d^{-2}\text{Tr}( P_i(t) P^*_{j})\text{Tr}( P^*_i(t) P_{j})\bigg),
\end{align}
where the average is over individual two-qubit Haar random unitaries in the brickwork $\calR$. This can be done through the standard technique of mapping the unitaries to an Ising spin model, as explained in Refs. \cite{nahum2018operator, Bao2020}. This yields
\begin{align}
    \bbE_{\calR} \bigg(d^{-2}\text{Tr}( P_i(t) P^*_{j})\text{Tr}( P^*_i(t) P_{j})\bigg) = D(x, t),
\end{align}
where $x$ and $t$ indicate how far $P_j$ is from $P_i$ in the space and time directions. Specifically, $t$ is is equal to the depth of $\calR$, and $x$ represents how many qubits away the error $P_j$ is from $P_i$. For the full expression of $D(x, t)$, see Ref. \cite{gong2025robust}---here, we only discuss relevant properties of the function. Specifically, $D(x, t)$ decays exponentially in both $x$ and $t$. Thus, for signals $P_i$, $P_j$ that are relatively spaced out in the random circuit, $D(x, t) \ll 1$, and the expectation $ E_{R_1, R_2, \calR} \sum_z \pi_i(z) \pi_j(z) $ also approximately satisfies the i.i.d. Porter-Thomas result of $\frac{1}{d}$.

In all cases, we see that the second moments of the rows of $\Pi_{i, \cdot}$ match those of i.i.d. Porter-Thomas distributions up to exponentially small corrections in the system size $N$ and spacing of signals.

\section{Further Technical Background}
In this Appendix, we summarize 
several known technical results and definitions which are used throughout our proofs. 

\subsection{Technical Results}
We begin by stating a few standard
facts about Poisson distributions. 
The following is a standard upper bound on 
the Kullback-Leibler
divergence between Poisson random variables.
\begin{lemma}
\label{lem:kl_poisson}
For any $\mu,\nu > 0$, it holds that
$$\KL\big( \mathrm{Poi}(\mu)\,\|\,\mathrm{Poi}(\nu)\big) = \mu \log \frac{\mu}{\nu} + \nu - \mu.$$
Furthermore, for all $C > 0$, there exists $K > 0$ such that
if $|\nu-\mu| < C \nu$, then
$$\KL\big( \mathrm{Poi}(\mu)\,\|\,\mathrm{Poi}(\nu)\big) \leq K\cdot \frac{(\mu-\nu)^2}{\nu}.$$
\end{lemma}
The following is Lemma~2 of~\cite{acharya2016}.
\begin{lemma}
\label{lem:acharya}
Let $X\sim \mathrm{Poi}(\lambda)$. For any integer $r \geq 1$, let
$$\hat T = X! / (X-r)!.$$
Then, 
$$\bbE[\hat T] = \lambda^r, \quad 
\Var[\hat T] \leq \lambda^r\big((\lambda+r)^r - \lambda^r\big).$$
\end{lemma}

Next, we state a technical result
about perturbations of polynomials, 
which is adapted
from~\cite{hundrieser2025}.
\begin{lemma}
\label{lem:real_roots}
Let $f$ be a polynomial of degree $k$ with $k$ real roots $x_1,\dots,x_k \in \bbR$.
Assume these roots are pairwise distinct, and let $\delta = \min_{i\neq j} |x_i-x_j|$. 
Then, for any $\epsilon < (\delta/2)^k$, the polynomial $f+\epsilon$ also has $k$ real roots
$x_1^\epsilon,\dots,x_k^\epsilon \in \bbR$ satisfying
$$W_1\left(\frac 1 k \sum_{i=1}^k \delta_{x_i},\frac 1 k \sum_{i=1}^k \delta_{x_i^\epsilon}\right) \leq \delta/2.$$
\end{lemma}


The following Lemma collects several elementary facts about Dirichlet and exponential
distributions which are used throughout our proofs.
\begin{lemma}\label{lem:dirichlet}
Let $d \geq 2$, and 
let $\pi=(\pi_1,\dots,\pi_d)^\top \sim \calD_d$ be a flat Dirichlet-distributed random vector, 
and let $\varpi=(\varpi_1,\dots,\varpi_d)^\top$ be a random vector consisting
of i.i.d.  Exp$(d)$-distributed random variables. 
Then, the following assertions hold. 
\begin{enumerate}
\item $\pi_{i} \sim \mathrm{Beta}(1,d-1)$ for $i=1,\dots,d$.
\item $\sum_{i=1}^d \varpi_i \sim \mathrm{Gamma}(d,d)$. 
\item $\pi \overset{d}{=} (\varpi_1,\dots,\varpi_d) / \sum_{i=1}^d \varpi_i$. 
\item If $G\sim \mathrm{Gamma}(d,d)$  is independent of $\pi$, then the vector $(G\pi_1,\dots,G\pi_d)^\top$ 
consists of i.i.d. Exp$(d)$ random variables. 
\item If $G \sim \mathrm{Gamma}(\alpha,\lambda)$ with $\alpha\in \bbN$ and $\lambda > 0$, 
and $Y|G\sim \mathrm{Poi}(G)$, then the marginal
law of $Y$ is Negative Binomial with number of trials $\alpha$ and probability parameter $\lambda / (\lambda+1)$. 

\item We have for all $\ell=1,2,\dots$
$$\bbE[\pi_1] = \bbE[\varpi_1] = \frac 1 d,\quad \Var[\pi_1] = \frac 1 {d^2}, \quad \bbE[\pi_1^\ell] = \frac{\ell!}{d^\ell}.$$
\end{enumerate}
\end{lemma}

\subsection{Classical Polynomial Families}
We now  recall the definitions and basic properties
of three polynomial families which play an important role in our development.

\subsubsection{Elementary Symmetric Polynomials}
\label{app:elementary_symmetric_polynomials}
The  {\it elementary symmetric polynomials} 
are a family
of $k$-dimensional polynomials, defined for 
all $c_1,\dots,c_k \in \bbC$ by
\begin{align}\label{eq:elementary_symmetric_polynomials}
	e_0(c_1,\dots,c_k) := 1 \quad \text{ and } 
	\quad  e_j(c_1, \dots, c_k) &:= \sum_{1\leq i_1< i_2 < \dots <  i_j \leq k} \;\; \prod_{\ell =1}^j c_{i_\ell}.
\end{align} 
By   {\it Vieta's formula}, the elementary symmetric polynomials can be used to describe
the coefficients of a univariate monic polynomial $f(z) = \prod_{i=1}^k (z-c_i)$, with roots $c_1,\dots,c_k\in \bbC$.
Concretely, one has:
\begin{align}
\label{eq:vieta_formula}
f(z) = z^k + \sum_{j=1}^k (-1)^j e_j(c_1, \dots, c_k) z^{k-j},\quad z \in \bbC.
  \end{align} 
Since the polynomials $e_j$ are symmetric, they admit an algebraic representation in terms of the 
moments $m_1(c), \dots, m_k(c)$ of the vector $c$, namely
$$m_j(c) = \frac 1 k \sum_{i=1}^k c_i^j,\quad j=1,\dots,k.$$
This representation can be made explicit using
 {\it Newton's identities} 
which state that for any $\ell =1,\dots,k$, 
 \begin{align}\label{eq:newton_identity}
 	e_\ell(c_1, \dots, c_k) 
 	&= \frac{k}{\ell} \sum_{j=1}^{\ell} (-1)^{j-1} e_{\ell-j}(c_1, \dots, c_k)  m_j(c),
 \end{align}
 In particular, equations~\eqref{eq:vieta_formula}--\eqref{eq:newton_identity} together imply that the vector $c$ is uniquely
 determined, up to permutation of its entries, by the vector of moments $m(c) = (m_j(c): 1 \leq j \leq k)$. 
That is, one has the following simple fact.
\begin{lemma}
\label{lem:identifiability_from_moments}
For any $c,c'\in \bbC^k$, it holds that
$$m(c) = m(c') ~~\Longrightarrow ~~ \{c_1,\dots,c_k\} = \{c'_1,\dots,c'_k\}.$$
\end{lemma}
Some of our results will rely on a {\it quantitative} analogue of Lemma~\ref{lem:identifiability_from_moments}. Concretely, 
when constructing statistical estimators $\hat c$ of $c$ via moment estimation, we will be led to 
the question of quantifying the distance between $\hat c$ and $c$ in terms of their moment distance.  
It turns out that such quantitative bounds can be obtained by combining Newton's identities with existing perturbation bounds
for polynomial roots, which were first developed by Refs.~\citep{ostrowski1940,ostrowski1970}. 
This strategy was recently used by~\citet{hundrieser2025}, who proved the following 
result.
\begin{lemma}[\cite{hundrieser2025}]
\label{lem:stability_moments}
There exists a constant $C =C(k) > 0$ such that for any $c,c' \in \bbC^k$, 
$$W(c,c') \leq C \|m(c)-m(c')\|^{\frac 1 k}.$$
\end{lemma}
This Lemma shows that the 
sorted loss function $W$
is $(1/k)$-H\"older continuous with respect to the $\ell_1$ distance between moment vectors. 

\citet{hundrieser2025} additionally
showed that the H\"older exponent $1/k$
can be improved if the coordinates of
the elements $c,c'$
admit some separation. In order to state
this refined result, recall the
set $\Delta_{k,k_0}$ defined in Appendix~\ref{app:moment}. Let 
$c^\star\in \Delta_{k,k_0}$ be given, and let
$v_1 > \dots > v_{k_0}$ denote its $k_0$ distinct
entries.
Define the Voronoi cells
$$V_0 = \emptyset, 
\quad V_\ell = \left\{ z \in \bbC^k : 
\|z-v_i\| \leq \|z-v_j\|, \, \forall i\neq j\right\} \setminus V_{\ell-1},
\quad \ell=1,\dots,k_0.$$
Furthermore, let $r_\ell$
denote the multiplicity of $v_\ell$ among
the entries of $c^\star$, for all $\ell=1,\dots,k_0$. 
Given $c \in \Delta_k$, write
 $c_{V_\ell} = \{c_i \in V_\ell: 1 \leq i \leq k\}$.
We then define,
for all $c,c'\in \bbC^k$:
\begin{equation}
\widebar \calD_{c^\star}(c,c') = 1 \wedge 
\sum_{\ell=1}^{k_0} 
W^{r_\ell}(c_{V_\ell},c'_{V_\ell}),
\end{equation}
with the convention that $W(c_{V_\ell},c_{V_{\ell'}}) = \infty$ when $|c_{V_\ell}| \neq |c_{V_{\ell'}}|$. 
Finally, let 
$$\delta(c^\star) = \min_{1 \leq \ell < \ell'\leq k} 
|v_\ell-v_{\ell'}|.$$
We then have the following refined stability bound.
\begin{lemma}[\cite{hundrieser2025}]
\label{lem:refined_stability_moments}
Let $1 \leq k_0 \leq k$
and $c^\star\in \Delta_{k,k_0}$. Then, 
there exists a constant  $C=C(k,k_0,\delta(c^*)) > 0$ such that for any $c,c'\in \bbC^k$, we have 
$$\widebar\calD_{c^\star}(c,c') \leq 
W^{k-k_0-1}(c,c') \leq C \big\|m(c)-m(c')\big\|.$$
\end{lemma}


\subsubsection{Charlier Polynomials}
\label{app:charlier_polynomials}
Let $f(x;\lambda) = e^{-\lambda}\lambda^x / x!$ denote the $\mathrm{Poi}(\lambda)$
density, evaluated at $x=0,1,\dots$. The  family of Charlier polynomials 
$$\varphi_\ell(x;\lambda) := \sum_{r=0}^\ell (-1)^{\ell-r} {\ell \choose r} \frac{(x)_r}{\lambda^r},\quad x,\ell=0,1,\dots,$$
indexed
by a parameter $\lambda > 0$, 
are a classical family of polynomials on $\bbR$ which are orthogonal with respect to the $L^2(\mathrm{Poi}(\lambda))$ norm~\citep{szego1939}.
One has the relation
\begin{equation}
\label{eq:charlier_ortho}
\sum_{x=0}^\infty \varphi_\ell(x;\lambda)\varphi_{\ell'}(x;\lambda) = {\ell!}{\lambda^\ell}I(\ell=\ell'),\quad \ell,\ell'=0,1,\dots
\end{equation}
The exponential generating function associated to the Charlier polynomials is
\begin{equation}
\label{eq:charlier_generating_fn}
G(x,t) = \sum_{\ell=0}^\infty \varphi_\ell(x;\lambda) \frac{t^\ell}{\ell!} = e^{-t} \left(1 + \frac{t}{\lambda}\right)^x,
\quad \text{for all } t \in \bbR. 
\end{equation}

\subsubsection{Bell Polynomials}
\label{app:bell_polynomials}
  Given an integer $p \geq 1$, the family of incomplete Bell polynomials $\{B_{\ell,p}\}_{\ell=1}^p$ consists
of the set of polynomials on $\bbR^{p-\ell+1}$ defined by
\begin{align}
\label{eq:bell_poly}
 B_{p,\ell}\big(\xi_1,\dots,\xi_{p-\ell+1}\big) =
p! \sum_{(h_1,\dots,h_{p-\ell+1}) \in \calH_{p,\ell}}
\prod_{i=1}^{p-\ell+1} \frac{\xi_i^{h_i}}{(i!)^{h_i} h_i!},
\end{align}
for all $\xi_1,\dots,\xi_{p-\ell+1} \in \bbR$. 
Here, $\calH_{p,\ell}$ consists of all tuples
$(h_1,\dots,h_{p-\ell+1})$ of nonnegative integers such that
$$ \sum_{i=1}^{p-\ell+1} h_i = \ell,\quad \sum_{i=1}^{p-\ell+1} ih_i = p.$$
Furthermore, the $p$-th complete Bell polynomial is defined by
$$B_p(\xi_1,\dots,\xi_p) = \sum_{\ell=1}^p B_{p,\ell}(\xi_1,\dots,\xi_{p-\ell+1}) = p! \sum_{r_1 + 2r_2 + \dots + pr_p=p} \prod_{i=1}^p \frac{\xi_i^{jr_i}}{(i!)^{r_i} r_i!}.$$
One has the identity $|\calH_{p,\ell}| 
= S(p,\ell)$, where 
\begin{align}
S(p,\ell) = \sum_{i=1}^\ell \frac{(-1)^{\ell-i} i^p}{(\ell-i)! \ell!} \leq \frac{\ell^p}{\ell!}.
\end{align} 
Furthermore, one has the basic identities
\begin{align}
B_{p,\ell}(a,\dots,a) &= a^\ell S(p,\ell) \\
B_{p,\ell}(a,a^2,\dots,a^{p-\ell+1}) &= a^p S(p,\ell)
\end{align}

Given a random variable $X$ with cumulants $\xi_p = \kappa_p(X)$ and moments $\eta_p = \bbE[X^p]$,
for $p=1,2,\dots$, one has the relations
\begin{equation}
\label{eq:bell_moments_cumulants}
\begin{aligned}
\xi_p &= \sum_{\ell=1}^p (-1)^{\ell-1} (\ell-1)! B_{p,\ell}(\eta_1,\dots,\eta_{p-\ell+1}) \\
\eta_p &= \sum_{\ell=1}^p B_{p,\ell}(\kappa_1,\dots,\kappa_{p-\ell+1}).
\end{aligned}
\end{equation}

\end{document}